%% file: main-stoch-knapsack-adaptivity-gap.tex
\title{Stochastic Knapsack:\\Semi-Adaptivity Gaps and Improved Approximation} 
\author{Zohar Barak, Inbal Talgam-Cohen}
\date{2025}

\documentclass[11pt,letterpaper]{article}
\usepackage[left=1in, right=1in, top=1in, bottom=1in]{geometry}
\usepackage{amsmath}
\usepackage{booktabs} 
\usepackage[utf8]{inputenc}
\usepackage{graphicx} 
\usepackage[]{local}
\usepackage{makecell}
\usepackage{algorithm}
\usepackage{algpseudocode}
\usepackage{caption}
\usepackage{float}

\usepackage{tabularx,array}
\newcolumntype{Y}{>{\raggedright\arraybackslash}X}
\usepackage{multirow}

\author{
Zohar Barak\thanks{Tel Aviv University, \url{zoharbarak@mail.tau.ac.il}} 
\and
Inbal Talgam-Cohen\thanks{Tel Aviv University, \url{inbaltalgam@gmail.com}} 
}
\date{}

\begin{document}

\pagenumbering{gobble}

\maketitle

\input{abstract}
\clearpage
\tableofcontents
\newpage

\pagenumbering{arabic}
\setcounter{page}{1}
\newpage

\input{intro}

\input{preliminaries}
\input{technical-overview}
\input{one-to-n-gap-upper-bound}

\input{k-to-n-gap}
\input{zero-to-one-gap}
\input{Bernoulli_zero_one}
\input{conclusions}
\input{acknowledgements}

\newpage
\appendix
\input{appendix}

\printbibliography

\end{document}

%% file: abstract.tex
\begin{abstract}

In stochastic combinatorial optimization, algorithms differ in their \emph{adaptivity}: whether or not they query realized randomness and adapt to it. Dean et al. (FOCS '04) formalize the \emph{adaptivity gap}, which compares the performance of fully adaptive policies to that of non-adaptive ones.

We revisit the fundamental \emph{Stochastic Knapsack} problem of Dean et al., where items have deterministic values and independent stochastic sizes. A \emph{policy} packs items sequentially, 
stopping at the first knapsack overflow or before. 

We focus on the challenging \emph{risky} variant, in which an overflow forfeits all accumulated value, and study the problem through the lens of \emph{semi-adaptivity}: We measure the power of $k$ adaptive queries for constant $k$ through the notions of $0$-$k$ semi-adaptivity gap (the gap between $k$-semi-adaptive and non-adaptive policies), and $k$-$n$ semi-adaptivity gap (between fully adaptive and $k$-semi-adaptive policies).

Our first contribution is to improve the classic results of Dean et al.~by giving tighter upper and lower bounds on the adaptivity gap. Our second contribution is a smoother interpolation between non-adaptive and fully-adaptive policies, with the rationale that when full adaptivity is unrealistic (due to its complexity or query cost), limited adaptivity may be a desirable middle ground. We quantify the $1$-$n$ and $k$-$n$ semi-adaptivity gaps, showing how well $k$ queries approximate the fully-adaptive policy. We complement these bounds by quantifying the $0$-$1$ semi-adaptivity gap, i.e., the improvement from investing in a single query over no adaptivity.

As part of our analysis, we develop a 3-step ``Simplify-Equalize-Optimize'' approach to analyzing adaptive decision trees, with possible applications to the study of semi-adaptivity in additional stochastic combinatorial optimization problems.

\end{abstract}

%% file: intro.tex
\section{Introduction}
\label{sec:intro}

\paragraph{Adaptivity.}
In optimization, the \emph{adaptivity gap} quantifies the benefit of using an adaptive decision-making policy compared to a non-adaptive one. Adaptivity is a fundamental concept, which refers to observing the instantiation of some elements of the problem before deciding on the next steps of the policy. 
We refer to this as \emph{adaptive querying}. 

Adaptivity 
has been extensively studied across a wide spectrum of fields in computer science. 
In the context of stochastic optimization, the notion of \emph{adaptivity gap} was formalized in the seminal work of Dean, Goemans, and Vondrák~\cite{2004DeanGoemansVon,dean2008approximating}. 
They introduce it for the \emph{Stochastic Knapsack} problem, as the worst-case ratio between the expected values of an optimal adaptive policy and an optimal non-adaptive one. 
As they explain, the adaptivity gap plays a similar role to the 
integrality gap, by ``telling us the best approximation bound we can hope to achieve by considering a particular simple class of solutions.'' Here, a ``simple class of solutions" refers to solutions with no adaptive queries. 
Since the work of Dean et al., the measure of adaptivity gap has been applied to a variety of stochastic optimization problems (see Section~\ref{sec:related-work}). 

\paragraph{Semi-adaptivity.}

A natural generalization of the adaptivity gap is a measure of the algorithmic power of a \emph{limited} number of adaptive queries. The number $k\ge 1$ of adaptive queries is a \emph{smoother} notion of simplicity, and keeping this number small often reaps most of the benefits of no adaptivity, while allowing the policy more flexibility.
While a fully adaptive policy might require an exponentially-large decision tree that is intractable to store in memory, a partially adaptive policy may require only a polynomial amount of memory. 
Limited adaptivity also enables parallelization of the computation, which is important in practical applications~\citep{balkanski2018adaptive}. 

From an economic point of view, there are many applications in which each query for information has a \emph{cost}. Performing medical tests, labeling samples, or deploying sensors are examples of such costly information queries~\citep{DBLP:conf/icml/GhugeGN21}. A recent example is prompting an LLM: Each prompt (and tool-calls it triggers) has nontrivial cost and latency. Fully adaptive interactions with an AI agent that involve repeated follow-up questions, output verification, and query refinement can rapidly drive up costs, even though the incremental benefits are often small. 

These considerations motivate the limitation, but not necessarily elimination, of adaptive queries, and limited adaptivity has begun to gain traction in recent research (see Section~\ref{sec:related-work}). 

\paragraph{Semi-adaptivity for Stochastic Knapsack.}

Perhaps surprisingly, the problem of Stochastic Knapsack for which the adaptivity gap was first introduced has not yet been revisited from the perspective of limited adaptivity. In this work we apply semi-adaptivity to this classic problem. En route, we also improve the best known bounds on the adaptivity gap for this problem.

In the Stochastic Knapsack (SK) problem, there are $n$ items. Each item has a known reward and a random size. The size is drawn independently from the item's known size distribution. The goal is to select a subset of items subject to a knapsack capacity constraint, such that the total expected reward is maximized. There are two natural variants:
In $\nrsk$, if an item overflows the remaining capacity, only its own value is lost.
In $\rsk$, an overflow causes the entire accumulated reward to be forfeited.
This version is also known as Fixed-Set~\cite{dean2008approximating}, All-Or-None~\cite{bhalgat2011improved}, or Blackjack~\cite{levin2014adaptivity,fu2018ptas}, and is the main focus of our investigation. 

In the context of SK, a \emph{non-adaptive} policy commits to a fixed strategy of placing items into the knapsack, including item ordering and after which item to stop, regardless of the items' realized sizes. In contrast,
an \emph{adaptive} policy chooses items sequentially based on previously-realized sizes. 
An \emph{adaptive query} is an observation of the realized state of the knapsack (i.e., the total size of items inserted so far). The possible number of adaptive queries is thus between $0$ (non-adaptive) and $n$ (fully-adaptive). A \emph{$k$-semi-adaptive} policy chooses items sequentially based on up to $k$ adaptive queries. Given the motivations for semi-adaptivity, we are most interested in the constant $k$ regime.

\paragraph{The stopping challenge of Risky Stochastic Knapsack.}

Imagine an autonomous robot or drone that collects high-value samples under a finite battery. Each sample yields a known reward but drains a random amount of energy; if the battery is exhausted mid-mission, all collected samples are lost.
 
This is a representative $\rsk$ setting.

The following example and its analysis in \cref{clm:ber-eps-items-behavior} give some intuition for the difference between $\nrsk$ and $\rsk$ in terms of their adaptivity gap:%
\footnote{We revisit a generalization of this example ($\epsnoisy$ items) in \cref{sec:1-n-semi-lb-and-eps-noisy}.}

\begin{example}[SK with small Bernoulli items]
\label{example:ber-eps-items}
For $0<\eps \ll 1$, consider an SK instance with $n \gg \frac{1}{\eps}$ identical
items, each item $i$ with value $v_i = \eps$ and size $S_i \sim \Ber(\eps)$ (i.e., $S_i = 1$ with probability $\eps$, and $0$ otherwise). The expected item size is $\eps$. Normalize the knapsack capacity to $1$.
\end{example}

\begin{prop}[Optimal adaptive policies]
\label{clm:ber-eps-items-behavior}
For the instance in \cref{example:ber-eps-items}, the expected value of the optimal adaptive policy and its overflow probability is:
\begin{enumerate}[nosep, noitemsep]
    \item $\nrsk$: Expected value $2 - O(\eps)$, overflows with probability $1$.
    \item $\rsk$: Expected value $1 + \frac{1}{e} - O(\eps)$, overflows with probability $\frac{1}{e}$.
\end{enumerate}
\end{prop}

See \cref{sec:proof-of-clm-ber-eps-items-behavior} for a proof.
The analysis of this example reveals a key challenge of $\rsk$: In $\nrsk$, it is optimal to keep inserting until overflow occurs. In contrast, $\rsk$ policies must balance the expected gain of adding an item, with the risk of overflowing and losing everything. 
%
Interestingly, even though overflowing forfeits all accumulated reward, the optimal policy still allows a non-negligible overflow probability ($\approx 0.37$).
In other words, $\rsk$ requires a strategy that incorporates \emph{optimal stopping} into the decision-making, due to the all-or-nothing penalty on overflow.

\paragraph{Notions of semi-adaptivity gaps and their connection.}

In this work, 
we are interested in quantifying the tradeoff between the number of adaptive queries and the policy's performance. 
Correspondingly, we study 
two notions of semi-adaptivity \emph{gaps}, which together quantify the power of semi-adaptive policies. 

\begin{itemize} 
    \item $\ztokgap$ (\cref{def:ztok-gap}): The ratio between the expected value of the optimal $k$-semi-adaptive policy and that of the optimal non-adaptive policy. This gap captures the benefit of augmenting a non-adaptive policy with a small number of adaptive queries.

    \item $\ktongap$ (\cref{def:kton-gap}): The ratio between the expected value of the optimal fully adaptive policy and that of the optimal $k$-semi-adaptive policy. This gap measures how well a $k$-semi-adaptive policy can approximate full adaptivity.
\end{itemize}

\begin{figure}[ht]
    \centering
    \includegraphics[scale=0.18]{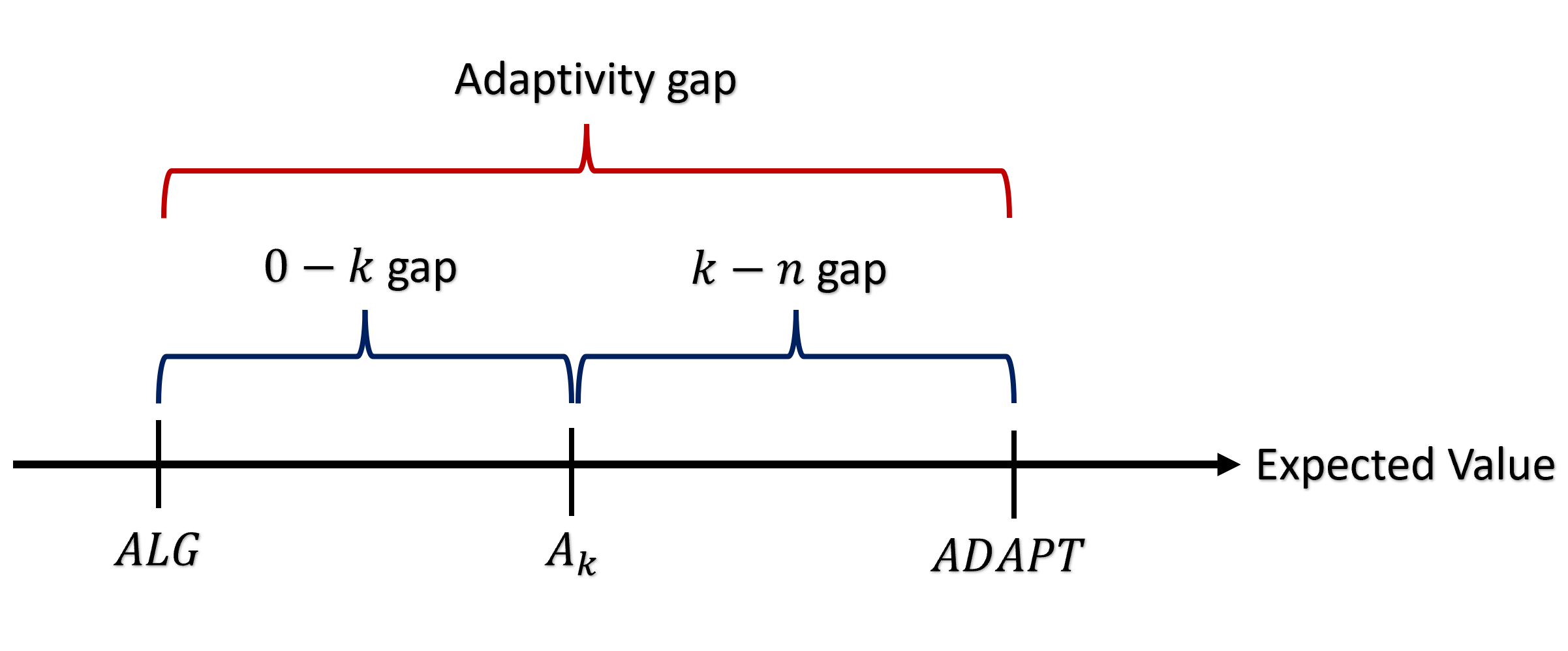}
    \caption{\footnotesize Illustration of the semi-adaptivity gaps. $\ALG$, $A_k$, and $\ADAPT$ denote the expected values of the optimal non-adaptive, $k$-semi-adaptive, and fully-adaptive policies, respectively.}
    \label{fig:semi-adaptivity-gaps}
\end{figure}

Notice that the standard adaptivity gap coincides with the notion of $0$-$n$ semi-adaptivity gap.
The gaps are illustrated in \cref{fig:semi-adaptivity-gaps}. 

Starting with the second notion, variants of the $\ktongap$ were previously studied under names such as \emph{adaptive complexity} (see \cref{sec:related-work}). 
However, in all previous work on limited adaptivity we are aware of, $k$ depends on $n$, e.g., $k=\Theta(\log n)$. Our work differs in focusing on constant $k$. 
As for the first notion, $\ztokgap$ was not previously studied to our knowledge, despite being of natural interest: In many cases, full adaptivity might be too costly, for example when it prevents parallelization, or when each adaptive query incurs a significant overhead or cost. Thus, $k$-semi-adaptive policies may provide a more realistic benchmark against which to evaluate non-adaptive policies.

Our formalization of the $\ktongap$ alongside the $\ztokgap$ reveals that the two are linked in a clean way. In \cref{lem:gaps-lemma}, we prove that their product upper-bounds the standard adaptivity gap. 
Interestingly, the connection we show between the gaps in \cref{lem:gaps-lemma} has the following consequence: an upper bound on the $\ztokgap$ together with a lower bound on the standard adaptivity gap yield a lower bound on the $k$-$n$ semi-adaptivity gap. This is useful because the $\ztokgap$ is often easier to analyze directly than the $\ktongap$ (see, e.g., our result in \cref{cor:zero-to-one-gap-of-rsk} for $k=1$).

\input{results}


\input{related-work}

%% file: results.tex
\subsection{Our Results}
\label{sec:results-techniques}

Our main results for $\rsk$ are summarized in \cref{tab:results-glance} and described below. We also extend some of the results to $\nrsk$ --- see \cref{tab:results-glance-nrsk}. For an overview of our techniques see \cref{sec:tech-overview}. 
\begin{table}[h]
\centering
\small
\renewcommand{\arraystretch}{1.15}
\begin{tabularx}{\linewidth}{@{} l | c | c | >{\centering\arraybackslash}X | >{\centering\arraybackslash}X @{} }
\hline
\textbf{Quantity} & \textbf{Bound Type} & \textbf{Our Bound} & \textbf{Previous Best} & \textbf{Where to Find} \\
\hline
\hline
\multirow{2}{*}{$0$-$n$ full adaptivity gap}
  & Upper
  & $8.47$
  & $9.5$ \citep{dean2008approximating} \textsuperscript{(*)}
  & \cref{thm:non-adaptive-risky-greedy-2-phi-cube}, Sec.~\ref{sec:one-to-n-gap} \\
\cline{2-5}
  & Lower
  & $2$
  & $1.5$ \citep{levin2014adaptivity}
  & \cref{thm:risky-lb-two}, Sec.~\ref{sec:1-n-semi-lb-and-eps-noisy} \\
\hline\hline
\multirow{2}{*}{$0$--$1$ semi\text{-}adaptivity gap}
  & Upper
  & $1.69$
  & ---
  & \multirow{2}{*}{\cref{cor:zero-to-one-gap-of-rsk}, Sec.~\ref{sec:zero-to-one-gap}} \\
\cline{2-4}
  & Lower
  & $1.69$
  & ---
  & \\
\hline\hline
\multirow{2}{*}{$1$--$n$ semi-adaptivity gap}
  & Upper
  & $8.26$
  & ---
  & \cref{thm:alg-single-adaptive-choice-has-gap-at-most-8}, Sec.~\ref{sec:one-to-n-gap} \\
\cline{2-5}
  & Lower
  & $1.18$
  & ---
  & \cref{thm:single-adaptive-choice-1-to-n-gap-lb}, Sec.~\ref{sec:1-n-semi-lb-and-eps-noisy} \\
\hline\hline
$k$--$n$ semi-adaptivity gap,
  & \multirow{2}{*}{Upper}
  & \multirow{2}{*}{$6.44+\sqrt{\varepsilon}$}
  & \multirow{2}{*}{--- \textsuperscript{(**)}}
  & \multirow{2}{*}{\cref{thm:rsk-k-to-n-gap}, Sec.~\ref{sec:k-to-n-gap}} \\
where $k=\tilde{O}(1/\varepsilon)$
  & & & & \\
\hline
\end{tabularx}
\caption{\footnotesize \textbf{Results at a glance: $\rsk$ gaps.}\\ ``---'' means that no better bound was previously known for the same quantity (semi-adaptivity gap).\\ \textsuperscript{(*)} The improvement is also in running time: $O(n\log n)$ compared to the \cite{dean2008approximating} $\eps$-dependent polytime algorithm.\\ \textsuperscript{(**)} \citet{fu2018ptas} give a fully adaptive policy that is an $8$-approximation to the optimal adaptive policy. We improve upon their approximation ratio even though our policy uses only constant-many adaptive choices.}
\label{tab:results-glance}
\end{table}

\renewcommand\tabularxcolumn[1]{m{#1}} 
\newcolumntype{Y}{>{\centering\arraybackslash}X} 
\newcolumntype{C}[1]{>{\centering\arraybackslash}m{#1}} 
\begin{table}[h]
\centering
\small
\setlength{\tabcolsep}{9pt}        
\setlength{\extrarowheight}{2pt}   
\renewcommand{\arraystretch}{1.25} 

\begin{tabularx}{\linewidth}{@{} Y | C{1.9cm} | C{1.6cm} | Y | Y @{}}
\hline
\textbf{Quantity} & \textbf{Bound Type} & \textbf{Our Bound} & \textbf{Previous Best} & \textbf{Where to Find}\\
\hline\hline

\shortstack{$0$-$n$ full adaptivity gap\\ on $\varepsilon$-Noisy Bernoulli}
& Upper
& $2$
& $4$ \citep{dean2008approximating}
& \cref{thm:nrsk-noisy-bernoulli-ub-2}, Sec.~\ref{sec:nrsk-analysis}\\
\hline

$0$-$n$ full adaptivity gap
& Lower
& $1.37$
& $1.25$ \citep{dean2008approximating}
& \cref{cor:nrsk_lb_of_1_plus_1_over_e}, Sec.~\ref{sec:nrsk-analysis}\\
\hline\hline

\shortstack{$k$-$n$ semi-adaptivity gap, \\ where $k=\tilde{O} (1/\varepsilon)$}
& Upper
& $3+\varepsilon$
& ---
& \cref{thm:nrsk-k-to-n-gap}, Sec.~\ref{sec:k-to-n-gap}\\
\hline
\end{tabularx}

\caption{\footnotesize \textbf{Results at a glance: $\nrsk$ gaps.}\\ ``---'' means that no better bound was previously known for the same quantity (semi-adaptivity gap).}
\label{tab:results-glance-nrsk}
\end{table}


\paragraph{$k$-$n$ semi-adaptivity gap upper bounds (Sections~\ref{sec:one-to-n-gap}-\ref{sec:k-to-n-gap}).}

Our first set of results establishes upper bounds on $k$-$n$ semi-adaptivity gaps. Notably, this is done via polynomial-time algorithms. In Section~\ref{sec:one-to-n-gap} we address $k = 0$ and $k = 1$. 
The best known non-adaptive algorithm by \citep{dean2008approximating} achieves an adaptivity gap of $9.5$. 
In \cref{sec:non-adapt-improved} we show a simple non-adaptive algorithm (\cref{alg:non-adaptive-risky-greedy}) that achieves an approximation of $8.47$. 
In Section~\ref{sub:single-query-improvement} we show how to augment it by a single adaptive choice to get a better approximation guarantee of at most $8.26$. 
This demonstrates that even a single adaptive choice suffices to obtain non-negligible improved performance compared to (currently best known) non-adaptive policies. 

In Section~\ref{sec:k-to-n-gap}, for any small constant $\eps \in (0,1)$,
we show that with $k=\tilde{O}(1/\eps)$ adaptive queries, the $k$-$n$ semi-adaptivity gap is upper-bounded 
by $2e + 1 + \sqrt{\eps} \approx 6.44+ \sqrt{\eps}$. 
The previously best known \emph{fully adaptive} policy for $\rsk$ provides an $8$-approximation~\citep{fu2018ptas}. 
Thus, we achieve an improvement both in the approximation ratio and in the number of adaptive choices.
We also deduce an upper-bound for $\nrsk$. 
These results demonstrate that a constant number of adaptive choices suffice to obtain an additional significant improvement.

In terms of techniques, to establish our upper bounds on the $\ktongap$ in~\cref{sec:k-to-n-gap}, we consider $\eps$ to be the cutoff between large and small items (where an item is \emph{large} if its expected size is at least $\eps$ and otherwise it is \emph{small}), 
and combine the following results: 

\begin{itemize}
    \item For instances consisting only of \emph{small} items, 
    while we cannot use the simple greedy non-adaptive policy of~\cite{dean2008approximating} that obtains an approximation of $2 + O(\eps)$ for $\nrsk$, we design 
    a $k$-semi-adaptive
    policy (\cref{alg:semi-adaptive-greedy}). We show this policy achieves approximations of $8$, $6.75$ and $2e + O(\sqrt{\eps})$ for $k = 0$, $k=1$ and $k= \Theta\prn*{1/\sqrt{\eps}}$, respectively (\cref{cor:rsk-semi-adaptivity-gaps}), and of $2 \prn*{\nicefrac{k+2}{k+1}}^{k+2} + O(k \eps)$ for general $k$ (\cref{thm:alg-block-adaptive-gap-from-Phi-one}). This demonstrates a smooth trade-off between the number of adaptive choices and resulting approximation, as illustrated in \cref{fig:semi-gap-trade-off}. 
    \item For instances consisting only of \emph{large} items, with $k=\tilde{O}(1/\eps)$ adaptive queries, for both $\nrsk$ and $\rsk$ the $\ktongap[k]$ is at most $1 + \eps$ (\cref{thm:semi-adaptive-large-items}).
\end{itemize} 


\paragraph{$0$-$1$ semi-adaptivity gap tight bound (Section~\ref{sec:zero-to-one-gap}).}

We complement the above results by bounding the $\ztokgap$,
where we focus on the regime of a \emph{single} adaptive query ($k = 1$). 
We prove a tight upper bound of $1 + \ln(2) \approx 1.69$ on the $\ztokgap[1]$, 
showing the power of a single adaptive query to achieve a constant gain over non-adaptive policies.
Tightness 
is established by lower-bounding the gap for an interesting family of instances which we denote by $\Hc_2$ (\cref{def:H_k}). In fact, the upper bound is also established via this family (\cref{thm:single-adaptive-choice-gap-risky}). For the same family of instances, we also show that for $\nrsk$ the adaptivity gap is precisely $1 + \frac{1}{e}\approx 1.37$ (\cref{thm:single-adaptive-choice-gap-non-risky}). 
This improves the previously-known lower bound on the adaptivity gap, which was $1.25$~\cite{dean2008approximating}. 

\paragraph{$1$-$n$ semi-adaptivity gap lower bound (Section~\ref{sec:1-n-semi-lb-and-eps-noisy}).}

The power of a single adaptive choice is of course limited: 
We prove a lower bound of $\approx 1.18$ on the $\ktongap[1]$ 
by first strengthening the lower bound on the standard adaptivity gap from $1.5$ (\cite{levin2014adaptivity}) to $2$, 
then applying Lemma~\ref{lem:gaps-lemma} that connects the full and semi-adaptivity gaps. 
The strengthened lower bound is again established through analysis of a simple class of instances, those with $\epsnoisy$ distributions. This natural class generalizes \cref{example:ber-eps-items}, and is a generalization of a class of Bernoulli instances studied by \cite{levin2014adaptivity}. 

\paragraph{Implied improvements on the best known approximation guarantees.} 

We summarize the improvements to previously known approximation bounds, achieved via the techniques we develop for the semi-adaptivity gap analyses, and displayed in Tables~\ref{tab:results-glance} and \ref{tab:results-glance-nrsk}:
For $\rsk$, we first improve the best known adaptivity gap upper and lower bounds.  
We remark that our upper bound applies to the \emph{ordered adaptive} SK model of \cite{dean2008approximating} as well.%
\footnote{In this model, the policy must process the items in some given order. As \cite{dean2008approximating} mention, any non-adaptive policy for $\rsk$ with an adaptivity gap of $G$ has the same adaptivity gap for the ordered adaptive model. Therefore, \cref{thm:non-adaptive-risky-greedy-2-phi-cube} applies to this model as well.}
Our algorithm (\cref{alg:non-adaptive-risky-greedy}) is also better in terms of running time, running in $O(n \log n)$ compared to the $O(\text{poly}(n))$ time of~\cite{dean2008approximating}. Second, our semi-adaptive policy that achieves a $\ktongap[\tilde{O}(1/\eps)]$ of $6.44$ improves upon the best known $8$-approximate fully adaptive policy of~\cite{fu2018ptas}. The improvement is both in the approximation ratio and in the number of adaptive choices. 
Third, for $\nrsk$, we strengthen the lower bound on the previously-known adaptivity gap, and improve the upper bound for the special case of $\epsnoisy$ distributions. 

\paragraph{Organization.}

In \cref{sec:preliminaries} we introduce notation and establish the simple but useful \cref{lem:gaps-lemma}. In \cref{sec:tech-overview-conc-fut-dir} we give an overview of our techniques. In Sections~\ref{sec:one-to-n-gap} to \ref{sec:1-n-semi-lb-and-eps-noisy} we present our results as described above (\cref{sec:results-techniques}).
\cref{sec:conclusions-and-future-directions} concludes.
Additional related work appears in \cref{sec:more-related}.

%% file: related-work.tex
\subsection{Related Work}\label{sec:related-work}

Adaptivity gaps, introduced by \cite{dean2008approximating}, have been studied for many stochastic optimization problems. These include Stochastic Covering~\citep{goemans2006stochastic}, Stochastic Orientation~\citep{bansal2015adaptivity}, Stochastic Probing~\citep{gupta2016algorithms,gupta2017adaptivity,DBLP:conf/approx/Bradac0Z19}, Submodular Cover~\citep{agarwal2019stochastic, DBLP:conf/icml/GhugeGN21}, Influence Maximization~\citep{DBLP:conf/faw/TaoWY23,DBLP:conf/aaai/DAngeloPV21,DBLP:conf/isaac/ChenP19}, Stochastic $k$-TSP~\citep{DBLP:conf/innovations/JiangLL020}, Informative Path Planning~\citep{DBLP:conf/aistats/TanGN24}, Stochastic Score Classification ~\citep{Ghuge2022_NonAdaptStochasticScoreClass} and others.

\paragraph{Limited adaptivity.}

Several recent works in stochastic optimization have begun to explore the power of $k$ adaptive queries compared to full adaptivity: \cite{DBLP:conf/icml/GhugeGN21,agarwal2019stochastic} study Stochastic Submodular Cover, while \cite{DBLP:conf/aistats/TanGN24} study Informative Path Planning from an information-theoretic perspective (putting aside computational considerations). 
\citet{esfandiari2021adaptivity} consider both Stochastic Submodular Maximization and Stochastic Minimum Cost Coverage.
The studied problems and techniques are distinct from ours, and focus on $\ktongap$s; to our knowledge $\ztokgap$s have not been previously studied.

The power of limited adaptivity has also been studied in the context of non-stochastic optimization, in particular, in maximization of submodular and other complement-free set functions: see e.g.~the works of \cite{balkanski2018adaptive,balkanski2018non,DBLP:conf/soda/BalkanskiRS19,chekuri2019parallelizing,kupfer2020adaptive}.

In concurrent work, \citet{barakknapsackcosts} investigate a substantially different notion of limited adaptivity for a variant of Stochastic Knapsack in which items have costs. Rather than limiting the number of adaptive queries, they restrict the policy to consider items according to a predetermined, fixed order.

\paragraph{Stochastic Knapsack.} 

For $\rsk$,~\cite{dean2008approximating} give a non-adaptive policy with an adaptivity gap of $9.5$ (which we improve). \cite{levin2014adaptivity} show that the adaptivity gap is exactly $1.5$ for instances with two items, and provide an alternative upper bound of $11.66$ for general instances, as well as an upper bound of $8$ 
for Bernoulli size distributions. 
\cite{fu2018ptas} develop an $(8 + \varepsilon)$-approximate fully adaptive policy. 
We improve on this result in two respects: we provide a $2e + 1 \approx 6.44$-approximate policy, and it is only $\tilde{O}(1/\varepsilon)$-semi-adaptive. 

For $\nrsk$, it is known that the adaptivity gap is in $[1.25,4]$ ~\citep{dean2008approximating}. 
As for approximating the optimal fully adaptive policy,
\citet{dean2008approximating} give a $(3 + O(\varepsilon))$-approximately optimal fully adaptive policy, 
by combining two components: a $(2 + \varepsilon)$-approximate non-adaptive policy for ``small'' items, and a $(1 + \varepsilon)$-approximate fully adaptive policy for ``large" ones (a standard useful technique that we also use and appears also in \cite{schuurman2001approximation,shachnai2008approximation}).
An interpretation of our \cref{thm:semi-adaptive-large-items} is that the adaptive policy of \cite{dean2008approximating} for large items can be ``cut off'' after $\tilde{O}\prn*{\frac{1}{\eps}}$ item insertions without too much loss (see \cref{sec:remark-large-items}). 
\citet{bhalgat2011improved} give an $(\frac{8}{3} + O(\varepsilon))$-approximate fully adaptive policy, later improved to a $(2 + \varepsilon)$-approximation by~\cite{bhalgat20112+}. \citet{ma2014improvements} provides an alternative $(2 + \varepsilon)$-approximation 
for a setting with correlated rewards and sizes. 

\medskip


%% file: preliminaries.tex
\section{Preliminaries}\label{sec:preliminaries}

\paragraph{Problem Setup.}
A stochastic knapsack instance $I$ of size $n$ consists of $n$ items ($|I| = n$). Each item $i \in [n]$ 
is represented by a pair $(v_i,F_i)$, where $v_i\ge 0$ is a (deterministic) value, and $F_i$ is a distribution
from which a random size $S_i$ is independently drawn.

The goal is to insert items into a knapsack of capacity $1$ (normalized w.l.o.g.), maximizing the total value in expectation (where the expectation is taken over the randomness of the item sizes). 
If the total size of the items inserted exceeds $1$, we say that a knapsack \emph{overflow} has occurred, and no additional item can be inserted. There are two variants of the problem: 
In $\nrsk$, a knapsack overflow results in losing only the value of the overflowing item, while in $\rsk$, an overflow results in losing the entire accumulated value of items inserted into the knapsack, leading to a total value of $0$.

A \emph{policy} is a (possibly adaptive) decision rule for a given instance
that, at each step, maps the information revealed so far 
to the next action: insert a specific remaining item or stop.
An \emph{algorithm} is the procedure that gets an instance and outputs a policy. We define adaptivity in the context of Stochastic Knapsack:

\begin{definition}[Adaptive query]\label{def:adaptive-query}
    An \emph{adaptive query} is any step at which a policy inspects the \textit{current state of the knapsack}, namely, the sum of realized item sizes inserted so far.  
\end{definition}

\begin{definition}[$k$-semi-adaptive policy]\label{def:k-semi-adaptive-algorithms}
    A policy is $k$-semi-adaptive if it makes $\le k$ adaptive queries.
\end{definition}

That is, a $k$-semi-adaptive policy observes the state of the knapsack at most $k$ times, and between these observations, it inserts items non-adaptively.
An \emph{adaptive decision tree} is a representation of a policy, where each internal node represents observing the current state (\itc{i.e.}, remaining knapsack capacity) and choosing the next item based on that observation, with different branches corresponding to different realized outcomes, continuing until a stop or overflow leaf.

\paragraph{Adaptivity and semi-adaptivity gaps.}
For any $k$-semi-adaptive policy $p$, let $\text{val}(p, I)$ be the value obtained by policy $p$ executed on instance $I$.
Let $A_k(I) := \max_{p \in \mathscr{P}_k} \E\brk*{\text{val}(p, I)}$,\footnote{As mentioned by \cite{dean2008approximating}, we use $\max$ since one can show that the supremum implicit in the definition of $A_k(I)$ is attained.} where $\mathscr{P}_k$ is the set of all $k$-semi adaptive policies. Let $\ADAPT(I) := A_n(I)$ and $\ALG(I) := A_0(I)$ be the expected value of the optimal adaptive and non-adaptive policies, respectively. For these definitions, we do not restrict the computational complexity of computing these policies, although all of our positive guarantees are achieved by polytime algorithms. We also note that randomization does not yield a better approximation ratio: In any state (defined by the accumulated reward, remaining items and the remaining unoccupied knapsack size), the optimal strategy is deterministic.


\begin{definition}[Adaptivity Gap]\label{def:adaptivity-gap}
    For a stochastic knapsack instance $I$, the adaptivity gap is defined as $G(I) := \ADAPT(I) / \ALG(I)$.
    For a family of instances $\mathcal{I}$,
    $G(\mathcal{I}) := \sup_{I \in \mathcal{I}} G(I)$.
    The \emph{adaptivity gap of the SK problem} is the supremum over \emph{all} instances $G := \sup_I G(I)$.
\end{definition}


\begin{definition}[$\ztokgap$]
\label{def:ztok-gap} 
    The $\ztokgap$ is $G_{0}-{k} := \sup_I \crl*{\frac{A_k(I)}{\ALG(I)}}$.
\end{definition}

\begin{definition}[$\ktongap$]
\label{def:kton-gap}
    The $\ktongap$ is $G_{k}-{n} = \sup_{I} \crl*{\frac{\ADAPT(I)}{A_k(I)}}$.
\end{definition}

The next lemma ties the different gaps together. 

\begin{lem}[Semi-adaptivity gaps]
\label{lem:gaps-lemma}
    Let $G$, $G_{0\text{-}k}$, $G_{k\text{-}n}$ be the adaptivity gap,
    $0$-$k$ semi-adaptivity gap, and $k$-$n$ semi-adaptivity gap, respectively. Then $G_{k\text{-}n} \ge \frac{G}{G_{0\text{-}k}}$.
\end{lem}

\begin{proof}
For any stochastic knapsack instance $I$,
    \[
        G(I) = \frac{\ADAPT(I)}{\ALG(I)} = \frac{\ADAPT(I)}{A_k(I)} \cdot \frac{A_k(I)}{\ALG(I)} \le G_{k\text{-}n} \cdot G_{0\text{-}k},
    \]
where the last inequality follows from Defs.~\ref{def:ztok-gap}-\ref{def:kton-gap}. Thus $G = \sup_I \crl*{G(I)} \le G_{k\text{-}n} \cdot G_{0\text{-}k}$, as desired.
\end{proof}

\subsection{Auxiliary Definitions}

\begin{definition}[Effective Value]
\label{def:effective-value}
    The \emph{effective value} of item $i$ is $w_i := v_i \cdot \Pr(S_i \le 1)$.
\end{definition}

In words, the effective value of an item is an upper bound on the value any policy may obtain by inserting the item.
Relatedly, the mean truncated size of an item $i$ is the mean size of the item where any size realization bigger than $1$ is relaxed to $1$.

\begin{definition}[Mean Truncated Size]
\label{def:mean-truncated-size}
    The mean truncated size of item $i$ is defined as $\mu_i := \mathbb{E}[\min(S_i, 1)]$. For a set $A$ of items, define $\mu(A) := \sum_{i \in A} \mu_i$.
\end{definition}

The next definitions are useful for our analyses and policy design. 

\begin{definition}[Small and Large Items]
\label{def:small-items}
    For any $\eps \in (0,1)$, we say an item $i$ is $\eps$-\emph{small} if $\mu_i \le \eps$, and $\eps$-\emph{large} otherwise. Throughout we omit $\eps$ as $\eps$ is clear from context.
\end{definition}

Semi-adaptive policies insert blocks of items (an ordered set of items) and observe the knapsack remaining capacity after each inserted block.
\begin{definition}[Block of Items]
    A \emph{block} $B$ is an ordered sequence of items. Its total size, value, and effective value are defined as: $S(B) := \sum_{i \in B} S_i$, 
    $v(B) := \sum_{i \in B} v_i$, and $w(B) := \sum_{i \in B} w_i$.
\end{definition}

The following definition restricts the family of instances to ones where any policy gets its value from at most $k$ items (corresponding to a height-$k$ decision tree):
\begin{definition}[$\Hc_k$ instance family]
\label{def:H_k}
    For any arbitrarily small $\eps$, let $\Hc_k^\eps$ be the family of instances for stochastic knapsack problem with finite discrete distributions where the optimal solution gets at least $1 - \eps$ of its value from taking at most $k$ items. We often drop the $\eps$ in writing and simply denote this family of instances by $\Hc_k$.
\end{definition}

The last definition we give below addresses the importance of the order in which non-adaptive policies insert items into the knapsack. 
\begin{definition}[Order invariant]
\label{def:order-invariant}
    A non-adaptive policy for Stochastic Knapsack is considered order invariant if its expected value does not depend on the order in which it inserts the items.
\end{definition}

Policies for $\rsk$ are order invariant (the policy obtains all of the values for each realizations of the sizes if and only if all of them fit the knapsack together), while policies for $\nrsk$ are in general not.%
\footnote{A simple counter-example consists of a two-item instance $\crl*{1,2}$ where $S_1 = S_2 = \frac{2}{3}$ and $v_1=1$, $v_2=2$.}

%% file: technical-overview.tex
\section{Technical Overview}
\label{sec:tech-overview-conc-fut-dir}

\label{sec:tech-overview}

\subsection{$0$-$n$, $1$-$n$ and $k$-$n$ Semi-Adaptivity Gap Upper Bounds (Sections \ref{sec:one-to-n-gap}, \ref{sec:k-to-n-gap})} 

Our non-adaptive algorithm is intentionally simple. Unlike the algorithm of \cite{dean2008approximating}, which enumerates many candidate subsets of small items, we evaluate only three candidates. We first sort items by the greedy density \(\frac{w_i}{\mu_i}\) (Definitions~\ref{def:effective-value} and~\ref{def:mean-truncated-size}). Let \(P\) be a prefix of this ordering and let \(j\) be the next item. We compare the expected performance of: (i) \(P\), (ii) \(\{j\}\), and (iii) \(P\cup\{j\}\), and return the best one.

The analysis hinges on \cref{lem:greedy-block-value}, which shows that sufficiently ``large'' greedy prefixes capture a constant fraction of an LP upper bound on \(\ADAPT\). We then argue that one of the three candidates attains a comparable fraction of this LP value, yielding the claimed approximation guarantee.

To improve the bound with a single adaptive query, we add one more candidate policy: insert a block first, observe the remaining capacity, and then conditionally insert one additional item. We parameterize the resulting approximation ratio by five variables, and optimize this expression to obtain the stated bound.

\paragraph{From a single adaptive choice to $k$ adaptive choices.}
To move to $k = \tilde{O}(\nf{1}{\eps})$ adaptive choices, we 
partition the instance into two sub-instances of $\eps$-small and large items (an item is small if its truncated expected size is at most $\eps$ --- see \cref{def:small-items}). By \cref{lem:gen-two-choices-partition}, achieving a good approximation for each sub-instance implies an overall good approximation.

For \emph{large} items instances, we show (\cref{thm:semi-adaptive-large-items}) that the $\ktongap[\tilde{O}\prn*{\frac{1}{\eps}}]$ is at most $1 + \eps$. The proof idea is to show that using the following semi-adaptive strategy achieves a $1 + O(\eps)$ approximation: follow the optimal adaptive policy (or a $1 + \eps$ approximation of it) for $\tilde{O}\prn*{\frac{1}{\eps}}$ steps, and then switch to using the best known non-adaptive constant approximation policy.

For \emph{small} items instances, while inserting the items according to the greedy ordering yields a $2 + O(\eps)$ approximation for $\nrsk$, it completely fails for $\rsk$ as it results in an unbounded gap (\cref{example:ber-eps-items}). This necessitates a more careful strategy, which we provide via a semi-adaptive greedy policy.

Our iterative algorithm (\cref{alg:semi-adaptive-greedy}) benefits from a simple structure: it outputs a policy that attempts to insert \emph{blocks} of items based on the greedy ordering, targeting a specific expected block size. Crucially, we \emph{adaptively} set the block size as a \emph{linear function} of the remaining knapsack capacity, and optimize over the set of all linear functions to balance risk and reward. See \cref{fig:semi-adaptive-risky-greedy} for an illustration of how the policy operates.

\begin{figure}[h]
    \centering
    \includegraphics[scale=0.4]{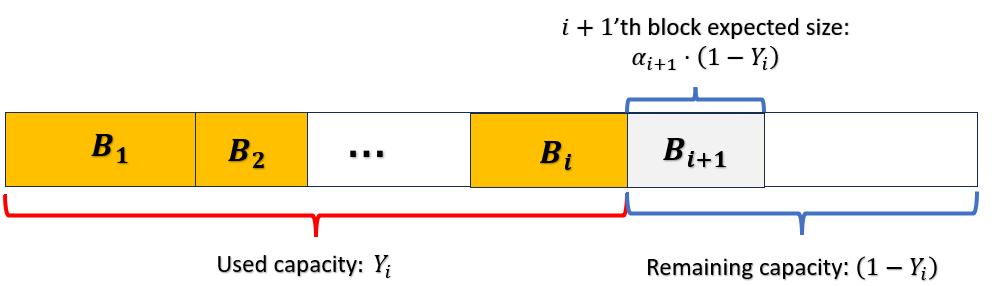}
    \caption{\footnotesize Illustration of the policy returned by \cref{alg:semi-adaptive-greedy} before inserting block $i+1$. After successfully inserting the first $i$ blocks: $B_1, \ldots, B_i$ with a total size of $Y_i$, the policy attempts to insert the $(i+1)$th block such that its expected size is a linear fraction (given by parameter $\alpha_{i+1}$) of the remaining capacity.}
    \label{fig:semi-adaptive-risky-greedy}
\end{figure}

Although in $\rsk$ the optimal policy must condition on both accumulated value and remaining capacity (see \cref{clm:ber-eps-items-behavior}), we show that using only the remaining capacity suffices to achieve a good approximation. 

We proceed by lower-bounding the remaining unused knapsack capacity after each successful block insert (\cref{lem:bound-rand-block-adapt-by-f-Phi-of-1}). This allows us, together with \cref{lem:greedy-block-value}, to show that \cref{alg:semi-adaptive-greedy} manages to get approximately an $f(\alpha)$-fraction of the LP value $\Phi(1)$, where $\alpha = (\alpha_1,\ldots,\alpha_{k+1}) \in (0,1)^{k+1}$ is a tunable parameter, $k$ is the parameter corresponding to the number of adaptive queries the policy makes, and $f(\alpha_1,\ldots,\alpha_{k+1}) := \min\crl*{1, \sum_{i=1}^{k+1} \alpha_i} \cdot \prod_{j=1}^{k+1} (1 - \alpha_j)$. As \cref{thm:alg-block-adaptive-gap-from-Phi-one} shows, by optimizing over all linear functions we get 
$
f(\alpha^*) \approx \prn*{\frac{k+1}{k+2}}^{k+1} \ \stackrel{{k\to\infty}}{\to} \ \ e,
$. The expected value of the optimal adaptive policy is bounded by twice the LP value (\cref{prop:adapt_lp_upper_bound}), which implies the $2e$ gap (\cref{cor:rsk-semi-adaptivity-gaps}).

\subsection{The Power of a Single Adaptive Choice: $0$-$1$ and $1$-$n$ semi-adaptivity gaps (Sections~\ref{sec:zero-to-one-gap}-\ref{sec:1-n-semi-lb-and-eps-noisy})}

Computing the supremum of the gaps over all instances is a hard problem. Indeed, so far the only existing lower bounds for $\rsk$ \cite{levin2014adaptivity} and $\nrsk$ \cite{dean2008approximating} are obtained via analyzing a simple $2$- or $3$-item instance. 
We introduce a $3$-step approach: 
\begin{enumerate}
    \item \textbf{Simplify.} Instead of general distributions, we consider simple item size distributions (e.g. deterministic, terminal (\cref{def:terminal-item}), $\epsnoisy$ (\cref{def:epsnoisy})).
    \item \textbf{Equalize non-adaptive options.} We equalize the expected values of different viable options for the optimal non-adaptive policy.
    \item \textbf{Optimize.} Find the worst gap by analyzing an optimization problem.
\end{enumerate}
Steps 1 and 2 help us limit the number of options of the optimal non-adaptive policy and get simple-structured instances that are easier to analyze. This enables bounding the adaptivity gap by a simple function of the instance (represented by a vector of real numbers). Step 3 optimizes to find the worst (largest gap) instance.

We apply this approach multiple times. In our $\ztokgap[1]$ analysis (\cref{sec:zero-to-one-gap}), each of the three steps preserves the adaptivity gap, allowing us to derive a tight result—that is, we obtain both the upper and lower bounds simultaneously. In \cref{sec:1-n-semi-lb-and-eps-noisy}, we apply the same three-step method with a recursive construction: starting from a ``worst-case'' instance (i.e., one with the largest gap) where the optimal adaptive decision tree has height $k-1$, we construct a new instance with height $k$. This leads to a recursive optimization that converges to a gap of $2$.

\paragraph{Overview of \cref{sec:zero-to-one-gap}.}

We start by introducing \cref{thm:reduction-to_H_k} for $\rsk$. It is a simple reduction that allows us to analyze the (full) adaptivity gap of a simpler family of instances, $\Hc_k$, instead of the $\ztokgap$ of general instances for $\rsk$. Any instance in $\Hc_k$ roughly translates to an instance where the optimal adaptive decision tree is of height $k$, as illustrated in \cref{fig:Hk-instance-reduction-illustration}.
\begin{figure}[h]
    \centering
    \includegraphics[scale=0.22]{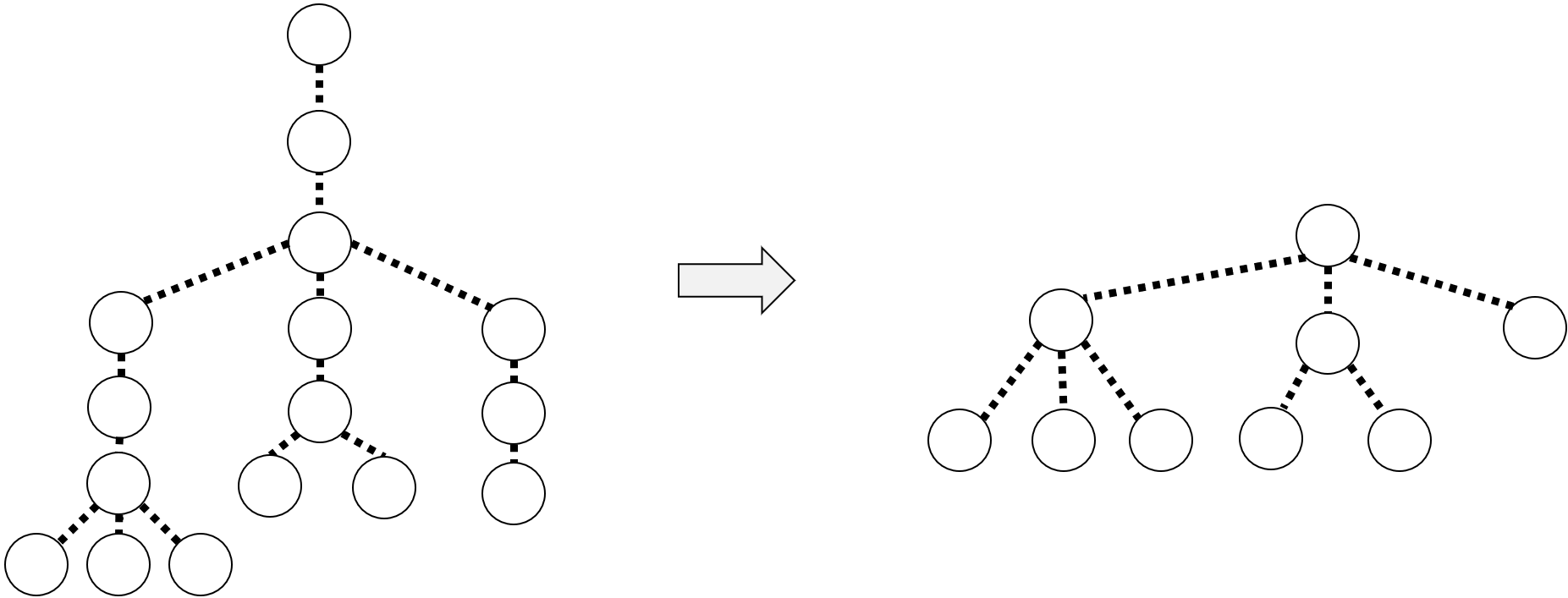}
    \caption{\footnotesize Illustration of the \cref{thm:reduction-to_H_k} reduction for $k=2$. The decision tree of the optimal $k$ semi-adaptive policy is on the left, and the tree for the corresponding $\Hc_{k+1}$ instance is on the right.}
    \label{fig:Hk-instance-reduction-illustration}
\end{figure}

\cref{thm:reduction-to_H_k} implies that instead of bounding the $\ztokgap[1]$, we may bound the adaptivity gap of $\Hc_2$ instances instead, which we do by utilizing our three steps approach, leading to a tight $1 + \ln(2)$ adaptivity gap.
Our first simplification step is given via \cref{lem:risky-F2-size-distributions}, which shows that almost all items can be replaced with simpler deterministic-sized items, while preserving the adaptivity gap; the only item that remains untouched is the first item that the optimal adaptive policy attempts to insert.

Our second adaptivity gap–preserving step is to equalize all viable choices available to the optimal non-adaptive policy (\cref{lem:risky-all-alg-options-yield-the-same}). The core idea, proved via induction, is that if one option is strictly better than another, we can modify the item values and redistribute the probability mass over the first item's size to eliminate this advantage, and we can do this without decreasing the adaptivity gap.
We then show that the restrictions of the second step imply a closed form for the adaptivity gap as a function of the first item size distribution. The third step is to show that the supremum of the mentioned function over all distributions is $1 + \ln(2)$, which implies the desired result (\cref{cor:zero-to-one-gap-of-rsk}). Similar ideas are used to show the $1 + \frac{1}{e}$ gap for $\Hc_2$ instances for $\nrsk$, with an additional idea used to handle the fact that $\nrsk$ policies are not order-invariant (see \cref{sec:nrsk-H2-gap}). To address this, we introduce terminal items (see \cref{def:terminal-item}), which—roughly speaking—nullify the value of items selected after them. This allows us to reduce the space of viable non-adaptive options while preserving the adaptivity gap.


\paragraph{Overview of \cref{sec:1-n-semi-lb-and-eps-noisy}.}

We apply the same three-step framework, now extended from $\Hc_2$ to more general $\Hc_k$ instances. We begin by establishing a lower bound of $2$ on the adaptivity gap for $\rsk$. By focusing on simpler $\epsnoisy$ distributions (see \cref{def:epsnoisy}), we significantly reduce the number of viable options for the optimal non-adaptive policy to linear in $n$ (Step 1).

In Step 2, we equalize (or restrict) the expected value across different non-adaptive choices. This leads to a recursive formula that bounds the adaptivity gap $G_k$ for an instance $I_k \in \Hc_k$ in terms of $G_{k-1}$, the gap for $\Hc_{k-1}$ instances. In the final step, we show that this recursive bound converges to $2$ as $k \to \infty$ (\cref{thm:risky-lb-two}). This lower bound of $2$, combined with \cref{lem:gaps-lemma} and \cref{cor:zero-to-one-gap-of-rsk}, implies a $\ktongap[1]$ lower bound of approximately $1.18$ (\cref{thm:single-adaptive-choice-1-to-n-gap-lb}).

A similar structure underlies the upper bound of $2$ on the adaptivity gap for $\epsnoisy$ instances in the $\nrsk$ setting (\cref{thm:nrsk-noisy-bernoulli-ub-2}). The key difference lies in Step 1, where we reduce the analysis to a carefully constructed instance that includes terminal items and admits a simple optimal adaptive decision tree (\cref{lem:non-risky-noisy-bernoulli-instance-decision-tree}).


%% file: one-to-n-gap-upper-bound.tex
\section{Power of Single Adaptive Choice: $0$-$n$ and $1$-$n$ Gap Upper Bounds}
\label{sec:one-to-n-gap}

In this section we give upper bounds on the $0$-$n$ (full adaptivity gap) and $1$-$n$ semi-adaptivity gap of $\rsk$, by showing the following theorem:

\begin{cor}
    For the problem of $\rsk$, the $0$-$n$ (full) adaptivity gap is at most $2\phi^3 \approx 8.47$, and the $1$-$n$ semi-adaptivity gap is at most $8.26$.
\end{cor}

We start with $0$ adaptive choices, improving the upper bound on the full adaptivity gap of $\rsk$. Specifically, we improve the $9.5$ bound of \cite{dean2008approximating} to $2\phi^3 \approx 8.47$. We show this by introducing an algorithm yielding a simple non-adaptive policy, described next.

Then, we show that we may get an even better approximation of at most $8.26$ by modifying our algorithm --- we add the option of utilizing a single adaptive choice for finding a non-adaptive policy.

The idea in both algorithms is to capture a large fraction of a linear program $\Phi(t)$, which is used to provide an upper bound on the expected value of an optimal adaptive policy (see \cref{prop:adapt_lp_upper_bound}).

\paragraph{An upper bound on the expected value of the optimal adaptive policy.}
We state a useful upper bound on $\ADAPT$, given by  \citet{dean2008approximating}\footnote{\cite{dean2008approximating} show that $\ADAPT \le \Phi(2)$, but  $\Phi(t)$ is a concave, non-decreasing function such that $\Phi(0) = 0$ and therefore $\Phi(2) \le 2 \ \Phi\prn*{1}.$}:

\begin{prop} \label{prop:adapt_lp_upper_bound}
    Let \[
    \Phi(t) := \max_{x} \crl*{\sum_i x_i w_i \;\bigg|\; \sum_{i} x_i \mu_i \le t, x_i \in [0,1] }. \] Then:
    \[\ADAPT \le 2 \ \Phi(1).
    \]
\end{prop}

The idea in the proof of \cref{prop:adapt_lp_upper_bound} is to first show that $\E[\mu(A)] \le 2$ where $A$ is the random set of items an adaptive policy attempts to insert (which follows from a simple martingale property and the fact that the sum of "truncated" inserted sizes in any execution path is at most $2$) and then use this property to show that a feasible solution to $\Phi(2)$ upper bounds $\ADAPT$.

\begin{rem}
$\ADAPT$ in \cref{prop:adapt_lp_upper_bound} refers to the expected value of the stronger optimal adaptive policy in the setting of $\nrsk$, which upper bounds the expected value of any policy in $\rsk$. Thus it holds for $\rsk$ as well.
\end{rem}

In fact, we conjecture that \cref{prop:adapt_lp_upper_bound} is not tight for $\rsk$:
\begin{rem}\label{remark:conj-rsk-lp-ub}
We conjecture that $\ADAPT \le 1.5 \ \Phi(1)$ where $\ADAPT$ is the expected value of an optimal adaptive policy in the setting of $\rsk$. Any such improvement will directly imply an improvement for the $\rsk$ algorithms we give in \cref{sec:k-to-n-gap}. \cref{example:ber-eps-items} implies a lower bound: $\ADAPT \ge \prn*{1 + \frac{1}{e}} \Phi(1) \approx 1.37  \ \Phi(1)$.
\end{rem}

We show an important property of the $\Phi(t)$:
\begin{lem}[Greedy block value]\label{lem:greedy-block-value}
    Let the greedy ordering of i=1 to n be: $\frac{w_1}{\mu_1} \ge \ldots \ge \frac{w_n}{\mu_n}$ and let $J$ denote the block of items $J=\crl*{1,...,j}$.
    Then: $w(J) \ge \min\crl*{1,\frac{\mu(J)}{t}} \Phi(t)$.
\end{lem}
We defer the proof of \cref{lem:greedy-block-value} to \cref{sec:proof-lem-greedy-block-value}. This lemma states that if a policy successfully inserts a prefix of the greedy ordering, the value it obtains is at least a $\mu(J)/t$ fraction of the LP benchmark $\Phi(t)$—up to a cap of 1.
\cref{lem:greedy-block-value} is used for both the algorithms in this section (\cref{alg:non-adaptive-risky-greedy}, \cref{alg:1-semi-adaptive-risky-greedy}), and \cref{alg:semi-adaptive-greedy} in \cref{sec:k-to-n-gap}.

\input{risky-non-adaptive-algorithm}
\input{risky-1-adaptive-algorithm}

%% file: risky-non-adaptive-algorithm.tex
\subsection{An Improved Non-Adaptive Policy} \label{sec:non-adapt-improved}

By \cref{lem:greedy-block-value}, if we could choose to insert a prefix block of the greedy ordering with an expected size of $\tfrac12$ we would get an adaptivity gap of~$8$: By Markov inequality w.p. at least $\nf{1}{2}$ the entire block would fit the knapsack. In case the entire block fits, a value of at least $\nf{1}{2} \cdot \Phi(1)$ is obtained (\cref{lem:greedy-block-value}), which implies an expected value of at least $\nf{1}{4} \cdot \Phi(1)$. Due to \cref{prop:adapt_lp_upper_bound}, this is at least $\nf{1}{8} \cdot ADAPT$.

However, such a ``$\tfrac12$-size prefix'' of the greedy ordering may not exist. And so, we introduce \cref{alg:non-adaptive-risky-greedy}, which attempts to get as close as possible to the $\tfrac12$ prefix. \cref{alg:non-adaptive-risky-greedy} chooses among three block candidates the one giving the best expected guarantee.

\begin{algorithm}[h]
\caption{\textsc{Non-Adaptive-Greedy}}
\label{alg:non-adaptive-risky-greedy}
\begin{algorithmic}[1]
    \State Sort items in non-increasing order of $\tfrac{w_i}{\mu_i}$
    \State $B \gets \emptyset,\; \ell \gets 1$
    \While{$\mu(B) + \mu_\ell \le 0.5$}
        \State $B \gets B \cup \{\ell\}$,\; $\ell \gets \ell + 1$
    \EndWhile
    \State $\alpha \gets \mu(B)$,\quad $\beta \gets \alpha + \mu_\ell$,\quad $\gamma \gets \tfrac{w_\ell}{\Phi(1)}$
    \State
        \begin{align*}
            o_1 &\gets \alpha(1-\alpha), &
            o_2 &\gets (\beta-\gamma)(1-\alpha), &
            o_3 &\gets \gamma, &
            o_4 &\gets \beta(1-\beta)
        \end{align*}
    \State $j \gets \operatorname*{arg\,max}_{i \in \{1,2,3,4\}} o_i$
    \If{$j \in \{1,2\}$}
        \State Insert all items in $B$
    \ElsIf{$j = 3$}
        \State Insert item $\ell$
    \Else \Comment{$j = 4$}
        \State Insert all items in $B \cup \{\ell\}$
    \EndIf
\end{algorithmic}
\end{algorithm}

See \cref{fig:non-adaptive-risky-greedy} for illustration of \cref{alg:non-adaptive-risky-greedy}.

\begin{figure}[h]
    \centering
    \includegraphics[scale=0.5]{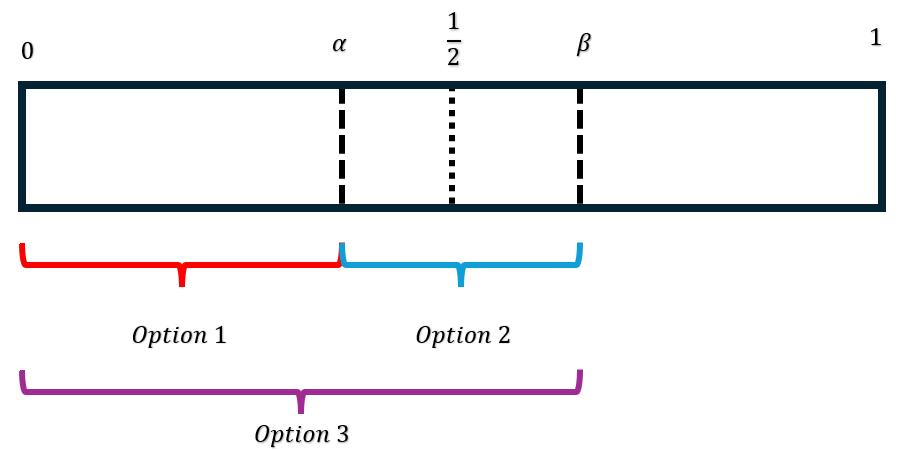}
    \caption{Illustration of the 3rd option of \cref{alg:non-adaptive-risky-greedy}: either take the $\alpha$-prefix of the greedy ordering, item $l$ only, or both (taking both is the same as taking the $\beta$-prefix of the greedy ordering).}
    \label{fig:non-adaptive-risky-greedy}
\end{figure}

\begin{thm}\label{thm:non-adaptive-risky-greedy-2-phi-cube}
    \cref{alg:non-adaptive-risky-greedy} has an adaptivity gap of at most $2 \phi^{3} \approx 8.47$, where $\phi := \frac{1 + \sqrt{5}}{2}$ is the golden ratio.
\end{thm}

\begin{proof}
If $\alpha = 0.5$ then the option of inserting only $B$ yields at least $\alpha(1-\alpha) \Phi(1) = 0.25 \ \Phi(1) \ge \frac{1}{8} ADAPT$, where the last inequality holds due to \cref{prop:adapt_lp_upper_bound}.
Otherwise, $\alpha < 0.5 \le \beta$, and thus there is such item $l$ such that $\mu_l = \beta - \alpha$.
If $\beta > 1$ then $\mu_\ell \ge 0.5$.  By \cref{lem:greedy-block-value},
$w(B) + w_\ell \ge \Phi(1)$, so
$\max\{w(B), w_\ell\} \ge \tfrac12\,\Phi(1)$.
Consequently, the chosen option earns at least $\tfrac12\,\Phi(1)$,
yielding an adaptivity gap of at most $4$.
Hence we may assume $\beta \in [0.5,1]$ for the remainder of the proof.

\medskip
Next, we show lower bounds for the three candidate options.
\begin{itemize}
    \item \emph{Options 1,2 (block $B$).}  
        Markov’s inequality gives
        $\Pr\bigl( S(B) \le 1 \bigr) \ge 1-\alpha$.
        Conditioned on fitting, \cref{lem:greedy-block-value} guarantees value
        at least $\alpha\,\Phi(1)$.
        Thus $\ALG \ge o_1\,\Phi(1)$. From \cref{lem:greedy-block-value} we know that $w(B \cup \crl*{l}) \ge min(1,\mu(B \cup \crl*{l})) \Phi(1) = \beta \cdot \Phi(1)$. But on the other hand $w(B) = w(B \cup \crl*{l}) - w_l$ and thus: $\ALG \ge w(B)(1-\alpha) \ge (\beta-\gamma)(1-\alpha) \Phi(1) = o_2 \ \Phi(1)$.

    \item \emph{Option 3 (single item $\ell$).}  
        By definition, inserting $\ell$ yields at least $\gamma\,\Phi(1)$,
        so $\ALG \,\ge\, o_3\,\Phi(1)$.

    \item \emph{Option 4 (block $B\cup\{\ell\}$).}  
        Again by Markov and \cref{lem:greedy-block-value},
        $\ALG \ge \beta(1-\beta)\,\Phi(1) = o_4\,\Phi(1)$.
\end{itemize}

\medskip
Next, we show a lower bound on the maximum of all options.
For every triple
$\alpha \in [0,0.5]$, $\beta \in [0.5,1]$, $\gamma \in [0,1]$,
we have shown:
\[
    \ALG \;\ge\; \max\{o_1,o_2,o_3,o_4\}\;\Phi(1). \numberthis \label{eq:ALG-T-Phi}
\]
Define
\(
    T := \min_{\alpha,\beta,\gamma} \max\{o_1,o_2,o_3,o_4\}.
\)
The following technical claim, shows a lower bound on $T$:
\begin{clm}\label{clm:T-ge-phi-minus-3}
\[
    T \;\ge\; \phi^{-3}.
\]
\end{clm}

\begin{proof}
Fix any $\alpha, \beta$. Since $o_2$ is a decreasing linear function of $\gamma$ such that $o_2(0) > 0$ and $o_3$ is an increasing linear function of $\gamma$ s.t. $o_3(0) = 0$ then : $\min_{\gamma \in [0,1]} \max \crl*{o_2, o_3}$ is obtained in the intersection point of $o_2, o_3$. So $(\beta-\gamma^*)(1-\alpha) = \gamma^*$ which implies $\gamma^* (2-\alpha)=\beta(1-\alpha)$ and thus $\gamma^* = \frac{\beta(1-\alpha)}{2-\alpha}$.

So we get: $T \ge \min_{\alpha \in [0,0.5], \beta \in [0.5, 1]} \crl*{o_1, o_4, h(\alpha,\beta)}$, where $h(\alpha,\beta) := \frac{\beta(1-\alpha)}{2-\alpha}$. For any fixed $\alpha$, in the domain $\beta \in [0.5,1]$: $h$ is a linear function which is increasing in $\beta$ while $o_4$ is a parabola decreasing in $\beta$ and so $ \max \crl*{o_4, h(\alpha,\beta)}$ is minimized in the intersection point $(1-\beta^*) = \frac{1-\alpha}{2-\alpha}$ or $\beta^* = \frac{1}{2-\alpha}$.
So $T \ge \min_{\alpha \in [0,0.5]} \ \max \crl*{o_1, g(\alpha)}$, where $g(\alpha) := \frac{1-\alpha}{(2-\alpha)^2}$. Both $o_1$ is increasing and $g$ is decreasing in $[0,0.5]$ and thus $\max \crl*{o_1, g(\alpha)}$ is minimized in the intersection point: $\alpha^*(1-\alpha^*) = \frac{1-\alpha^*}{(2-\alpha^*)^2}$ which is equivalent to $(1-\alpha^*)({\alpha^*}^2 - 3 \alpha^*+1) = 0$ which implies $\alpha^* = \frac{3-\sqrt{5}}{2}$.

Thus finally we get: $T \ge \frac{3-\sqrt{5}}{2} \prn*{1 - \frac{3-\sqrt{5}}{2}} = \frac{3-\sqrt{5}}{2} - \frac{9 - 6\sqrt{5} + 5}{4} = \sqrt{5} - 2 = \phi^{-3}$.

\end{proof}
From \cref{eq:ALG-T-Phi}, \cref{clm:T-ge-phi-minus-3}:
\[
    \ALG \ge \phi^{-3} \ \ \Phi(1).
\]
Combining this with \cref{prop:adapt_lp_upper_bound} (which states $\ADAPT \le 2\,\Phi(1)$) yields the claimed gap:
\[
    \frac{\ADAPT}{\ALG} \;\le\; \frac{2\,\Phi(1)}{T\,\Phi(1)}
    \;=\; 2\,\phi^3.
\]
\end{proof}

%% file: risky-1-adaptive-algorithm.tex
\subsection{Improved Approximation Via a Single Adaptive Choice}
\label{sub:single-query-improvement}

We show an efficient algorithm (\cref{alg:1-semi-adaptive-risky-greedy}) that finds a $1$-semi-adaptive policy --- a policy that makes at most a single adaptive choice that has improved performance compared to the best known non-adaptive policy.

The idea is similar to the non-adaptive case, but uses the idea of \textsc{Semi-Adaptive-Greedy}.

\begin{algorithm}[h]
\caption{\textsc{1-Semi-Adaptive-Greedy}}
\label{alg:1-semi-adaptive-risky-greedy}
\begin{algorithmic}[1]
    \State Sort items in non-increasing order of $\tfrac{w_i}{\mu_i}$
    \State $B \gets \emptyset,\; \ell \gets 1\;$
    \While{$\mu(B) + \mu_\ell \le 0.5$}
        \State $B \gets B \cup \{\ell\}$,\; $\ell \gets \ell + 1$
    \EndWhile
    \State $\alpha \gets \mu(B)$,\quad $\beta \gets \alpha + \mu_\ell$,\quad $\gamma \gets \tfrac{w_\ell}{\Phi(1)}$,\quad $p \gets \Pr(S(B) > 1)$, $\zeta = \min\crl*{\frac{\beta(\beta - \alpha)}{\gamma}, 1-\dfrac{2(\alpha-p)}{1-p}}$,
    \State
        \begin{align*}
            o_1 &\gets \alpha(1-p), \qquad
            o_2 \gets (\beta-\gamma)(1-p), \qquad
            o_3 \gets \gamma, \qquad o_4 \gets \beta(1-\beta), \\
            o_5 &\gets \max_{t \in \brk*{\zeta, \ 1}} \ \ \crl*{(1 - p) (\beta - \gamma) + \max\crl*{0,1-p-\dfrac{\alpha-p}{1-t}} \prn*{\gamma - \frac{\beta(\beta - \alpha)}{t}}}.
        \end{align*}
    \State $j \gets \operatorname*{arg\,max}_{i \in \{1,2,3,4,5\}} o_i$
    \If{$j \in \{1,2\}$}
        \State Insert all items in $B$
    \ElsIf{$j = 3$}
        \State Insert item $\ell$
    \ElsIf{$j = 4$}
        \State Insert all items in $B \cup \crl*{\ell}$.
    \Else \Comment{$j = 5$}
        \State Insert all items in $B$, and observe the size of the knapsack $S$.
        \State \[t \gets \argmax_{t \in \brk*{\zeta, \ 1}} \ \crl*{(1 - p) (\beta - \gamma) + \max\crl*{0,1-p-\dfrac{\alpha-p}{1-t}} \prn*{\gamma - \frac{\beta(\beta - \alpha)}{t}}}. \]
        \If{$1 - S \ge t$}
            \State Insert item $\ell$
        \EndIf
    \EndIf
\end{algorithmic}
\end{algorithm}

\begin{rem}
    \cref{alg:1-semi-adaptive-risky-greedy} computes the maximum of the term $o_5$ in line $7$, and the maximizing $t$ in line $17$. The term is a continuous function in $t$ and thus such a maximum exists in the closed segment $\brk*{\zeta, \ 1}$ by the Weierstrass extreme value theorem. In any case, for the analysis we give of the $8.26$ approximation it is enough for the algorithm to choose the fixed value of $t=1-\dfrac{2(\alpha-p)}{1-p}$. The reason for maximizing over $t \in \brk*{\zeta, \ 1}$ is that it may lead to a better outcome, as explained in \cref{rem-grid-search}.
\end{rem}

We show that the single adaptive choice yields an improved approximation ratio:
\begin{thm}\label{thm:alg-single-adaptive-choice-has-gap-at-most-8}
    The $1$-$n$ semi-adaptivity gap of \cref{alg:1-semi-adaptive-risky-greedy} is at most $8.26$.
\end{thm}

\begin{proof}
Just like in the proof of \cref{thm:non-adaptive-risky-greedy-2-phi-cube}, we may assume that $0< p \le\alpha<0.5 \le \beta \le 1$, and $\gamma \in [0,0.25]$ since otherwise we get an approximation ratio of at most $8$.

\cref{alg:1-semi-adaptive-risky-greedy} considers the options it has by either choosing a non-adaptive option like its non-adaptive counterpart, or choosing a new option: inserting the block of items $B$ and then adaptively deciding whether or not to insert item $l$ according to the realization of the size of block $B$.

Let us consider the adaptive alternative (case where $j=5$) of \cref{alg:1-semi-adaptive-risky-greedy}.

In this alternative the policy inserts the items of $B$, observes the knapsack and then decides to either stop inserting or insert item $l$ as well as a function of the remaining capacity: the policy inserts item $l$ iff the remaining capacity is at least $t$, where $t$ is a value the policy determines for each instance.

Let $q_t$ denote the probability the policy inserts item $l$ after inserting item $B$ given a fixed $t$. So \[
q_t := \Pr\prn*{1 - S(B) \ge t} = \Pr\prn*{S(B) \le 1 - t}.
\]
Let $p$ denote the overflow probability of $B$; that is $p := \Pr(S(B) > 1)$.

\begin{align*}
    \alpha = \E[S(B)] & = \E[S(B) \mid S(B) > 1] p  \\
    & \quad + \E[S(B) \mid 1 - t < S(B) \le 1] (1 - p - q_t) \\
    & \quad + \E[S(B) \mid 0 \le S(B) \le 1 - t] \ q_t. \numberthis \label{eq:bound-of-q_t}
\end{align*}

Consider the following lower bounds for the terms in the RHS of \cref{eq:bound-of-q_t}: $\E[S(B) \mid S(B) > 1] > 1$, $ \E[S(B) \mid 1 - t < S(B) \le 1] \ge 1 -t$ and $\E[S(B) \mid 0 \le S(B) \le 1 - t] \ge 0$. These together imply (by \cref{eq:bound-of-q_t}) that $\alpha \ge p + (1 -p -q_t) \ (1-t)$, or equivalently for $t < 1$:

\[
    1 - p -q_t \le \frac{\alpha - p}{1 - t}.
\]
or
\[
    q_t \ge 1 - p - \frac{\alpha - p}{1 - t}.
\]
So:
\[
    q_t \ge h(\alpha,p,t),
\]
where $h(\alpha, p, t) = 1 - p -\frac{\alpha - p}{1 - t}$ for $t < 1$ or $0$ otherwise.

On the other hand, $q_t \le 1 - Pr(S > 1) = 1 - p$.

We get:
\begin{equation}\label{eq:q_t-range}
    q_t \in \brk*{\max \crl*{h(\alpha,p,t), 0} , 1 - p}. 
\end{equation}

The expected value of this policy alternative is:
\begin{align*}
    \E[ALG] & \ge \Pr(S(B) > 1) \cdot 0 + \Pr(1 - t < S(B) \le 1) \cdot (\beta - \gamma) \ \Phi(1) \ \\
    & \quad + \Pr(0 \le S(B) \le 1 - t) \cdot \Pr\prn*{S(B) + S_l \le 1 \mid S(B) \le 1 - t} \cdot \beta \ \Phi(1) \\
    & \ge (1 - p - q_t) (\beta - \gamma) \cdot \Phi(1) + q_t \cdot \Pr(S_l \le t) \cdot \beta \ \Phi(1) \\ 
    & \ge \brk*{(1 - p - q_t) \ (\beta - \gamma) + q_t \prn*{1 - \frac{\beta - \alpha}{t}} \beta} \Phi(1).
\end{align*}

From the way the algorithm chooses $t$ (as the argmax over $(0,1]$ of the above expression) we get:

\begin{equation}\label{eq:alg-ge-o-4}
\E[ALG] \ge o_5 \cdot \Phi(1). \numberthis 
\end{equation}

As in \cref{thm:non-adaptive-risky-greedy-2-phi-cube}, the first option of inserting just the items of block $B$ yields an expected value of at least $\max\crl*{o_1, o_2} \ \Phi(1)$, and the second option of inserting just item $l$ yields an expected value of at least $o_3 \ \Phi(1)$. Analysis similar to before shows that 
\[
    \E[ALG] \ge T \cdot \Phi(1),
\]
where \[T := \min_{\substack{\alpha \in [0, 0.5],\\\beta \in [0.5, 1],\\\gamma \in [0,1],\\p \in [0,\alpha]}} \ \ \max_{i \in [5]} \crl*{ o_i}.\]

\begin{clm}\label{clm:T'-ge-quarter}
    $T \ge 0.24215$.
\end{clm}

The proof of \cref{clm:T'-ge-quarter} is quite technical and is postponed to~\cref{sec:proof-of-clm-T'-ge-quarter}. From \cref{eq:ALG-T-Phi} and \cref{clm:T'-ge-quarter}:
\[
    \ALG \ge 0.24215 \ \ \Phi(1).
\]
Combining this with \cref{prop:adapt_lp_upper_bound} (which states $\ADAPT \le 2\,\Phi(1)$) yields the claimed gap:
\[
    \frac{\ADAPT}{\ALG} \;\le\; \frac{2\,\Phi(1)}{T\,\Phi(1)}
    \;\le \; 8.2594.
\]
\end{proof}

\begin{rem}\label{rem-grid-search}
    A computer assisted grid search suggests a slightly better lower bound of $0.25$ on $T$, yielding a slight improvement in the approximation ratio to $8$ rather than $8.26$. Note however, that our analytic bound of $T \ge 0.24215$ is almost tight: by choosing $\alpha=p=0, \beta=0.5,\gamma=0$ we get $T \le 0.25$.
\end{rem}

%% file: k-to-n-gap.tex
\section{Power of Constant-Many Adaptive Choices: $k$-$n$ Gap}
\label{sec:k-to-n-gap}

In this section, we show that with only $\tilde{O}\prn*{\frac{1}{\eps}}$ adaptive queries, one can achieve a significantly better approximation of the fully adaptive policy than is known for non-adaptive policies:

\begin{thm}\label{thm:nrsk-k-to-n-gap}    
    The $\ktongap[\tilde{O}{\prn*{\frac{1}{\eps}}}]$ of $\nrsk$ is at most $3 + O(\eps)$.
\end{thm}

\begin{thm}\label{thm:rsk-k-to-n-gap}
    The $\ktongap[\tilde{O}{\prn*{\frac{1}{\eps}}}]$ of $\rsk$ is at most $2e + 1 + O(\sqrt{\eps}) \approx 6.44 + O(\sqrt{\eps})$.
\end{thm}

We begin with a simple yet useful lemma, which roughly states that if an instance can be partitioned into two sub-instances and we can approximate either one well, then we can approximate the full instance almost as well:

\begin{lem}\label{lem:gen-two-choices-partition}
    Let $I$ be a stochastic knapsack instance, and let $I_1, I_2$ be a partition of $I$ into two sub-instances. Suppose $A_1$ and $A_2$ are approximation algorithms for $I_1$ and $I_2$ with approximation ratios $c_1$ and $c_2$, respectively.

    Then the algorithm that partitions $I$ into $I_1$ and $I_2$ and runs either $A_1$ on $I_1$ or $A_2$ on $I_2$ (whichever yields the higher expected value) achieves an approximation ratio of at most $c_1 + c_2$.
\end{lem}

We provide the proof of \cref{lem:gen-two-choices-partition} in \cref{sec:proof-of-lem-gen-two-choices}. This lemma implies that if we obtain good approximations for ``small" and ``large" items separately, then we also obtain good approximations for general instances.

\subsection{Small Items}
For small items, \cite{dean2008approximating} show that a simple greedy non-adaptive ``bang-per-buck'' algorithm gives a $2 + O(\eps)$ approximation gap.
Together with \cref{thm:semi-adaptive-large-items} and \cref{lem:gen-two-choices-partition}, this yields the $3 + O(\eps)$ semi-adaptivity gap for $\nrsk$ (\cref{thm:nrsk-k-to-n-gap}). 

Unfortunately, this simple approach (of inserting the items according to the greedy ordering, non-increasing order of $\frac{w_i}{\mu_i}$) fails for $\rsk$: it has an unbounded adaptivity gap, as it may overflow w.p. $1$ (e.g. a simple instance with identical items with value $1$ and deterministic size $\eps$).

We proceed by giving a more careful algorithm that takes the \emph{optimal stopping} element of $\rsk$ into account to obtain a $2e + O(\sqrt{\eps})$ approximation of the optimal adaptive solution for the case of small items via a $O\prn*{\frac{1}{\sqrt{\eps}}}$ semi-adaptive policy. This implies \cref{thm:rsk-k-to-n-gap} (due to \cref{lem:gen-two-choices-partition}).

\input{k-to-n-gap-risky-small-items}

\subsection{Large Items}

In this section we show that for instances consisting of large items, $\tilde{O}(\frac{1}{\eps})$ adaptive queries suffice to achieve a $(1 + \eps)$ approximation:

\begin{thm}\label{thm:semi-adaptive-large-items}
Let $\eps < \frac{1}{8}$. For any instance in which all items have mean truncated size at least $\eps$ (i.e., $\forall i \in [n]: \mu_i \ge \eps$), the $\ktongap[O\prn*{\frac{1}{\eps} \log(\frac{1}{\eps})}]$ is at most $1 + \eps$. This holds for both $\rsk$ and $\nrsk$.
\end{thm}
\begin{proof}
    Let $m = \ceil{\frac{1}{\eps}}$ and let $k := 5m \ \ln(m)$.
    Let $T$ be the decision tree of the optimal adaptive policy. Let us assume w.l.o.g that items $1, \ldots , k$ are the items chosen by the policy (otherwise we can rename the items). 

    We observe the semi-adaptive policy that follows an optimal adaptive policy for the first $k$ rounds, and then follows the $c$-approximation non-adaptive policy that we know exists for a constant $c$. For $\nrsk$ we could use the $4$-approximate algorithm of \cite{dean2008approximating}, and for $\rsk$ we could use our $8.47$ approximate non-adaptive \cref{alg:non-adaptive-risky-greedy}. Choosing $c=8.5$ will work for both cases.

    For any $i \ge 1$ let $X_i$ be the indicator random variable of the event where knapsack overflow has occurred within the first $i$ rounds, given that items are picked according to the optimal adaptive policy. The expected inserted size in the first $k$ rounds is:
    \[
        \E[\sum_{i=1}^k S_i] = \sum_{i=1}^k \E[S_i] \ge k \min_{i \in [k]} \mu_i \ge \frac{k}{m} = 5 \ \ln(m),
    \] where the last inequality holds as the items are all large ($\mu_i \ge \eps$ for any $i \in [n]$).

    Let $\delta := 1 - \frac{1}{5 \ln(m)}$. Via Chernoff's multiplicative bound: 
    \begin{align*}
        1 - \Pr(X_k) &= \Pr \prn*{\sum_{i=1}^k S_i \le 1} \le \Pr \prn*{\sum_{i=1}^k S_i \le \frac{\E[\sum_{i=1}^k S_i]}{5 \ln(m)}} =\Pr \prn*{\sum_{i=1}^k S_i \le \prn*{1 - \delta}  \E[\sum_{i=1}^k S_i]} \\ 
        & \le exp\prn*{-\frac{\delta^2 \mu}{2}} = exp\prn*{-\frac{\prn*{1 - \frac{1}{5 \ln(m)}}^2 5 \ \ln(m)}{2}} \le \frac{1}{ m^2} \le \frac{1}{c m}, \numberthis \label{eq:bound-bad-event}
    \end{align*}

    where the last couple of inequalities are true for small enough yet constant $\eps$ ($\eps < \frac{1}{8}$).

    Let $G$ be the ratio between the value $\ADAPT$ gets and the value the semi-adaptive policy.

    By plugging the above \cref{eq:bound-bad-event} in the below total equation expectation of $G$ gets we get the desired:
    \begin{align*}
        \E[G] & = \E [G \mid X_k] \Pr(X_k) + \E [G \mid \neg X_k](1 - \Pr(X_k)) \\
        & \stackrel{(\star)}{\le} \E [G \mid X_k] + \frac{1}{cm} \E [G \mid \neg X_k] \stackrel{(\star\star)}{\le} 1 + \frac{c}{c m} \le 1 + \eps,
    \end{align*}

    Where $(\star)$ is due to the fact that trivially $\Pr(X_k) \le 1$ and \cref{eq:bound-bad-event}, and $(\star\star)$ is due to the fact that $\E[G \mid X_k] = 1$ (by the definition of $X_k$) and due to the bound of $c$ on the adaptivity gap in the general case.

\end{proof}
See \cref{sec:remark-large-items} for a remark on the computation of the optimal adaptive decision tree.

\subsection{Implied Improvements on the Best-Known Approximation Guarantees}\label{sec:k-to-n-implied-improvements}

Our algorithm implies an improvement on the best known fully adaptive policy approximation for $\rsk$: The $6.44 + O(\sqrt{\eps})$ approximation guarantee of \cref{alg:semi-adaptive-greedy} ( \cref{thm:rsk-k-to-n-gap}) improves the best known $8$-approximate fully adaptive policy of~\cite{fu2018ptas}.
The improvement is both in the approximation ratio and in the number of adaptive choices.

%% file: k-to-n-gap-risky-small-items.tex

\paragraph{A $2e$-approximate semi-adaptive policy for small items.}\label{sec:risky-k-to-n-gap-small-items}

In this section we show a semi-adaptive policy that achieves a large fraction (approximately $\frac{1}{e}$) of $\Phi(1)$ for an instance consisting of only small items.
Our algorithm, \cref{alg:semi-adaptive-greedy}, receives two parameters:  
$k$, the number of adaptive decisions;  
and $\alpha = (\alpha_1, \ldots, \alpha_{k+1})$, a vector controlling the target expected mass of each inserted block, where $\alpha_i$ is the fraction of remaining capacity targeted at round $i$.

While in $\rsk$, the optimal policy depends on both the remaining capacity and the accumulated value (as demonstrated in \cref{example:ber-eps-items}), our policy only adapts based on the remaining capacity, simplifying analysis.

\begin{algorithm}[H]
\caption{\textsc{Semi-Adaptive-Greedy}$(k, \alpha)$}
\begin{algorithmic}[1]
\State Sort items in decreasing order of $\frac{w_i}{\mu_i}$
\State $Y_0 \gets 0$, $m = k + 1$, $l=1$
\For{$i = 1$ to $m$}
    \State Initialize $B_i \gets \emptyset$
    \While{$\mu(B_i) + \mu_l \le \alpha_i (1 - Y_{i-1})$}
        \State Add item $l$ to $B_i$, $l \gets l+1$.
    \EndWhile
    \State Insert items in $B_i$ into the knapsack
    \State Observe the state of the knapsack and set $Y_i \gets Y_{i-1} + S(B_i)$
\EndFor
\end{algorithmic}
\label{alg:semi-adaptive-greedy}
\end{algorithm}

\cref{alg:semi-adaptive-greedy} orders the items according to the greedy ordering and inserts some prefix of it via inserting at most $m := k+1$ blocks. Let $p_j$ denote the overflow probability of the $j$-th block $B_j$, conditional on the success of blocks $1$ to $j-1$. That is, $p_j = \Pr(Y_j > 1 \mid Y_{j-1} \le 1)$ (since $Y_j$ is the used capacity of the knapsack after inserting blocks $1\ldots,j$).

We now prove a technical lemma bounding the expected remaining capacity after successfully inserting $i$ blocks.

\begin{lem}\label{lem:tech-lemma-remaining-capacity}
    \[
    1 - \E[Y_i \mid Y_i \le 1] \ge \prod_{j=1}^i \frac{1-\alpha_j}{1-p_j}
    \]
\end{lem}

\begin{proof}
We show the lemma by induction. For convenience, let $\alpha_0 = p_0 = 0$ and so $\frac{1- \alpha_0}{1 - p_0} = 1$. Thus, the inductive claim we make is $1 - \E[Y_i \mid Y_i \le 1] \ge \prod_{j=0}^i \frac{1-\alpha_j}{1-p_j}$.

The base of the induction $(i=0)$ clearly holds as:  $1 - \E[Y_0 \mid Y_0 \le 1] = 1 = \frac{1 - \alpha_0}{1 - p_0}$.
For the induction step, we assume $1 - \E[Y_{i-1} \mid Y_{i-1} \le 1] \ge \prod_{j=0}^i \frac{1-\alpha_j}{1-p_j}$.
By law of total expectation: 
\begin{align*}
    \E[Y_i \mid Y_{i-1} \le 1] &= \E[Y_i \mid Y_{i} > 1] \Pr(Y_i > 1 \mid Y_{i-1} \le 1) + \E[Y_i \mid Y_{i} \le 1] \Pr(Y_i \le 1 \mid Y_{i-1} \le 1) \\
    & \ge p_i + (1-p_i)\E[Y_i \mid Y_{i} \le 1]. \numberthis \label{eq:expected_y_i_ub}
\end{align*}

On the other hand, from linearity of expectation: 
\begin{align*}
    \E[Y_i \mid Y_{i-1} \le 1] & = \E[S(B_i) \mid Y_{i-1} \le 1] + \E[Y_{i-1} \mid Y_{i-1} \le 1] \\
    & \le \alpha_i (1 - \E[Y_{i-1} \mid Y_{i-1} \le 1]) + \E[Y_{i-1} \mid Y_{i-1} \le 1] \\
    & = \alpha_i + (1 - \alpha_i) \E[Y_{i-1} \mid Y_{i-1}]. \numberthis \label{eq:expected_y_i_second_equality}
\end{align*}

It follows from \cref{eq:expected_y_i_second_equality}, \cref{eq:expected_y_i_ub} that:
\begin{align*}
    \E[Y_i \mid Y_{i} \le 1] & \le \frac{\alpha_i - p_i + (1 - \alpha_i) \cdot \E[Y_{i-1} \mid Y_{i-1} \le 1]}{1 - p_i} \\
    & \implies 1 - \E[Y_i \mid Y_{i} \le 1] \ge \frac{(1 - \alpha_i)\prn*{1 - \E[Y_{i-1} \mid Y_{i-1} \le 1]}}{1 - p_i}  \ge \prod_{j=0}^i \frac{1-\alpha_j}{1-p_j},
\end{align*}
where the last inequality is due to the induction hypothesis.

\end{proof}

We are now ready to show the following lower bound on the policy's value:

\begin{lem}\label{lem:bound-rand-block-adapt-by-f-Phi-of-1}
    The expected value $\E[A]$ of \cref{alg:semi-adaptive-greedy} for small items $(\forall i: \mu_i \le \eps)$ is lower bounded by:
    \[\E[A] \ge \Phi(1) \cdot f(\alpha_1,\ldots,\alpha_m) - \eps m,
    \]
    where $f(\alpha_1,\ldots,\alpha_m) = \min\crl*{1, \sum_{i=1}^m \alpha_i} \cdot \prod_{j=1}^m (1 - \alpha_j)$.
\end{lem}
\begin{proof}
Let us bound $\E[A]$ where $A$ is the total value our policy gets:
\begin{align*}
    \E[A] & =  \Pr(\text{All $m$ inserted blocks fit the knapsack}) \cdot \E[A \mid \text{All inserted blocks fit the knapsack}] \\
    & = \prod_{i=1}^m \Pr(Y_i \le 1 \mid Y_{i-1} \le 1) \cdot \E[A \mid \text{All inserted blocks fit the knapsack}] \\
    & \ge \prod_{i=1}^m (1 - p_i) \cdot \min \crl*{\Phi(t), \frac{\Phi(t)}{t} \E[ \mu(\cup_{i=1}^m B_i) \mid Y_m \le 1]}.
\end{align*}
The second equality holds since the probability that the $i$'th block fit given that all previous blocks fit is exactly:
$\Pr(Y_i \le 1 \mid Y_{i-1} \le 1)$.
The inequality is true since we insert blocks according to the greedy ordering $\frac{w_1}{\mu_1} \ge \ldots \ge \frac{w_n}{\mu_n}$, and therefore it is implied from \cref{lem:greedy-block-value}.

By Markov's inequality: 
\begin{equation}\label{eq:p_i_le_alpha_i}
    p_i = \Pr(S(B_i) > 1 - Y_{i-1} \mid Y_{i-1} \le 1) \le \frac{\mu(B_i)}{1-{Y_{i-1}}} \le \alpha_i.
\end{equation}

If $\min \crl*{\Phi(t), \frac{\Phi(t)}{t} \E[ \mu(\cup_{i=1}^m B_i) \mid Y_m \le 1]} = \Phi(t)$ then by \cref{eq:p_i_le_alpha_i}: \[
\E[A] \ge \prod_{j \in [m]} (1 - p_j) \Phi(t) \ge \prod_{j \in [m]} (1 - \alpha_j) \Phi(t).\]

Otherwise: $\min \crl*{\Phi(t), \frac{\Phi(t)}{t} \E[ \mu(\cup_{i=1}^m B_i) \mid Y_m \le 1]} = \frac{\Phi(t)}{t} \E[ \mu(\cup_{i=1}^m B_i) \mid Y_m \le 1]$. In this case:
\begin{align*}
    \E[A] & \ge \frac{\Phi(t)}{t} \E[ \mu(\cup_{i=1}^m B_i) \mid Y_m \le 1] = \frac{\Phi(t)}{t} \prod_{i=1}^m (1 - p_i) \cdot \E[\sum_{i=1}^m \mu(B_i) \mid Y_m \le 1] \\
    & = \frac{\Phi(t)}{t} \prod_{i=1}^m (1 - p_i) \sum_{i=1}^m \E[ \mu(B_i) \mid Y_m \le 1] \stackrel{(\star)}{\ge} \frac{\Phi(t)}{t} \brk*{\prod_{i=1}^m (1 - p_i) \sum_{i=1}^m \prn*{\alpha_i (1 - \E[Y_i \mid Y_i \le 1] ) - \eps}}\\
    & \ge \frac{\Phi(t)}{t} \brk*{ \prod_{i=1}^m (1 - p_i) \sum_{i=1}^m \alpha_i \prod_{j=1}^i \frac{1-\alpha_j}{1-p_j} - \eps m \prod_{i=1}^m (1 - p_i)} \\
    & \ge \frac{\Phi(t)}{t} \brk*{\sum_{i=1}^m \alpha_i \prn*{\prod_{j=1}^i (1 - \alpha_j) \cdot \prod_{j=i+1}^m (1 - p_j)} - \eps m} \ge \frac{\Phi(t)}{t} \brk*{\prn*{\sum_{i=1}^m \alpha_i} \prod_{j=1}^m (1 - \alpha_j) - \eps m}.
\end{align*}

The first equality holds as the blocks we insert are disjoint. Inequality $(\star)$ holds because each inserted expected block size is at least $\alpha_i(1 - Y_i) - \eps$. This is true since: (1) we add item to block $B_i$ until its size is at most $\alpha_i(1 - Y_i)$, and (2) all items are small $(\forall i: \mu_i \le \eps)$. The last inequality is due to Markov's inequality, as $p_i = \Pr(S(B_i) > 1 - Y_{i-1} \mid Y_{i-1} \le 1) \le \frac{\mu(B_i)}{1-{Y_{i-1}}} \le \alpha_i$.

To summarize: Either $\E[A] \ge \Phi(t) \prod_{j = 1}^m (1 - \alpha_j)$ or $\E[A] \ge \frac{\Phi(t)}{t} \prn*{\sum_{i=1}^m \alpha_i} \prod_{j=1}^m (1 - \alpha_j) - \eps m$ and thus we get the desired by choosing $t=1$.

\end{proof}

We now optimize this expression by choosing $\alpha_i = \frac{1}{m+1}$ for all $i$, the symmetric maximizer of $f$.

\begin{thm}\label{thm:alg-block-adaptive-gap-from-Phi-one}
    By choosing $\alpha_i = \frac{1}{m+1}$ for all $i \in [m]$,
    \cref{alg:semi-adaptive-greedy} with parameters $(k,\alpha)$ achieves an expected value of at least $\Phi(1) \brk*{\prn*{\frac{k+1}{k+2}}^{k+2} - (k+1)\eps}$ for small items instances.
\end{thm}
\cref{thm:alg-block-adaptive-gap-from-Phi-one} is proved by analyzing the function~$f$; the full proof is deferred to \cref{sec:proof-of-thm-alg-block-adaptive-gap-from-Phi-one}.

We obtain the semi-adaptivity gap results as a direct corollary of \cref{thm:alg-block-adaptive-gap-from-Phi-one}:

\begin{cor}\label{cor:rsk-semi-adaptivity-gaps}
    The following $\ktongap$ results hold for $\rsk$ for a small items instance:
    \begin{enumerate}
        \item \cref{alg:semi-adaptive-greedy} with parameters $(k=0, \alpha=0.5)$ has an adaptivity gap of at most $8$.
        \item \cref{alg:semi-adaptive-greedy} with parameters $(k=1, \alpha=(\frac{1}{3}, \frac{1}{3})$ has $\ktongap[1]$ of at most $\frac{27}{4} = 6.75 - \eps$.
        \item Let $\eps \in (0,1)$. \cref{alg:semi-adaptive-greedy} with parameters $(k=\ceil{\frac{1}{2 \sqrt{\eps}}}, \alpha)$ such that $\alpha_i = \frac{1}{m+1}$ for all $i \in [m]$ has $\ktongap[\ceil{\frac{1}{2 \sqrt{\eps}}}]$ of at most $2e + \sqrt{\eps}\approx 5.44 + \sqrt{\eps}$.
    \end{enumerate}

\end{cor}
\begin{proof}
Due to \cref{prop:adapt_lp_upper_bound}, any policy that gets (in expectation) at least a $\frac{1}{c}$ fraction of $\Phi(1)$ has a gap of at most $2c$ from $\ADAPT$.

Therefore, all properties follow directly by plugging $k=0, k=1$ and $k = \ceil{\frac{1}{2\sqrt{\eps}}}$  in \cref{thm:alg-block-adaptive-gap-from-Phi-one},
where for $k = \ceil{\frac{1}{2 \sqrt{\eps}}}$ (and $m = k + 2$):
\begin{align*}
    \prn*{\frac{k+1}{k+2}}^{k+2} &= \prn*{\frac{m-1}{m}}^{m} = exp\prn*{\ln\prn*{\prn*{1 - \frac{1}{m}}^{m}}} = exp\prn*{ m\ln\prn*{1 - \frac{1}{m}}} \\
    & \stackrel{(\star)}{=} exp\prn*{m \prn*{-\frac{1}{m} - \sum_{n \ge 2} \frac{1}{n} \cdot \frac{1}{m^n}}} \ge exp\prn*{m \prn*{-\frac{1}{m} - \frac{1}{2} \sum_{n \ge 2} \frac{1}{m^n}}} \\
    & = exp\prn*{-1 - \frac{1}{2} \sum_{n \ge 1} \frac{1}{m^n}} = exp\prn*{-1 - \frac{1}{2(m-1)}} \\
    & = \frac{1}{e} - \frac{1}{2e(m-1)} = \frac{1}{e} - \frac{1}{2e(k+1)},
\end{align*}
where equality $(\star)$ follows from the Taylor expansion of $\ln(1+x)$. Thus:
\[
\prn*{\frac{k+1}{k+2}}^{k+2} - (k+1)\eps \ge \frac{1}{e} - \frac{1}{2e(k+1)} - (k+1)\eps \ge \frac{1}{e} - \sqrt{\eps},
\]
where the last inequality holds for the chosen $k$ (for small enough $\eps < 0.066$).

\end{proof}

\begin{rem}\label{remark:trade-off-rsk-small-items}
    \cref{cor:rsk-semi-adaptivity-gaps} highlights the trade-off between the number of adaptive choices and the quality of approximation. See \cref{fig:semi-gap-trade-off} for illustration. Non-adaptive execution ($k=0$) yields a gap of $8$, while using $O(1/\sqrt{\eps})$ adaptive choices improves the bound to $2e + \sqrt{\eps}$. Remarkably, even a single adaptive decision already improves the gap from $8$ to $6.75$.
\end{rem}

\begin{figure}[ht]
    \centering
    \includegraphics[width=0.7\linewidth]{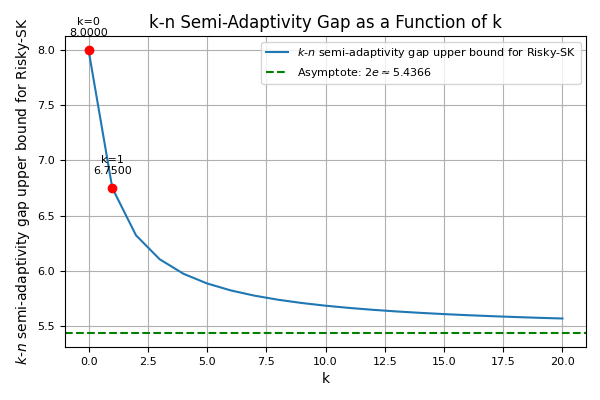}
    \caption{Illustration of the trade-off between $k$, the number of adaptive choices, and the resulting $\ktongap$ for small items instances of $\rsk$.}
    \label{fig:semi-gap-trade-off}
\end{figure}


%% file: zero-to-one-gap.tex
\section{Single Adaptive Choice Compared to No Adaptivity: $0$-$1$ Gap}
\label{sec:zero-to-one-gap}

In this section, we establish a tight $\ztokgap$ for stochastic knapsack with a single adaptive choice ($k = 1$):

\begin{thm}
\label{cor:zero-to-one-gap-of-rsk}
    For any instance of Risky-SK with finite discrete distributions, the $0$-$1$ semi-adaptivity gap is exactly $1 + \ln 2$.
\end{thm}

Our approach focuses on a simplified yet expressive family of instances, denoted by $\Hc_k$, which is easier to analyze (see \cref{def:H_k}), mostly for the case of $k=2$.
 The idea is to consider a more restrictive setting where any policy is not only constrained by the capacity of the knapsack (hence by the \textit{sizes} of the items it collects), but also by the \textit{number} $k$ of items it collects. The height of the optimal decision tree can then be seen as $k$, as any additional decisions may yield at most $\eps$ additional value.

In \cref{sec:reduction-from-0-k-gap-to-Z-k+1} we show that for $\rsk$, any adaptivity gap bound for $\Hc_{k+1}$ instances implies a $\ztokgap$ bound for general instances, where the proof is obtained via a reduction.


We show an exact analysis for the adaptivity gap for the family of instances $\Hc_2$ for both problem variants. In \cref{sub:adap-gap-H2} we analyze the gap for $\rsk$. In \cref{sub:0-1-improves} we give a similar analysis for $\nrsk$. 
The analysis of $\Hc_2$ instances is also of interest for $\nrsk$: it yields an improvement on the best known adaptivity gap lower bound; and also introduces new techniques and ideas, which we apply in \cref{sec:1-n-semi-lb-and-eps-noisy} to improve the upper bound on the adaptivity gap. 
Finally, it also sheds light on the similarity and difference between the two problem variants.

The class $\Hc_k$ may have broader applicability. In \cref{sec:1-n-semi-lb-and-eps-noisy} we show that if all items are ``almost" Bernoulli (taking sizes that are either close to $0$ or close to $1$) then we can get adaptivity gap bounds for $\nrsk$ and $\rsk$ 
by analyzing the gap of $\Hc_k$ instances for $k > 2$ and taking $k \to \infty$.

\subsection{Reduction to $\Hc_2$}\label{sec:reduction-from-0-k-gap-to-Z-k+1}

In this section we prove \cref{thm:reduction-to_H_k}, implies that a $\ztokgap[k]$ bound may follow from a bound on the adaptivity gap of instances in $\Hc_{k+1}^T$, given by the following definition:

\begin{definition}[Stochastic Knapsack with tree constraints]\label{def:SK-with-tree-constraints}
Let $I$ be a Stochastic Knapsack instance, and let $T$ be a rooted tree such that the nodes of $T$ are the items of $I$.

We say $I^T$ is an instance of Stochastic Knapsack with tree constraints if any policy may only insert items in $I$ that are all contained in at least one root-to-leaf path in $T$.

For any set $\mathcal{I}$ of instances of Stochastic Knapsack, let $\mathcal{I}^T$ denote the same set of instances with all possible tree constraints.

\end{definition}

The idea in tree constraints is to limit the options of the non-adaptive policies while maintaining the ones of the optimal adaptive policy.

\begin{lem}\label{thm:reduction-to_H_k}
    Let $G_k$ be the $\ztokgap$ of general instances and let $G(\Hc_{k+1}^T)$ be the (full) adaptivity gap of $\Hc_{k+1}$ instances with tree constraints.
    For $\rsk$:  $G_k \le  G(\Hc_{k+1}^T)$.
\end{lem}
\begin{proof}
We show that for any $\rsk$ instance $I$, there exists an instance with tree constraints $I'^T \in \Hc_{k+1}^T$ such that $G_k(I) \le G(I'^T)$ where $G_k(I)$ is the $\ztokgap[k]$ of $I$, and $G(I'^T)$ is the (full) adaptivity gap of $I'^T$.

Construct $I'^T$ as follows. Let $\sigma_{AD}$ be an optimal $k$-semi-adaptive policy for $I$, and let $T_I$ be its adaptive decision tree with $k$ adaptive adaptive decision points (where an adaptive decision point is a tree node with at least two children). Between any two adaptive decision points, the policy behaves non-adaptively—i.e., deterministically inserts a list of items $L = (i_1, \ldots, i_l)$.

For each such list $L$, define a new “compound item” $r_L$ with:
\[
    S_{r_L} = \sum_{j=1}^l S_{i_j}, \quad v_{r_L} = \sum_{j=1}^l v_{i_j}.
\]

We create a new tree $T$ that is obtained from $T_I$ be replacing each list of items $L$ with their compound item $r_L$.

Let $I'^{T}$ be the Stochastic Knapsack with tree constraints instance that consist of all such compound items $r_L$ (see \cref{fig:Hk-instance-reduction-illustration} for illustration) with the constraints of tree $T$.
\cref{fig:Hk-instance-reduction-illustration}.
\begin{figure}[h]
    \centering
    \includegraphics[scale=0.25]{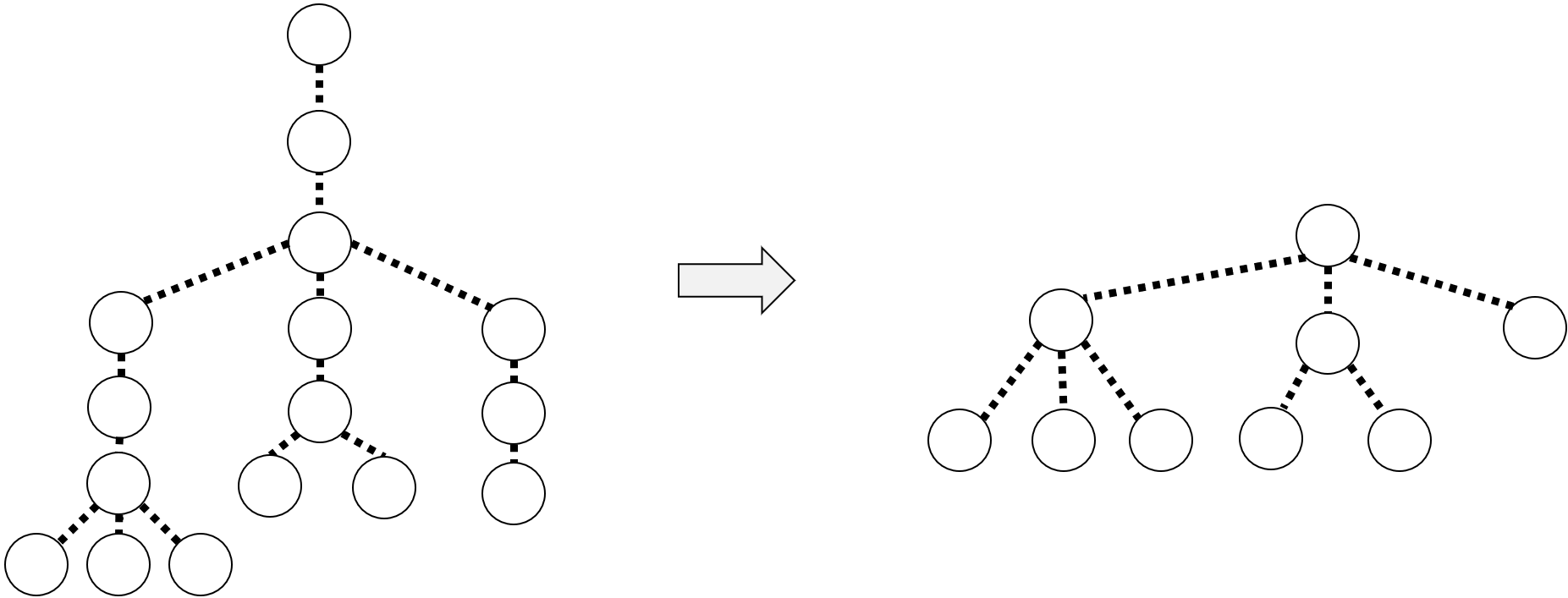}
    \caption{Illustration of the reduction for $k=2$. The adaptive decision tree of the optimal $k$ semi-adaptive policy is on the left, and the corresponding tree for the $\Hc_{k+1}^T$ instance is on the right. Each block of items between any two consecutive adaptive decision points on the left, become a single node (item) on the right.}
    \label{fig:Hk-instance-reduction-illustration-copy}
\end{figure}

Then:
\begin{itemize}
    \item The adaptive policy $\sigma_{AD}$ for $I$ naturally induces a $k$-adaptive policy in $I'^T$ with the same distribution over accumulated size and value: Consider an optimal adaptive policy of instance $I$. For any root-to-leaf path in its decision tree, by replacing each list $L$ in the path by its corresponding compound item in $I'^T$, we get an adaptive policy in $I'^T$. Since the adaptive policy inserts all items between two consecutive adaptive decision points, the knapsack capacity distribution in any decision point is the same for both policies. Also, the value of items inserted between any two decision points is also the same. Thus, $\ADAPT(I'^T) \ge \ADAPT(I)$.
    \item Any non-adaptive policy $\sigma'$ in $I'^T$ corresponds to one in $I$ that inserts the underlying original items in each $r_L$. The tree constraints make sure that for any two compound items $r_{L_1}$, $r_{L_2}$ that $\sigma'$ attempts to insert, it must be that $L_1 \cap L_2 = \emptyset$, and so the resulting non-adaptive policy for $I$ is well-defined. Hence $\ALG(I'^T) \le \ALG(I)$.
\end{itemize}
It follows that $G({I'}^T) = \frac{\ADAPT({I'}^T)}{\ALG({I'}^T)} \ge \frac{\ADAPT(I)}{\ALG(I)} = G_k(I)$.

We deduce $G_k \le G(\Hc_k^T)$.

\end{proof}
By \cref{thm:reduction-to_H_k}, we may analyze the adaptivity gap of $\Hc_2^T$ and get a semi-adaptivity gap for general instances for a single adaptive choice $(k=1)$. Although \cref{thm:reduction-to_H_k} yields an upper bound, the resulting instance we get also yields a tight lower bound, thus pinning down the exact $\ztokgap[1]$.

\input{single-adaptive-choice-H2-Gap}

\subsection{Implied Improvements on the Best-Known Approximation Guarantees}
\label{sub:0-1-improves}

Our $\Hc_2$ gaps imply improvements on the best known lower bounds on the full adaptivity gaps: From $1.5$ to $1.69$ for $\rsk$ (\cref{thm:single-adaptive-choice-gap-risky}) and from $1.25$ to $1 + \frac{1}{e} \approx 1.37$ for $\nrsk$ (\cref{thm:single-adaptive-choice-gap-non-risky}). We further improve the $\rsk$ adaptivity gap lower bound to $2$ in \cref{sec:1-n-semi-lb-and-eps-noisy} (\cref{thm:risky-lb-two}).

We defer the analysis for $\nrsk$ to \cref{sec:nrsk-H2-gap}.

%% file: single-adaptive-choice-H2-Gap.tex

\subsection{Adaptivity Gap of $\Hc_2$} 
\label{sub:adap-gap-H2}

In this section we show that for $\rsk$, the adaptivity gap of $\Hc_2$ family of instances is $1 + \ln(2)$:
\begin{prop}\label{thm:single-adaptive-choice-gap-risky}
    In $\rsk$ the adaptivity gap restricted to the family of instances $\Hc_2$ is the same with or without tree constraints, and is equal to $1 + \ln(2)$. That is: 
    $G(\Hc_2) = G(\Hc_2^T) = 1 + \ln(2) \approx 1.69$.
\end{prop}

We obtain \cref{cor:zero-to-one-gap-of-rsk} via \cref{thm:single-adaptive-choice-gap-risky} and \cref{thm:reduction-to_H_k}:

\paragraph{Notation.}

Let $I \in \Hc_2$ be an instance of $n+1$ elements with items $0,1,\ldots,n$. We assume $|I| > 2$ (otherwise the problem is trivial). W.l.o.g we assume $\sigma_{AD}$ always starts with item $0$, and always chooses another item after item $0$ given there was no overflow (otherwise we could always add an item with $0$ value and size). We normalize $v_0 = 1$ (any other choice can be scaled accordingly). We re-index the remaining items s.t. $v_1 \ge v_2 \ge \ldots \ge v_n > 0$. Let $p_0 := \Pr(S_0 > 1)$. For any $i \in [n]$, let $p_i := \Pr(\sad \text{ chooses item } i \text{ after item } 0 \mid S_0 \le 1)$, and $\sum_{i=0}^n p_i = 1$. The resulting adaptive decision tree is sketched in \cref{fig:z2-illustration}.

\begin{figure}[h]
    \centering
    \includegraphics[scale=0.4]{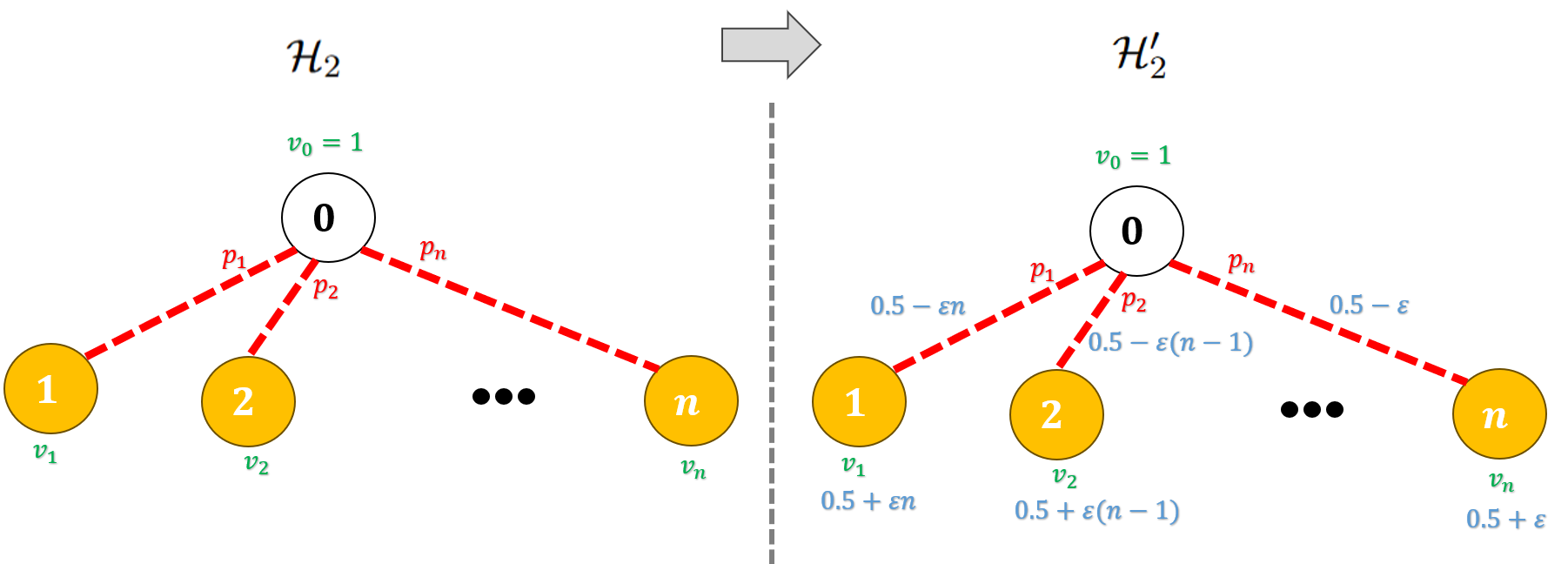}
    \caption{Illustration of the optimal adaptive decision tree in $\Hc_2$ and $\Hc'_2$ instances. On the left, there is the optimal adaptive decision tree of an $\Hc_2$ instance, where the first item taken is item $0$, giving a value of $v_0 = 1$. Each edge corresponds to taking item $i$ after item $0$, which happens w.p. $p_i$. In this event the optimal adaptive algorithm may get value $v_i$. On the right, after transforming the instance via \cref{lem:risky-F2-size-distributions}, we write the size distributions as well: each item $i$ has a deterministic size (written below its node in the tree), and item $0$ has item size $0.5 - \eps(n-i)$ for each edge leading to node $i$.}
    \label{fig:z2-illustration}
\end{figure}

We now perform a sequence of instance transformations, each of which preserves (up to~$\varepsilon$) the adaptivity gap.  By progressively ruling out instances whose gap cannot be maximal, we arrive at a structured sub‑family $\Hc''_2\subset \Hc_2$ containing a worst‑case instance.  Formally, we find a simpler $\Hc''_2 \subset \Hc_2$ such that: \[
\sup_{I \in \Hc_2} G(I) = \sup_{I' \in \Hc''_2} G(I'),
\]
where $G(I)$ is the adaptivity gap of instance $I$. Once we narrow down to  $\Hc''_2$, a family of instances with a specific structure, it is easier to argue about the worst adaptivity gap of an instance.

\paragraph{Step~1: Forcing simple size distributions on the items.}
\begin{lem}\label{lem:risky-F2-size-distributions}
    Let $n \in \N$, Let $\eps > 0$, $\eps \ll \frac{1}{n}$.
    Let $\Hc'_2 \subseteq \Hc_2$ be the family of instances $I \in \Hc_2$, $|I| = n$, s.t.:

\begin{enumerate}
    \item Item $i \in [n]$ has size $S_i = 0.5 + \eps \cdot (n-i)$.

    \item The size distribution of $S_0$ is: $S_0 = 0.5 - \eps \cdot i$ w.p. $p_i$ for any $i \in [n]$ and $S_0 = 2$ w.p. $p_0$ ($\sum_{i=0}^n p_i = 1$).
\end{enumerate}

    Then $G(\Hc'_2) = G(\Hc_2) = G(\Hc_2^T)$ for $\rsk$.
\end{lem}
The proof of \cref{lem:risky-F2-size-distributions} is deferred to \cref{sec:proof-risky-f2-size-dist}.

The intuition is that all items in the bottom layer of the adaptive decision tree (as depicted in \cref{fig:z2-illustration}) have size bigger than half, which assures no non-adaptive policy may pick more than one of these. One the other hand, any adaptive policy picking item $0$ and then some other item $i$ has the same probability of getting it as it did before this change. Determining the size distributions this way may not decrease $\ADAPT$ and may not increase $\ALG$ (the optimal adaptive decision tree is the same, but $\salg$ only less has options). In this case, any policy may choose to stop even after inserting a single element (to avoid overflow). Again, we get that there are only $n+1$ viable options for $\salg$:
Either insert item $0$ and then some item $i \in [n]$, or insert item $1$ and finish. 

As before, $\ADAPT = (1-p_0)\prn*{1 + \sum_{i=1}^n p_i v_i}$.

In $\rsk$ choosing item $0$ and then item $i$ yields a lower expected value of $\prn*{1 + v_i}\sum_{j=1}^i p_j$ (as $\sum_{j=1}^i p_j$ is the probability that both items fit the knapsack) compared to the expected value $\crl*{1 + v_i \sum_{j=1}^i p_j}$ of the same option in $\nrsk$.

\paragraph{Step~2: Equalizing all optimal non-adaptive branches.}
We show that we may focus on instances where all viable options for $\salg$ yield the same expected value:
\begin{lem}\label{lem:risky-all-alg-options-yield-the-same}
    Let $\Hc''_2$ be all instances of $\Hc'_2$ where:
    All options of $\salg$ yield the same expected value: $\ALG = v_1 = (1 - p_0)(1 + v_i)\sum_{j=1}^i p_j$ for any $i \in [n]$.

    Then $G(\Hc''_2) = G(\Hc'_2)$.
\end{lem}
The proof of \cref{lem:risky-all-alg-options-yield-the-same} is given in  \cref{sec:proof-risky-all-alg-options-yield-the-same}. The idea in the proof is to show that if there is an option of $\salg$ which yields strictly better expected value than other options, we may have a series of transformations such that each transformation redistributes the value and probability mass while potentially only increasing the adaptivity gap (increasing $\ADAPT$ and decreasing $\ALG$). Each transformation gets us closer to the goal of having all options having equal value by equalizing two options of $\salg$ at a time.

\paragraph{Step~3: Finding the worst instance via optimization.}
Let $q = 1 - p_0$. It follows from the above that $\ALG = v_1$ and $\ALG = q p_1 (1+v_1) = q (1+v_i) (\sum_{j=1}^i p_j)$. By reordering the last equation we get:
\[
v_i = \frac{p_1 (1 + v_1)}{\sum_{j=1}^i p_j} - 1. \numberthis \label{eq:G2-v_i}
\]

Hence:
\begin{align*}
    \frac{\ADAPT}{\ALG} & = \frac{q(1 + \sum_{i=1}^n p_i v_i)}{\ALG} = \frac{q(p_1 + p_1v_1)}{\ALG} + \frac{q(1 - p_1)}{\ALG} + q \frac{\sum_{i=2}^n p_i v_i}{\ALG} \\
    & = \frac{q(p_1 + p_1v_1)}{\ALG} + \frac{q(1 - p_1)}{\ALG} + q \frac{\sum_{i=2}^n p_i \prn*{\frac{p_1 (1 + v_1)}{\sum_{j=1}^i p_j} - 1}}{\ALG} \\
    & = 1 + \frac{q}{\ALG} \prn*{1 - p_1 - \sum_{i=2}^n p_i} + \frac{q p_1(1 + v_1)}{\ALG} \sum_{i=2}^n \frac{p_i}{\sum_{j=1}^i p_j} \\
    & = 1 + \sum_{i=2}^n \frac{p_i}{\sum_{j=1}^i p_j}, \numberthis \label{eq:G2-risky}
\end{align*}

where the second equality is due to \cref{eq:G2-v_i} and the last equality is due to \cref{lem:risky-all-alg-options-yield-the-same} and the fact that $1 - p_1 - \sum_{i=2}^n p_i = 0$.

It is interesting to compare \cref{eq:G-star-2} and \cref{eq:G2-risky}: in $\rsk$ we dropped the additional $p_1$ factor.

Since the function does not depend on $q$ we may set $q = 1$ (the probability that the first item overflows is $0$) as it has no effect on the adaptivity gap. 

The values $\crl*{v_i}$ are completely determined by the $\crl*{p_i}$ values. From \cref{lem:risky-all-alg-options-yield-the-same}: $v_1 = p_1(1+v_1)$ or: $v_1 = \frac{p_1}{1 - p_1}$ which implies \[
    1+ v_1 = \frac{1}{1 - p_1}. \numberthis \label{eq:one-plus-v1-in-risky-g2-analysis}
\]

For any $i \neq 1$: $v_i$ is determined by \cref{eq:G2-v_i}.

Unlike the case of $\nrsk$, $\crl*{v_i}$ values might be negative for some choices of $p_1,\ldots,p_n$.

From the non-negative constraints of the value ($v_i \ge 0$ for all $i \in [n]$) we get ($\cref{{eq:G2-v_i}})$ that for all $i \in [n]$:  $p_1 \ge \frac{\sum_{j=1}^i p_j}{1 + v_1}$. The maximum of these RHS values over all $i \in [n]$ is obtained for $i = n$ and yields that it is enough to demand
$p_1 \ge \frac{1}{1 + v_1}$ (as $\sum_{i=1}^n p_i =1$) for the non-negative constraints to hold. By plugging in \cref{eq:one-plus-v1-in-risky-g2-analysis} we get:
$p_1 \ge 1 - p_1$ or equivalently $p_1 \ge 0.5$.

So the optimization problem becomes to find the supremum of \[
f(p_1, \ldots p_n) = 1 + \sum_{i=2}^n \frac{p_i}{\sum_{j=1}^i p_j},
\] such that
\[
    \sum_{i=1}^n p_i = 1, \ \  p_1 \ge 0.5.
\]

We therefore now have the adaptivity gap as a function of $p = \prn*{p_1, \ldots, p_n} \in \Delta^{n-1}$, where $\Delta^{n-1}$ is the $n-1$ dimensional probability simplex. We continue via the following  technical lemma:
\begin{lem}\label{lem:tech-darboux-sum}
    \[
    \sup_{p \in \Delta^{n-1}} p_1 \sum_{i=2}^n \frac{p_i}{\sum_{j=1}^i p_j} \le \sup_{p_1 \in (0,1)} - p_1 \ln(p_1),
    \]
    and
    \[
    \lim_{n \to \infty}  \sup_{p \in \Delta^{n-1}} p_1 \sum_{i=2}^n \frac{p_i}{\sum_{j=1}^i p_j} = \sup_{p_1 \in (0,1)} - p_1 \ln(p_1),
    \]
\end{lem}
\begin{proof}

Let $T_i = \sum_{j=1}^i p_j$ for any $i \in [n]$. So for any $i \ge 2$: $p_i = T_i - T_{i-1}$. Let $g(x) = \frac{1}{x}$. So:
\begin{align*}
    & \sum_{i = 2}^n \frac{p_i}{\sum_{j=1}^i p_j} =  \sum_{i=2}^n \frac{T_i - T_{i-1}}{T_i}
     = \sum_{i=2}^n \prn*{T_i - T_{i-1}} g(T_i).
\end{align*}

$\sum_{i=2}^n \prn*{T_i - T_{i-1}} g(T_i)$ is the lower Darboux sum of $g$ for the partition defined by $(T_1, \ldots, T_n)$.

This sum is thus upper bounded by $\int_{T_1}^{T_n} g(t)dt =  \int_{p_1}^1 \frac{1}{t} dt = - \ln(p_1).$

And so:
\[
\sup_{p \in \Delta^{n-1}} p_1 \sum_{i=2}^n \frac{p_i}{\sum_{j=1}^i p_j} \le \sup_{p_1 \in (0,1)} - p_1 \ln(p_1).
\]

By choosing $p_i := \frac{1 - p_1}{n-1}$ and taking $n \to \infty$ we get equality instead of inequality (as the Darboux sum converges to the integral).

\end{proof}

From \cref{lem:tech-darboux-sum} we deduce that the optimum converges to (by taking $n \to \infty$) $1 + \ln(2) \approx 1.693$, and is obtained by choosing $p_1 = 0.5$, and $p_i = \frac{1}{2(n-1)}$ for $i \in [n] \setminus \crl*{1}$, yielding \cref{thm:single-adaptive-choice-gap-risky}.

%% file: Bernoulli_zero_one.tex
\section{Limitations of a Single Adaptive Choice: $1$-$n$ Gap Lower Bound}
\label{sec:1-n-semi-lb-and-eps-noisy}

In this section we show that a single adaptive choice is not enough for an algorithm to achieve a better-than-constant gap for the problem of $\rsk$:

\begin{cor}\label{thm:single-adaptive-choice-1-to-n-gap-lb}
    The $\ktongap[1]$ for $\rsk$ is at least $\frac{2}{1 + \ln(2)} \approx 1.18$.
\end{cor}

\cref{thm:single-adaptive-choice-1-to-n-gap-lb} follows from \cref{lem:gaps-lemma} (which ties together the two notions of semi-adaptivity gaps with the full adaptivity gap). We've already shown a $\ztokgap[1]$ upper bound of $1 + \ln(2)$ (\cref{cor:zero-to-one-gap-of-rsk}). We now show a (full) adaptivity gap lower bound of $2$ (\cref{thm:risky-lb-two}), and thus \cref{thm:single-adaptive-choice-1-to-n-gap-lb} follows.  


Technically, in this section we focus on the adaptivity gap for instances in which the item sizes are drawn from distributions that are ``almost'' Bernoulli: each item is either very small or nearly fills the entire knapsack. Specifically, we consider item size distributions supported only on values below $\eps$ or above $1 - \eps$, for an arbitrarily small $\eps > 0$. These can be interpreted as Bernoulli distributions perturbed by a small amount of noise.
Formally:
\begin{definition}[$\epsnoisy$ distribution]\label{def:epsnoisy}
    For any $\eps > 0$, a random variable $X$ has $\eps$-Noisy Bernoulli distribution if $\exists p \in [0,1]$ such that: \[
        X \in [1 - \eps, 1] \text{ with probability } p, \quad \text{and} \quad X \le \eps \text{ with probability } 1 - p.
    \]
\end{definition}

We refer to the event where $X \ge 1 - \eps$ as the \emph{large-size event}, and the event where $X \le \eps$ as the \emph{small-size event}.

We show that the adaptivity gap for this family of instances is at least $2$ for $\rsk$ (\cref{sec:eps-ber-lower-bound-rsk}).
For $\nrsk$, we show that the adaptivity gap for this family is at most $2$ (\cref{sec:eps-ber-nrsk-ub}).

\input{Bernoulli_zero_one_lb_2_risky}

\subsection{Implied Improvements on the Best-Known Approximation Guarantees}\label{sec:implied-improved-risky-lb}

For $\rsk$, our $2$ adaptivity gap lower bound (\cref{thm:risky-lb-two}) improves the previous best known $1.5$ lower bound of \cite{levin2014adaptivity}.

Interestingly, for $\nrsk$, similar ideas yield an improvement on the adaptivity gap upper bound rather than the lower bound.
We show that for  for instances with $\epsnoisy$ items, the adaptivity gap is at most $2$ (improving the existing $4$ upper bound gap \cite{dean2008approximating}). We defer this analysis to \cref{sec:eps-ber-nrsk-ub}.

%% file: Bernoulli_zero_one_lb_2_risky.tex
\subsection{Stronger Lower Bound on the Adaptivity Gap}

\label{sec:eps-ber-lower-bound-rsk}

In this section we show a lower bound of $2$ for the adaptivity gap of $\rsk$.

\begin{thm}\label{thm:risky-lb-two}
    For $\rsk$ the adaptivity gap is at least $2$. This is true even when restricting to the family of $\epsnoisy$ instances.
\end{thm}

The theorem is an immediate result of \cref{lem:risky_G_k_converges_to_two} below, to which we dedicate the remainder of the section.
For the first step of our $3$-step approach, we start by focusing on the simpler family of $\epsnoisy$ distributions:
Consider an instance with $\epsnoisy$ items: For any item $i$: $S_i = \eps_i \text{ w.p. } p_i$, and $1 - \eps_i \text{ w.p. } 1 - p_i$, for small $\eps_i \ll 1$. Let $w_i$ be the value of item $i$, and let $G_k$ be the adaptivity gap of the instance $I_k$, comprised of items $\crl*{1,\ldots,k}$. Finally, let $V_k$ be the expected value the best adaptive policy gets from instance $I_k$.

\begin{lem}\label{lem:risky_G_k_converges_to_two}
    There is a choice of $\crl*{w_i}_{i=1}^k$, $\crl*{p_i}_{i=1}^k$ such that: 
    \[
        \lim_{k \to \infty} G_k \ge 2.
    \]
\end{lem}
\begin{proof}
Let $k \in \N$ be the size of the instance we consider (instance $I_k$ with $k$ items).

\paragraph{Step 1: Simplify the instance.} We initialize $w_1 = 1$, $p_1 = 0$. Let $\eps > 0$ such that $\eps \ll 1$. We choose $\eps_i = \eps^i$. So $\eps_j \ll \eps_i$ for any $i < j$.

Consider the recursive adaptive strategy for instance $I_j$: insert item $j$ into the knapsack. If it realized to size $\eps_j$ continue recursively to instance $I_{j-1}$, and stop otherwise.
 By the definition of $V_{k-1}$, this strategy yields the following lower bound on the expected value of the best adaptive policy:
\[
    \ADAPT \ge w_k + p_k V_{k-1}. \numberthis \label{eq:adapt-lb-bernoulli-two}
\]

On the other hand, consider the options for $\salg$: it has three options: (1) insert item $k$ and stop, (2) insert item $k$ and then continue inserting other items from $I_{k-1}$, or (3) ignore item $k$ and simply insert items from $\crl*{1,\ldots,k}$.
Indeed, we may assume w.l.o.g. that if the non-adaptive policy inserts items according to the lexicographically decreasing order, since in $\rsk$ any non-adaptive policy is order-invariant.

Option (1) trivially yields a value of $w_k$. Option (2) yields a value of $p_k \prn*{w_k + \frac{V_{k-1}}{G_{k-1}}}$: if item $k$ realizes to $\eps_k$ then the policy may get at most $\frac{V_{k-1}}{G_{k-1}}$ from $I_{k-1}$ plus the value $w_k$ of item $k$. Otherwise item $k$ realizes to $1 - \eps_k$, leaving $\eps_k$ capacity in the knapsack which results in an overflow, since any item in $I_{k-1}$ realizes to size of at least $\eps_{k-1} > \eps_k$. Option (3) yields value of $\frac{V_{k-1}}{G_{k-1}}$. To summarize, we get an upper bound on the expected value of the best non-adaptive policy:
\[
    \ALG \le \max \crl*{w_k, p_k \prn*{w_k + \frac{V_{k-1}}{G_{k-1}}}, \frac{V_{k-1}}{G_{k-1}}}. \numberthis \label{eq:alg-ub-bernoulli-two}
\]

\paragraph{Step 2: Equalize all options of $\salg$.} Our next step is to impose restrictions such that all options of the best non-adaptive policy yield the same expected value. 
First, we impose the restriction: $\frac{V_{k-1}}{G_{k-1}} = p_k \prn*{w_k + \frac{V_{k-1}}{G_{k-1}}}$ and get:
\[
    w_k = \frac{V_{k-1}}{G_{k-1}} \prn*{\frac{1}{p_k} - 1}.
\]

Next, we impose the restriction $w_k = \frac{V_{k-1}}{G_{k-1}}$ and get together with the previous equation that $p_k = \frac{1}{2}$.
Plugging this back in \cref{eq:adapt-lb-bernoulli-two} and \cref{eq:alg-ub-bernoulli-two} we get:
\[
    \ADAPT \ge V_{k-1} \prn*{\frac{1}{2} + \frac{1}{G_{k-1}}}, \quad \ALG \le \frac{V_{k-1}}{G_{k-1}},
\] which together imply:
\[
G_k = \frac{\ADAPT}{\ALG} \ge \frac{V_{k-1} \prn*{\frac{1}{2} + \frac{1}{G_{k-1}}}}{\frac{V_{k-1}}{G_{k-1}}} = G_{k-1}\prn*{\frac{1}{2} + \frac{1}{G_{k-1}}} = \frac{G_{k-1}}{2} + 1.
\]
\paragraph{Step 3: Show the convergence.}
Let $L := \lim_{k \to \infty} G_k$. By taking the limit of the above inequality we get $L \ge \frac{L}{2} + 1$ which implies $L \ge 2$. We deduce that for any $\eps$ there exists $k$ and an instance in $I_k \in \Hc_k$ with adaptivity gap of at least $2 - \eps$, which implies the desired.

\end{proof}

%% file: conclusions.tex
\section{Conclusion and Future Directions}
\label{sec:conclusions-and-future-directions}

We study two families of semi-adaptivity gaps. 
Each notion provides a different perspective on the power of semi-adaptive policies. The $\ztokgap$ is highly relevant, since comparing against an optimal semi-adaptive policy storable in $\text{poly}(n)$ memory is a more realistic benchmark. Also, the $\ztokgap$ opens a new path via \cref{lem:gaps-lemma} for obtaining $\ktongap$ lower bounds. 
Our work differs from previous work studying $k$-$n$ gaps in its focus on constant~$k$. 
Overall, our results demonstrate the trade-off between the number of adaptive choices and performance, and in particular the power of a single adaptive choice. Our techniques turn out to be useful for improving the (full) adaptivity gap bounds of $\rsk$ and $\nrsk$. 



There are many avenues for further research. One interesting direction is to improve our semi-adaptivity gap results for Stochastic Knapsack. We state one concrete direction in \cref{remark:conj-rsk-lp-ub}. 
Also, the problem of pinning down the (full) adaptivity gap is still open. As we've shown, any further lower bound improvements for the adaptivity gap can also lead to bounds on the semi-adaptivity gaps. 
Separating between non-adaptive and $k$-semi-adaptive policies for constant $k$ is another future direction. 
Finally, we believe our notions of semi-adaptivity gaps, as well as some of our techniques (such as the adaptive decision tree analyses via the Simplify-Equalize-Optimize $3$-step approach), may be useful for studying other combinatorial stochastic optimization problems beyond Stochastic Knapsack (e.g. Stochastic Probing, Stochastic Orienteering).

%% file: acknowledgements.tex
\section{Acknowledgments}

We thank Yossi Azar, Asnat Berlin, Batya Berzack, Ilan Reuven Cohen, Alon Eden, Anupam Gupta and
Jan Vondrák for useful discussions and referring us to related work.
We are grateful to anonymous reviewers for their helpful feedback that greatly improved the paper.
This work received funding from the European Research Council (ERC) under the European Union's Horizon 2020 research and innovation program (grant No.: 101077862), from the Israel Science Foundation (grant No.: 3331/24), from the NSF-BSF (grant No.: 2021680), and from a Google Research Scholar Award.

%% file: appendix.tex







\input{appendix_intro}

\input{proofs-tech-claims-section-1-n-gap-ub}

\input{proofs-k-to-n-gap}

\input{proofs-zero-to-one-gap}

\section{$\nrsk$ Analysis}\label{sec:nrsk-analysis}

\input{single-adaptive-choice-nrsk}

\input{Bernoulli_zero_one_ub_2_nrsk}

%% file: appendix_intro.tex
\section{Appendix for \cref{sec:intro}}\label{sec:appendix-intro}

\input{Ber_eps_behavior_proof}

\input{more-related-work}

%% file: Ber_eps_behavior_proof.tex
\subsection{Proof of \cref{clm:ber-eps-items-behavior}}
\label{sec:proof-of-clm-ber-eps-items-behavior}

\begin{proof}[Proof of \cref{clm:ber-eps-items-behavior}]

Let $\sad$ be the optimal adaptive policy for Stochastic Knapsack.

For both $\rsk$ and $\nrsk$, since all items are identical any policy has only two options: insert or stop. In the case of $\nrsk$ there is nothing to lose from inserting and so the optimal policy always inserts.

Let $T$ be the stopping time of the execution path of $\sad$, that is the time step in which the policy either stops or an overflow occurs. Since the items are identical, the total value a policy gets is determined by the number of items it successfully inserts into the knapsack. So $\ADAPT = \E[T]$.

Since the items have a positive expected size ($\eps$) the overflow probability of $\sad$ is 1.

For $\nrsk$ the number of items inserted is a negative binomial random variable (the number of successful inserts until the second realization to size $1$ leading to an overflow minus the last element): $T \sim NB(2,\eps) - 1$ and thus $\E[T] = \frac{2}{\eps} - 1$. Since the value of each item is $\eps$ we get the first property. \\

The proof for the second  property ($\rsk$ setting) is a bit more involved.

We now show that for $\rsk$: $\E[T] \le \frac{1}{\eps} \prn*{1 + \frac{1}{e} + O(\eps)}$.

There are two phases of execution. In the first phase, the knapsack capacity is $0$ and $\sad$ keeps inserting items (as it is the optimal thing to do - there is no risk in this phase). The first phase is over when the first size $1$ realization occurs, after $T_1$ time and then the second phase begins, and lasts until $\sad$ chooses to stop or an overflow has occurred. Let $T_2$ denote the second phase duration.
So $T = T_1 + T_2$, and thus from linearity of expectation:

\begin{equation}\label{eq:T-is-the-sum-of-T1-and-T2}
    \E[T] = \E[T_1] + \E[T_2] = \frac{1}{\eps} + \E[T_2],
\end{equation}
where the last equality is due to the fact that $T_1 \sim Geo(\eps)$.

We will show that $\E[T_2] \le \frac{1}{e} \frac{1}{\eps}$. Via the law of total expectation:
$\E[T_2] = \E \brk*{\E\brk*{T_2 \mid T_1}}$.

Since we assume an unbounded number of items, the behavior of $\sad$ in the second phase is exactly as the one of the optimal policy in the following infinite MDP:
Second-Phase-MDP
\begin{itemize}
    \item The states are: 
    \begin{enumerate}
        \item Non-absorbing states (corresponding to states where $\sad$ can make decisions): $\N$
        \item Absorbing states (the overflow state and the stopping states): $\crl*{overflow, stop}$.
    \end{enumerate}
    \item Actions: For each non-absorbing state $i$ there are two actions:
    \begin{enumerate}
        \item Stop. This action leads to state $stop$ w.p. $1$ and its reward is $i$.
        \item Insert. This action has reward $0$ and it leads to state $stop$ w.p. $\eps$, and to $i+1$ w.p. $1 - \eps$.
    \end{enumerate}
\end{itemize}

\begin{figure}
    \centering
    \includegraphics[scale=0.7]{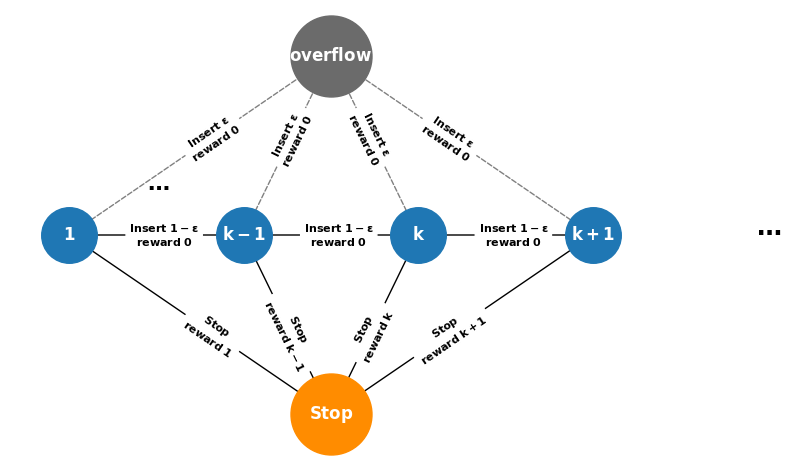}
    \caption{Second phase MDP illustration. Insert action has probability $\eps$ of overflowing and probability $1-\eps$ to move to a state with one more item successfully inserted into the knapsack. Stop has probability $1$ of not overflowing and reaching a finite state.
    }
    \label{fig:enter-label}
\end{figure}

Let $V(i)$ denote the optimal expected reward when we are in non‑absorbing state $i$.
Then the Bellman equation we get is: $V(i) = \max\crl*{i, (1 - \eps) V(i+1)}$.

We will show that the optimal policy for this $MDP$ is a threshold policy and find it.

\begin{lem}\label{lem:mdp-k-threshold}
There exists a threshold such that $k \in \N$ such that: $\forall i < k$: the best action in state $(i,1)$ is to insert, and
$\forall i \ge k$: the best action in state $(i,1)$ is to stop.
\end{lem}
\begin{proof}
Assume by contradiction that for any $i \in \N$ the best action is to insert. The policy will overflow w.p. $1$ (as the $MDP$ is infinite and there's a constant probability of overflow in each step), and in no step it will acquire any reward (since positive reward is only obtained via a stopping action). Thus, the policy would always have value $0$. However, the policy simply always stops for any $i$, has positive expected value. We've reached a contradiction, and thus there must exists a state $k$ where the optimal action is to stop.
W.l.o.g. let $k$ be the minimum index for which the best action in state $k$ is to stop (otherwise we could rename $k$ as the minimum index).
For any $i > k$ the optimal action in state $i$ must also be to stop since the additional expected value (over stopping) that might be gained in the future is no bigger than in state $k$, the overflow probability is the same as in state $k$, and the value obtained via stopping is bigger than in state $k$.




\end{proof}

Let $k$ be the threshold index mentioned in \cref{lem:mdp-k-threshold}. The optimal policy in state $k-1$ is thus to insert and therefore via the Bellman equation: ${k-1} \le V(k-1) = (1 - \eps) V(k)$. But $V(k) = k$, and thus:
${k-1} \le (1 - \eps) k$ which implies $k \le \frac{1}{\eps}$. On the other hand, the optimal policy in state $k$ is to stop and thus $k = V(k) \ge (1-\eps)V(k+1) = (1-\eps)(k+1)$ which implies: $0 \ge -\eps k + 1 - \eps$ or $k \ge \frac{1 - \eps}{\eps}$. In total, we get $k = \frac{1}{\eps} + O(1)$. For the sake of simplicity we ignore the $O(1)$ term as it is negligible and assume $k = \frac{1}{\eps}$ (this contributes the additional $O(\eps)$ term to the claim).

We note that given $T_1$ it must be that $T_2 = max(k - T_1, 0)$: If $T_1 \ge k$ then the optimal decision for $\sad$ according to our $MDP$ analysis is to stop, and otherwise ($T_1 < k$), $\sad$ is in state $T_1$ in the MDP and thus (via optimality) it will try and insert until reaching state $k$ which implies at most $T_2 = k - T_1$ inserts. And so:
\begin{equation}
    \E[T_2 \mid T_1 , T_1 \le k] = \E[k - \E[T_1 \mid T_1 \le k]] = k - \E[T_1 \mid T_1 \le k] = k - \E[T_1 \mid T_1 \le \frac{1}{\eps}].
\end{equation}

We are now ready to get a bound for $\E[T_2 \mid T_1]$:

\begin{align*}
    \E[T_2 \mid T_1] &= \E\brk*{T_2 \mid T_1, T_1 \le k} \Pr(T_1 \le k) + \E\brk*{T_2 \mid T_1, T_1 > k} \Pr(T_1 > k) \\
    & = \prn*{k - \E[T_1 \mid T_1 \le k]} \Pr(T_1 \le k). \numberthis \label{eq:E-T_2-given-T-1-bound-one}
\end{align*}

where the inequality is since $\E[T_2 \mid T_1, T_1 > k] = 0$.

Let us find $\E[T_1 \mid T_1 \le k]$:

\[E[T_1 \mid T_1 \le k] = \frac{\displaystyle\sum_{t=1}^k t\,P(T_1=t)}{P(T_1\le k)},\] and for $q = 1 - \eps$:
\begin{align*}
    \sum_{t=1}^k t\,P(T_1=t) &= \eps \sum_{t=1}^k t \,q^{\,t-1} = \eps \frac{1 - (k+1)q^k + k\,q^{\,k+1}}{\eps^2} = \frac{1 - (k+1)q^k + k\,q^{\,k+1}}{\eps} \\
    & = \frac{(1 - q^k) + k q^k (1 - q)}{\eps}.
\end{align*}

$\Pr(T_1 \le k) = 1- q^k$ and therefore:

\[
    \E[T_1 \mid T_1 \le k] = \frac{(1 - q^k) + k q^k (1 - q)}{\eps (1 - q^k)} = \frac{1}{\eps} + \frac{k q^k}{(1 - q^k)}.
\]

We plug this back in \cref{eq:E-T_2-given-T-1-bound-one} and get: 

\begin{align*}
    \E[T_2 \mid T_1] & = \prn*{k - \frac{1}{\eps} + k q^k ( 1 - q^k)} (1 - q^k) \\
    & = k q^k = k (1 - \eps)^k = \frac{1}{e} \frac{1}{\eps} + O(1), \numberthis \label{eq:T2-given-T1-bound}
\end{align*}
where the last equality is true since for $f(x) = x(1-\eps)^x = x e^{x \ln(1-\eps)}$: by plugging in $k = \frac{1}{\eps}$ we get:
$\frac{1}{\eps}e^{\frac{1}{\eps} \ln(1-\eps)} = \frac{1}{\eps}e^{\frac{-\eps + O(\eps^2)}{\eps}} = \frac{1}{\eps} \frac{1}{e} + O(1)$.

\cref{eq:T2-given-T1-bound} together with \cref{eq:T-is-the-sum-of-T1-and-T2} yield:
\begin{equation}\label{eq:expected-T-ber-eps-example-risky}
    \E[T] = (1 + \frac{1}{e}) \frac{1}{\eps} + O(1),
\end{equation} which implies the claim on the expected value of $\sad$. We now show that the overflow probability of $\sad$ is $\frac{1}{e} + O(\eps) \approx 0.37$.

Let $Z$ be the sum of item sizes over all items that $\sad$ has inserted the knapsack: $Z = \sum_{t=1}^T S_t$.

Let $p$ be the probability that an overflow occurs in the run of $\sad$. Then $p = \Pr(Z > 1)$. By law of total expectation:
\begin{align}\label{eq:E-of-Z-risky-one-plus-p}
    \E[Z] &= \E[Z \mid Z > 1] p + \E[Z \mid Z \le 1](1-p) \le 2p + (1-p) = 1+p.
\end{align}
On the other hand, via Wald's equation:
\begin{equation}\label{eq:E-of-Z-risky-one-plus-one-over-e}
    \E[Z] = \E[\sum_{t=1}^T S_t] = \E[T] \cdot \E[S_1] = \eps \E[T] = 1 + \frac{1}{e} + O(\eps),
\end{equation}
    
where the last equality is due to \cref{eq:expected-T-ber-eps-example-risky}.

It follows from \cref{eq:E-of-Z-risky-one-plus-one-over-e} and \cref{eq:E-of-Z-risky-one-plus-p} that
$p = \frac{1}{e} + O(\eps)$, finishing the proof.

\end{proof}

%% file: more-related-work.tex
\subsection{Additional Related Work}
\label{sec:more-related}


\paragraph{Stochastic Knapsack with resource augmentation.} Some previous work showed that given $O(\eps)$ resource augmentation, there are $PTAS$ for approximating the optimal adaptive and non-adaptive policies for $\nrsk$ and $\rsk$.

\cite{bhalgat2011improved} give a $PTAS$ for computing adaptive and non-adaptive policies for $\nrsk$ and $\rsk$. 
\cite{bhalgat2011improved} introduce and use the notion of block-adaptive policies (policies deciding on blocks of items to insert and not just individual items). They show that (1) there exists a $1 + \eps$ adaptive policy that inserts at most $O(poly(\frac{1}{\eps}))$ blocks, and (2) enumerating all of the possible block-adaptive policies is possible in $\eps$-dependent polynomial time. However, their policy construction requires querying the state of the knapsack for each item in each inserted block, leading to a linear number of adaptive queries for the state of the knapsack. Also, their policies require a large ($O(\frac{1}{\eps^{15}})$) number of blocks, and more importantly, $O(\eps)$ resource augmentation. The same principle applies to the PTAS of \cite{li2013stochastic} for $\nrsk$ and of \cite{fu2018ptas} to $\rsk$, which both follow the same idea of searching a close to optimal block-adaptive policies \cite{bhalgat2011improved} (while their exact methods for searching the block-adaptive policy space differ). All of these PTAS results require both $O(\eps)$ resource augmentation, and knowing (thus querying) the state of the knapsack inside each block to decide on the next block to insert. These results are not fully polynomial and the dependence on $\eps$ implies a  large run time run time in practice. In this work, we focus on policies that (1) have no resource augmentation, (2) use only at most $\tilde{O}\prn*{\frac{1}{\eps}}$ adaptive queries, and (3) query only the current state of the knapsack rather than the entire history of realizations.

\paragraph{Other variants of Stochastic Knapsack:} Many other variants of the problem have been studied.
\cite{Levin2018} study the adaptivity gaps of different variants of the Stochastic Knapsack problem. The closest variant is the one they call the OE where in case of overflow in the knapsack capacity, the algorithm gets the value of the last (overflowing) item. 

\cite{gupta2011approximation}
show an $8$ approximate adaptive policy for the variant of the problem where cancellations (throwing an item away) and correlations (between item sizes and values) are allowed. This was later improved by
\cite{ma2014improvements} who provides $2 - \eps$ approximate (fully adaptive) policy by introducing additional assumptions (they assume that if the knapsack capacity is $T$, the processing times of each jobs are integers. Also, the algorithm run time is polynomial in $T$).

\cite{li2013stochastic} study several stochastic combinatorial problems, including the expected utility maximization problem, the stochastic knapsack problem and the stochastic bin packing problem. Besides their $1 + \eps$ approximation with $O(\eps)$ resource augmentation for $\nrsk$, they also provide a $2 + \eps$ approximation algorithm for stochastic knapsack with cancellations.

Another studied variant is the $p$-Chance Constrained Stochastic Knapsack, which is also known as the bounded overflow probability model.
The goal in this variant is to maximize the value under the constraint that the overflow probability is at most $p$. In this bounded overflow probability model, \cite{kleinberg1997allocating} give a $log(1/\gamma)$-
factor approximation algorithms when items sizes have
Bernoulli size distribution. \cite{goel1999stochastic} give
PTAS and QPTAS results for restricted classes of distributions, namely Bernoulli, Exponential and Poisson
distribution.
\citet{de2018boolean} provide approximation schemes (that do not require $1 + O(\eps)$ resource augmentation) for the cases where the item sizes are Bernoulli, k-supported or hyper-contractive random variables.
These versions are just like $\rsk$, except overflow probability $p$ is allowed (the algorithm gets the value as long as the probability for overflow is at most $p$). We note that it is not true that any policy for $\rsk$ if also a policy for $\eps$-chance constrained problem for any small $\eps$ --- as \cref{example:ber-eps-items} demonstrates, the optimal adaptive policy may have probability as high as $\frac{1}{e}$ of overflow. One can show that the optimal non-adaptive policy for the same instance also has a constant positive probability of overflow.

\cite{parthasarathy2020adaptive} study a more general setting where the utility function is submodular and not necessarily linear, and give a policy with an approximation guarantee of $\frac{1}{6}(1 - e^{-\frac{\beta}{4}})$ for a parameter $\beta \in (0,1]$.

%% file: proofs-tech-claims-section-1-n-gap-ub.tex
\section{Proofs for \cref{sec:one-to-n-gap}}\label{sec:proofs-one-to-n-gap}

\subsection{Proof of \cref{clm:T'-ge-quarter}}\label{sec:proof-of-clm-T'-ge-quarter}
\begin{proof}

Let $v := 0.24215$.

We replace the maximum in $o_5$ with a specific $t:=1-\dfrac{2(\alpha-p)}{1-p}$ value to obtain our lower bound. Let $T'$ denote the optimal value of the resulting optimization problem
\[
\min_{p,\alpha,\beta,\gamma}\ \max\Big\{(1-p)\alpha,\ (1-p)(\beta-\gamma),\ \gamma,\ \beta(1-\beta),\ h(\alpha,\beta,\gamma,p,t)\Big\},
\]
subject to
\[
p\in[0,\tfrac12],\quad \alpha\in[p,\tfrac12],\quad \beta\in[0.5,1],\quad \gamma\in[0,0.25],
\]
where $h$ is given by
\begin{align*}
    h(\alpha,\beta,\gamma,p,t)& =
(1-p)(\beta-\gamma) + \max\crl*{0, \ 1-p-\dfrac{\alpha-p}{1-t}} \prn*{\gamma-\dfrac{\beta(\beta-\alpha)}{t}}.
\end{align*}

Note that the case where does not change the objective as $o_2 = (1-p)(\beta - \gamma)$ choose $h$ to be $(1-p)(\beta - \gamma)$ in case where 

So $T \ge T'$, and so it is enough to show the lower bound on $T'$.

Write
\[
A:=(1-p)\alpha,\quad B:=(1-p)(\beta-\gamma),\quad C:=\gamma,\quad D:=\beta(1-\beta),\quad E:=h(\alpha,\beta,\gamma,p,t).
\]
We show that for \emph{every} feasible quadruple $(p,\alpha,\beta,\gamma)$ one has
\(
\max\{A,B,C,D,E\}\ge v.
\)

Define
\[
p_{\min}(\beta):=\frac{4\beta-2}{\,4\beta-1\,}\qquad(\beta\in[0.5,1]),
\]

\[
\beta_v:=\tfrac12+\sqrt{\tfrac14-v}\qquad\text{so that}\qquad D(\beta_v)=\beta_v(1-\beta_v)=v,
\]
and
\[
p_v:=\frac{1-\sqrt{1-4v}}{2}\qquad\text{so that}\qquad (1-p_v)p_v=v.
\]
We show $\max\{A,B,C,D,E\}\ge v$ by cases:

\emph{(a) Small $\beta$.} If $\beta\le \beta_v$, then $D(\beta)\ge D(\beta_v)=v$.

\emph{(b) Large $\beta$, small $p$.} If $\beta\ge \beta_v$ and $p<p_{\min}(\beta)$, then $B\ge \frac14\ge v$ as we now show:
\[
1-p>1-\frac{4\beta-2}{4\beta-1}=\frac{1}{4\beta-1}.
\]
Since $\gamma\le 0.25$,
\[
B=(1-p)(\beta-\gamma)\ \ge\ (1-p)(\beta-0.25)\ \ge\ \frac{\beta-0.25}{4\beta-1}=\frac14\ \ge\ v.
\]

\emph{(c) Large $\beta$, large $p$.} If $\beta\ge \beta_v$ and $p\ge p_v$, then
\[
A=(1-p)\alpha\ \ge\ (1-p)p\ \ge\ (1-p_v)p_v\ =\ v.
\]

\emph{(d) The remaining strip.} Suppose $\beta\ge \beta_v$ and $p\in[p_{\min}(\beta),\,p_v]$.

 Let
\[
u:=\beta(\beta-\alpha),\qquad
r:=\frac{u}{t}=\beta(\beta-\alpha)\,\frac{1-p}{\,1+p-2\alpha\,}.
\]

Differentiate
\[
r(\alpha)=\beta(\beta-\alpha)\frac{1-p}{1+p-2\alpha}
\quad\Longrightarrow\quad
\frac{dr}{d\alpha}=\beta\,\bigl(2\beta-(1+p)\bigr)\frac{1-p}{(1+p-2\alpha)^2}\le 0,
\]
since $2\beta\le 1+p$ holds in Case~II: $p \ge \frac{4\beta - 2}{4 \beta - 1}$ is equivalent to $p \cdot 4 \beta - p \ge 4\beta - 1$ or $2 \beta (2-2p) \le 1 - p$ and therefore $2 \beta \le \frac{1 - p}{2(1-p)} < 1 + p$ . Thus $r(\alpha)$ is decreasing in $\alpha$ and
\begin{equation}\label{eq:rmax}
r(\alpha)\ \le\ r(p)=\beta(\beta-p).
\end{equation}
Split Case~(d) according to $\gamma$:

\begin{itemize}
    \item \emph{Case d.1: $\gamma \ge \frac{\beta(\beta-\alpha)}{t} = r(\alpha)$.}  

    By the choice of $t=1-\dfrac{2(\alpha-p)}{1-p}$ one has
    \(
    1-p-\dfrac{\alpha-p}{1-t}=\dfrac{1-p}{2}
    \).
    
    \begin{align*}
        E & = h(\alpha,\beta,\gamma,p,t) = (1-p)\prn*{\beta - \gamma} + \prn*{\frac{1-p}{2}} \prn*{\gamma - \beta(\beta-\alpha)\frac{1-p}{1+p-2\alpha}} \\
        & = (1-p)\prn*{\beta - \frac{\gamma}{2} - \frac{r}{2}} \ge \prn*{1-p} \prn*{\beta - \frac{1}{8} - \frac{r}{2}} \numberthis \label{eq:E-lb-as-a-func-of-r},   
    \end{align*}
    where the last inequality holds as $\gamma\in[0,0.25]$.

    By \cref{eq:E-lb-as-a-func-of-r} and \cref{eq:rmax}:
    \begin{equation}\label{eq:H}
    T' \ge E\ \ge\ H_p(\beta):=(1-p)\Bigl(\beta-\tfrac18-\tfrac12\,\beta(\beta-p)\Bigr).
    \end{equation}
    Moreover, for $\beta\in[0.5,1]$,
    \[
    \frac{\partial H_p}{\partial\beta}=(1-p)\Bigl(1-\beta+\tfrac p2\Bigr)\ \ge\ 0,
    \]
    so $H_p(\beta)$ is \emph{increasing} in $\beta$.

    Hence, by \cref{eq:H}:
    \[
        E \ge (1-p)\prn*{\beta_v - \frac{1}{8} - \frac{1}{2} \beta_v (\beta_v - p)}.
    \]
    Let $a,b$ be:\[
    a := \beta_v - \frac{1}{8} - \frac{1}{2} \beta_v^2 \approx 0.290375, \qquad b := \frac{1}{2} \beta_v. \approx 0.2943\] We get:
    \begin{align*}
         E & \ge (1-p) \prn*{a + bp} = -b p^2 + (b-a) p + a.
    \end{align*}
    $g(p) = -b p^2 + (b-a) p + a$ is a concave quadratic in $p$ function. So on the interval $[p_{min}(\beta_v), p_v]$ it obtains its minimum at an endpoint. One endpoint is $p_{min}(\beta_v) :=  1 - 1/(4*\beta_v - 1) \approx 0.261666$ for which $g\prn*{p_{min}(\beta_v)} \approx 0.26166 \ge v$ and the other is $p_v$ for which $g(p_v) \approx 0.24217 \ge v$.
    Hence 
    $E \ge v$.

    \item \emph{Case d.2: $\gamma < \frac{\beta(\beta-\alpha)}{t} = r(\alpha)$.}
    
    By \cref{eq:rmax}, we get $\gamma < \beta(\beta -p )$, and thus \[
    B \ge (1-p)(\beta - \gamma ) \ge (1-p)\beta(1 - \beta + p). \numberthis \label{eq:B-bound}\]
    For any fixed $\beta$, let $G_{\beta}(p) := (1-p)\beta(1-\beta+p) = - \beta p^2 + \beta^2 p + \beta(1-\beta)$. $G_{\beta}(p)$ is a concave quadratic function in $p$. Hence its minimum on any interval is attained at an endpoint of that interval. Plugging in the left endpoint $p_{\min} = \frac{4\beta-2}{\,4\beta-1\,} = 1 - \frac{1}{4\beta - 1}$ yields 
    \[
    G_\beta(p_{min}(\beta)) = \frac{1}{4\beta - 1} \cdot \beta \cdot \prn*{1 - \beta + 1 - \frac{1}{4\beta - 1}} = \beta \frac{(2 - \beta)(4\beta - 1) - 1}{(4\beta - 1)^2} = \beta \frac{-4\beta^2 + 9 \beta - 3}{(4\beta-1)^2}.
    \]

    A direct derivative computation shows this function of $\beta$ is minimized at an end-point ($G_\beta (p_{min})$ has a single local maximum and no local minimum). The left endpoint has $\beta = \frac{1}{2}$ yields $G_\beta\prn*{\frac{1}{2}} = 0.25 \ge v$.
    Since $p_{min}(\beta) \le 0.5$ it follows from the definition of $p_{min}(\beta)$ that $\beta \le \frac{3}{4}$, and so we may use $\beta = \frac{3}{4}$ at the right endpoint.
    $G_\beta(p_{min})(\frac{3}{4}) = 0.28125 \ge v$.

    For the right endpoint we may chose $p_v$, but the proof works even for $p = 0.5 \ge p_v$. By plugging in $p = 0.5$ we get:
    \[
        G_{\beta}\prn*{\frac{1}{2}}  = \frac{1}{2} \beta\prn*{1-\beta + \frac{1}{2}} = -\frac{1}{2} \beta^2 + \frac{3}{4} \beta,
    \]
    which is decreasing in $\beta$, and thus is minimized at the right endpoint: $\beta = \frac{3}{4}$. Evaluating $G$ at $\beta = \frac{3}{4}$ again yields $0.28125 \ge v$.
    
    In both cases, by \cref{eq:B-bound} we get $B \ge v$.
\end{itemize}

\medskip
Combining (a)--(d), in all feasible cases $\max\{A,B,C,D,E\}\ge v$, whence the optimal value $T$ of the problem satisfies $T\ge v$.

\end{proof}

%% file: proofs-k-to-n-gap.tex
\section{Appendix for \cref{sec:k-to-n-gap}}\label{sec:missing-proofs-k-to-n-gap}

\input{remark-large-items}

\input{proof-gen-two-choices}
\input{proof-greedy-block-value}
\input{proof-approximation-ratio-greedy-small}

%% file: remark-large-items.tex
\subsection{Computing the Optimal Semi-Adaptive Policy for Large Items} 
\label{sec:remark-large-items}

In \cref{thm:semi-adaptive-large-items} we show that given an optimal adaptive policy for large items,  we can get a semi-adaptive policy with a good approximation ratio, yielding the required existence result. However, we do not state how to find the optimal adaptive policy. For finite support discrete distributions, one could simply enumerate all trees of depth $k = \tilde{O}\prn*{\frac{1}{\eps}}$ in polynomial time and return the best one: For such tree option, we may continue via a constant approximation non-adaptive policy. Together this yields a semi-adaptive policy. We then may approximate its expected value by running the corresponding policy via sampling (polynomially many times) from the item size distributions and compute the empirical mean. Each such expected value computation is a $1+\eps$ approximation to the real expected value due to the law of large numbers. Alternatively, as \cite{dean2008approximating} show in the proof of their Lemma $8.1$, it is possible to find a $1+\eps$ approximation to the optimal adaptive policy via a recursive computation in polynomial time. Their proof constructs an adaptive policy with a decision tree of height $\Omega(\frac{1}{\eps^2})$.  \cref{thm:semi-adaptive-large-items} essentially shows that only the $\tilde{O}\prn*{\frac{1}{\eps}}$ top levels of the optimal adaptive decision tree are needed, and so another viable strategy is to run this policy and ``cut'' the remaining decision tree after $\tilde{O}\prn*{\frac{1}{\eps}}$ steps.

%% file: proof-gen-two-choices.tex
\subsection{Proof of \cref{lem:gen-two-choices-partition}}\label{sec:proof-of-lem-gen-two-choices}

\begin{proof}
W.l.o.g., assume $\ADAPT(I) = 1$ (otherwise we can rescale the problem).
Let $\alpha \in [0,1]$ be the expected value the optimal adaptive policy gets from the items of $I_1$. So algorithm $A_1$ achieves an expected value of at least $\frac{\alpha}{c_1} \ADAPT$, and algorithm $A_2$ achieves an expected value of at least $\frac{1-\alpha}{c_2} \ADAPT$.

So the algorithm that partitions $I$ into $I_1$, $I_2$ and either runs $A_1$ on $I_1$ or $A_2$ on $I_2$ according to the option with the maximum expected value gets an expected value of at least \[
\max \crl*{ \frac{\alpha}{c_1} \ADAPT, \frac{1-\alpha}{c_2} \ADAPT}.\]

Given any fixed $c_1$, $c_2$ values, the two terms in the above $\max$ are a decreasing and increasing linear functions of $\alpha$ and thus their intersection point minimizes $\max{\frac{\alpha}{c_1}, \frac{1-\alpha}{c_2}}$. These two functions intersect where $\frac{\alpha}{c_1} =  \frac{1-\alpha}{c_2}$ which is equivalent to $\alpha = \frac{c_1}{c_1 + c_2}$. We deduce that the algorithm gets at least $\frac{1}{c_1 + c_2} \ADAPT$. For any other $\alpha$ the algorithm also gets at least $\frac{1}{c_1 + c_2} \ADAPT$.
\end{proof}

%% file: proof-greedy-block-value.tex
\subsection{Proof of \cref{lem:greedy-block-value}}\label{sec:proof-lem-greedy-block-value}

\begin{proof}
Let $\rho(J) :=  \frac{\sum_{i \in J} w_i}{\sum_{i\in J} \mu_i}$ be the average value-size density of $J$. Since the ratios are sorted, it must be that 
$\forall j' > j$: 
\begin{equation}\label{eq:avg-density}
    \frac{w_j'}{\mu_{j'}} \le \rho(J).
\end{equation}

Let $x$ be the optimal solution to $\Phi(t)$. 
So
\begin{equation} \label{eq:tmp-bound-suffix-items}
    \sum_{j' > j} x_{j'} w_{j'} \le \rho(J) \sum_{j' > j} x_{j'} \mu_{j'}.
\end{equation}

Since $x$ is a feasible solution:
$\sum_i \mu_i x_i \le t$, and therefore:
\[
    \sum_{j' > j} x_{j'} \mu_{j'} = \sum_{i \in [n]} x_{i} \mu_{i} - \sum_{i \in J} x_{i} \mu_{i} \le t - \sum_{i \in J} x_{i} \mu_{i}.
\]
Plugging this back in \cref{eq:tmp-bound-suffix-items} we get: 
\[
    \sum_{j' > j} x_{j'} w_{j'} \le \rho(J)\prn*{t - \sum_{i \in J} x_{i} \mu_{i}}.
\]
Thus, 
\begin{align*}
    \Phi(t) & = \sum_i x_i w_i = \sum_{i \in J} x_i w_i + \sum_{j' > j} x_{j'} w_{j'} \le \rho(J) \cdot t + \sum_{i \in J} x_i \prn*{w_i - \rho(J) \mu_i}. \numberthis \label{eq:Phi-t-bound}
\end{align*}

The optimal solution to $\Phi(t)$ is of the form $x_i = 1$ for $i < k$, $x_k \in [0,1)$, and $x_i = 0$ for $i > k$, for some $k \in [n]$.
If $j \ge k$ then $w(J) \ge \Phi(t)$.
Otherwise $j < k$, and so $x_i = 1$ for any $i \in J$. Hence, from \cref{eq:Phi-t-bound}:
\begin{align*}
    \Phi(t) \le \rho(J) \cdot t + \sum_{i \in J} w_i - \rho(J) \mu_i = \rho(J)\cdot t + \sum_{i \in J} w_i - \rho(J)\mu(J) =\rho(J)\cdot t = w(J) \frac{t}{\mu(J)},
\end{align*}
where the last two equalities are due to the definition of $\rho(J)$.

\end{proof}

%% file: proof-approximation-ratio-greedy-small.tex
\subsection{Proof of \cref{thm:alg-block-adaptive-gap-from-Phi-one}}\label{sec:proof-of-thm-alg-block-adaptive-gap-from-Phi-one}
\begin{proof}

Let $f_m(\alpha) := \prn*{\sum_{i=1}^m \alpha_i} \prod_{j=1}^m (1 - \alpha_j)$.

The maximizer of $f_m$ in $(0,1)^m$ is obtained for $\forall i \in [m]: \alpha_i = a$ for some $a \in (0,1)$. The reason is that $log(f_m) = log(\sum_{i=1}^m \alpha_i) + \sum_{i=1}^m log(1 - \alpha_i)$ is strictly concave on $(0,1)^m$ and thus the maximizer of $log(f_m)$, which is the same as the maximizer of $f_m$, is unique. If the coordinates of the maximum are not the same, we get a contradiction to the symmetry of $f_m$.
And so: $\max_{\alpha \in (0,1)^m f_m(x)} = \max_{a \in (0,1)} g(a)$ where $g(a) := m a (1 - a)^m$.
Taking the derivative of $g$ and comparing to $0$ we get that for $a = \frac{1}{m+1}$: $g'(a) = 0$. Since $g(0) = g(1) = 0$ and that $g(a) \ge 0$ for any $g \in (0,1)$ we get that the maximizer is indeed $a = \frac{1}{m+1}$.

And so, we choose $\alpha_i = \frac{1}{m+1}$ and get (via \cref{lem:bound-rand-block-adapt-by-f-Phi-of-1}):
\begin{align*}
    \E[A] & \ge \Phi(1) \cdot 
\prn*{f(\alpha) - \eps m} = \Phi(1) \cdot \prn*{\frac{m}{m+1} \prn*{1 - \frac{1}{m+1}}^m - \eps m} \\
    & = \Phi(1) \prn*{\prn*{\frac{m}{m+1}}^{m+1} - \eps m}  = \Phi(1) \prn*{\prn*{\frac{k+1}{k+2}}^{k+2} - \eps(k+1)}.
\end{align*}
By applying \cref{prop:adapt_lp_upper_bound} we get the required.

\end{proof}

%% file: proofs-zero-to-one-gap.tex
\section{Appendix for \cref{sec:zero-to-one-gap}}\label{sec:proofs-zero-to-one-gap}

\subsection{Proof of \cref{lem:risky-F2-size-distributions}}\label{sec:proof-risky-f2-size-dist}

\begin{proof}
We show the proof for $H_2$, the proof for $H_2^T$ is identical: any tree constraint makes sure one can only choose root-to-leaf path, but so are the size distribution choices of $H_2'$.
Let any instance $I \in \Hc_2$. Then $\sad$ gets its value from taking item $0$ and then some other item $i$. First, we show the first property. That is, we show a transformation from instance $I$ to instance $I'$ where we replace each item with a set of deterministic items (we say an item is deterministic if it has a deterministic size distribution).

We assume $S_0$ follows a discrete finite support distribution. Let $C$ be the support of the distribution of $S_0$.

For any $c_i \in C$, let $i$ be the next item $\sad$ inserts when the realization $S_0 = c_i$ occurs. w.l.o.g $\sad$ never decides to stop inserting, otherwise we can add an item with zero size and value where $\sad$ can always choose this item instead of stopping.

We initialize $I'$ to contain item $0$, and for any $j \in [n], c \in C$ we add a new item, $j,c$, to $I'$ with the following size and value:

The value of the item is $v_{j,c} = \Pr(S_j + c \le 1) v_j - \Pr(S_j + c > 1)$. The size of the new item $j,c$ is $S_{j,c} = 1 - c$.

So $I'$ has $m+1$ items, where $m := n |C|$, and in instance $I'$ all items except for item $0$ have deterministic sizes.

We show that $\ADAPT(I') \ge \ADAPT(I)$:
\begin{align*}
    \ADAPT(I) &= (1-p_0) \prn*{1 + \sum_{c_i \in C} \Pr(S_0 = c_i) \brk*{\Pr\prn*{S_{i} + c_i \le 1} v_i - \Pr\prn*{S_i + c_i > 1}}} \\
    & = (1-p_0) \prn*{1 + \sum_{c_i \in C} \Pr(S_0 = c_i) \cdot v_{i,c_i}} \\
    & \le \ADAPT(I'), \numberthis \label{ineq:adapt_le_adapt_tag}
\end{align*}

where the second inequality is due to the definition of $v_{i,c_i}$ and the last inequality is since in $I'$ after taking item $0$ and seeing $S_0 = c_i$ the adaptive optimum may choose to take item ${i,c_i}$ which fits the knapsack w.p. $1$.

The intuition is that for every realization of $S_0$ to size $c$ in instance $I$, there is a choice of next item to insert in $I'$ that yields at least the same value and always fits the knapsack.

We will show that $\ALG(I') \le \ALG(I)$.

For any $j \in [n], c \in C$, Let $\ALG_j$, $\ALG'_{j,c}$ denote the value the best non-adaptive policy gets from choosing item $0$ and then item $j$ in instance $I$ and the value it gets from choosing item $0$ and then item $j,c$ in instance $I'$, respectively.

We now show that for any new such option for $\salg$ in $I'$ yields less value than the original option for $\salg$ in $I$:

\begin{clm}\label{clm:risky-alg-tag-j-c-le-alg-j}
    For any $j \in [n]$, $c \in C$:
    $\ALG'_{j,c} \le \ALG_j$.
\end{clm}
\begin{proof}
\begin{align*}
    & \ALG'_{j,c} = (1-p_0)\prn*{1 + \sum_{c_i \in C} \Pr(S_0 = c_i) \brk*{v_{j,c} \cdot \Pr\prn*{S_{j,c} + c_i \le 1} - \Pr\prn*{S_{j,c} + c_i > 1}}} \\
    & = (1-p_0)\prn*{1 + \sum_{c_i \in C} \Pr(S_0 = c_i) \brk*{v_{j,c} \cdot \Pr\prn*{S_{j,c} + c_i \le 1} - \Pr\prn*{S_{j,c} + c_i > 1}}}.
\end{align*}

On the other hand:
\[
    \ALG_j = (1-p_0)\prn*{1 + \sum_{c_i \in C} \Pr(S_0 = c_i) \brk*{v_j \Pr(S_j + c_i \le 1) - \Pr(S_j + c_i > 1)}}.
\]

By the above we get that $\ALG'{j,c} \le \ALG_j$ if and only if:
\begin{align*}
    \sum_{c_i \in C} & \Pr(S_0 = c_i) \brk*{v_j \Pr(S_j + c_i \le 1) - \Pr(S_j + c_i > 1)  - v_{j,c} \Pr(S_{j,c} + c_i \le 1) + \Pr(S_{j,c} + c_i > 1)} \\ &\ge 0,
\end{align*}

which is equivalent to:
\begin{equation}\label{eq:need-to-show-tmp-1}
    v_j \Pr(S_0 + S_j \le 1) - \Pr(S_0 + S_j > 1) - v_{j,c} \Pr(S_0 \le c) + \Pr(S_0 > c) \ge 0.
\end{equation}

Since:
\[
    \Pr(S_0 + S_j \le 1) = \Pr(S_0 + S_j \le 1 \mid S_0 \le c)\Pr(S_0 \le c) \ge \Pr(S_j + c \le 1) \Pr(S_0 \le c),
\]
and
\[
    \Pr(S_0 + S_j > 1) = \Pr(S_0 + S_j > 1 \mid S_0 \le c)\Pr(S_0 \le c) \le \Pr(S_j + c > 1) \Pr(S_0 \le c),
\]
then:
\begin{align*}
    v_j \Pr(S_0 + S_j \le 1) - \Pr(S_0 + S_j > 1) & \ge v_j \Pr(S_j + c \le 1) \Pr(S_0 \le c) - \Pr(S_j + c > 1)\Pr(S_0 \le c) \\
    & = \Pr(S_0 \le c)\prn*{v_j \Pr(S_j + c \le 1) - \Pr(S_j + c > 1)} \\
    & = \Pr(S_0 \le c) v_{j,c},
\end{align*}

where the last equality is due to the definition of $v_{j,c}$.

By \cref{eq:need-to-show-tmp-1} $\ALG'_{j,c} \le \ALG_j$ iff $\Pr(S_0 > c) \ge 0$ which always holds. 
\end{proof}

It follows directly from \cref{clm:risky-alg-tag-j-c-le-alg-j} that: 
\begin{equation}\label{eq:risky-alg-tag-worse-alg}
    \max_{i,c} \ALG'_{i,c} \le \max_{i} \ALG_i.
\end{equation}

We can remove any item $i \in [m]$ that $\sad$ doesn't choose in $I'$ as it keeps $\ADAPT(I')$ the same but decreases $\ALG(I')$. Let $n'$ be the number of remaining items.

For any $i \neq 0$ the option of choosing item $i$ and then item $0$ yields the same value for as choosing item $0$ and then item $i$ (due to order invariance). \cref{eq:risky-alg-tag-worse-alg} means that any such option yields less value in $I'$ than in $I$. 
Let us rename the items of $I'$ by their value such that $v_1 \ge ... \ge v_{n'}$. We can assume w.l.o.g. that $S_i > S_j$ for any $i < j$ as otherwise no policy would ever choose item $j$ over item $i$.

Let $\eps > 0$ s.t. $\eps \ll \frac{1}{n'}$. By choosing $S_0 = 0.5 - (n' - i)\eps$ w.p. $p_i$ and $S_{i,c} = 0.5 + (n'-i)\eps$ for any $i \in [n']$, we keep the structure of the adaptive tree of $I'$ ( $\ADAPT(I')$ stays the same) while ruling out any option of choosing two items from $[n']$ as inserting any two such items will cause overflow ($\ALG(I')$ may only decrease).

We conclude: for any non-adaptive policy for $I'$ there's a non-adaptive policy for $I$ that yields at least the same value: $\ALG(I') \le \ALG(I)$.

Together with \cref{ineq:adapt_le_adapt_tag}:
we get: $G(\Hc'_2) \ge G(\Hc_2)$. since $\Hc'_2 \subseteq \Hc_2$: $G(\Hc'_2) = G(\Hc_2)$.
\end{proof}


\subsection{Proof of \cref{lem:risky-all-alg-options-yield-the-same}}\label{sec:proof-risky-all-alg-options-yield-the-same}

\begin{proof}
We show the following claim and then show that the lemma follows as a direct result.

\begin{clm}\label{eq:clm-either-smaller-or-same-size-I-tag-where-all-options-are-the-same-risky}
Let $I \in \Hc'_2$. Either there's a smaller instance $I' \in \Hc'_2$ ($|I'| < |I|$ - $I'$ has less items) such that $G(I') \ge G(I)$, or there's an instance $I' \in \Hc''_2$ s.t. $I' = I$ and $G(I') = G(I)$.
\end{clm}

\begin{proof}

Let $I \in \Hc'_2$. Then, by \cref{lem:risky-F2-size-distributions}, taking item $0$ and then item $i$ yields an expected value of \[
Pr(\text{item }0 \text{ fits})\cdot \Pr(\text{item } i \text{ fits given that item }0\text{ fits})\cdot(1 + v_i) = (1-p_0)(\sum_{j=1}^i p_j)(1 + v_i).
\]
We show the claim via induction.

For the base: consider an instance of size 2. So the policy only has two options: either take item $0$ and then item $1$, or take item $1$  (order-invariance implies that taking item $1$ and then item $0$ yield the same value).
If taking item $1$ yields more value for $\salg$: $v_1 > (1-p_0)p_1 (1 + v_1)$, then we could decrease $p_0$ and increase $p_1$ until the two options are equal (from intermediate value theorem). $\ADAPT$ only increases while $\ALG$ stays the same. If $v_1 < (1 - p_0) p_1 (1 + v_1)$ then we increase $p_0$ until taking item $1$ and stopping yields the same value.\\

Induction step: assume the claim is true for any smaller instance of size $n' < n$, and consider an instance of size $n$. Like in the base case, we can transform the instance (by modifying $p_0, p_1$) such that $v_1 = (1-p_0)p_1 (1 + v_1)$.\\

Case (1): Assume taking item $1$ yields the (strictly) best expected value for $\salg$: $v_1 > (1-p_0)p_1(1 + v_1)$: just like in the base case, we could always decrease $p_0$ and increase $p_1$ (while decreasing the probability mass from the rest of the options of $S_0$) until the two terms are equal. As before intermediate value theorem implies there are such $p_0$,$p_1$ values such that the two terms are equal. In the new resulting instance, $I'_1$, $ADPAT$ increases (as we shift probability from lower valued items to a bigger valued item) and $\ALG$ does not change.

Next, we have a series of transformations from $I'_i$ to $I'_{i+1}$ for any $i \in [n-1]$:
We increase the value $v_{i+1}$ until either (1a) it equals $v_i$ or (1b) choosing item $0$ and then item $i+1$ yields the same expected value as choosing (only) item $1$. That is: $v_1 = (1 - p_0)\sum_{j=1}^i p_j(1 +  v_{i+1})$.

In Case (1a): we could simply drop item $i+1$ from $I'_{i+1}$ since any policy (adaptive or not) will always choose item $i$ over item $i+1$ as it has smaller size and larger value. Also, as a result, $\ADAPT$ may only increase while $\ALG$ does not change and thus we get an instance with $n-1$ items. Via the induction hypothesis we get the desired.

In case (1b), via an additional simple induction, we get that all of the options of $\salg$ (of choosing item $1$ only, or choosing item $0$ and then item $j \in \crl*{1,\ldots,i+1}$) yield the same expected value. The process stops when either an instance of $n-1$ items with the same adaptivity gap is obtained (and can therefore use the induction hypothesis) or when we reach $I'_n$ and get that all options yield the same expected value, as required.\\

Case (2): Taking item $0$ and then item $1$ is the (strictly) best option. In this case, we can decrease $p_1$ and increase $p_0$ by the same amount until this is no longer the case. From intermediate value theorem there are such $p_0,p_1$ values.
The adaptivity gap may only increase due to the change: Let $\ADAPT'$,$\ALG'$ be the expected value of $\sad$ and $\salg$ in the new resulting instance. The change in $p_0$ translates into a multiplicative change in $\ADAPT$, $\ALG$ that affects both of them in the same way and so the ratio between the two is not affected by a change in $p_0$. The change in $p_1$. Let $\Delta = (1-p_0)(p_1 - p'_1)(v_1+1)$. Then:
\begin{align*}
    \frac{\ADAPT'}{\ALG'} &= \frac{(1-p'_0)(1 + p'_1 v_1 + \sum_{i=2}^n p'_i v_i)}{(1-p'_0)p'_1(1 + v_1)} \ge \frac{(1-p_0)\prn*{(1 + \sum_{i=1}^n p_i v_i) - v_1(p_1 - p'_1)}}{(1-p_0)\prn*{p_1(1 + v_1) - (p_1 - p'_1)(1+v_1)}} \\
    & \ge \frac{\ADAPT - \Delta}{\ALG - \Delta} \ge \frac{\ADAPT}{\ALG}.
\end{align*}

The adaptivity gap has only increased: $G(I') \ge G(I)$. After this change, the option of taking only item $1$ and the option of taking item $0$ and then item $1$ both yield the same expected value, which is the value of $\ALG$. We thus we may continue as in case (1).\\

Case (3): If the (strictly) best option of $\salg$ is to take item $0$ and then some item $i > 1$, then $v_1 < (1-p_0)( \sum_{j=1}^i p_j)(1 + v_i)$. We could decrease $p_i$ and increase $p_1$ by the same amount until finally either (a) $p_i = 0$ (and therefore no optimal policy, adaptive or not, will choose item $i$) or (b) via intermediate value theorem we get that the option of choosing item $0$ and then item $1$ yields the same value as the option of choosing item $0$ and then item $i$.

Case (3a): In the first case, we get an instance with $n-1$ items where $\ADAPT$ may only increase but $\ALG$ may only decrease, yielding the desired via the induction hypothesis.

Case (3b): In the second case, the two options yield the same expected value (while the adaptivity gap may only increase). We continue this way for all $i > 1$ so there is no $i > 1$ such that taking it after item $0$ yields the strictly best option of $\salg$. In the resulting instance: either taking only item $1$ or taking item $0$ and then taking item $1$ is the best option.
If taking only item $1$ is strictly the better option, and we can continue according to case (1) or case (2).

\end{proof}

Given any instance $I \in \Hc'_2$ we can apply \cref{eq:clm-either-smaller-or-same-size-I-tag-where-all-options-are-the-same-risky} at most $|I|$ times to show the existence of an instance $I' \in \Hc''_2$ s.t. $G(I') = G(I)$, as required.

\end{proof}

\subsection{Proof of \cref{lem:terminal-choice}}\label{sec:proof-of-items-are-terminals}

\begin{proof}
We start with the first property:
Let $I$ be an instance in $\Hc'_2$. First, we can obtain an instance $I'$ by replacing each item $i \in [n]$ with a new item such that the value of the new item $v'_i = a \cdot v_i$ and its size $S'_i$ is $S_i$ w.p. $\frac{1}{a}$ and $2$ with probability $1 - \frac{1}{a}$. In expectation, every policy (adaptive or not) gets the same value for item $i$ and for item $i'$. The difference is that now if a policy picks item $i'$ first, it would gain almost nothing ($\frac{1}{a} \ll 1$) by choosing any other item later on. $\ADAPT$ may only increase due to this change but $\ALG$ can only decrease (all options yield the same expected value except for ones starting with item $i \neq 0$ and then picking some other item).
So indeed $G(I') \ge G(I) - O(\frac{1}{a})$. 

The new items are still not terminals since $S_i$ might not be deterministic (also note that a small counter example may show that replacing $S_i$ with $\E[S_i]$ does not work).

We therefore show how to gain an instance $I''$ from $I'$ by replacing each item $i'$ with a set of items such that each item is a terminal item.
We assume $F_0$ is discrete and has finite support (otherwise we could consider a discretization of it).

Let $C = \crl*{c \in supp\prn*{F_0}}$. Consider the instance obtained from $I'$ be replacing each item $i' \neq 0 \in I'$ with the following $|C|$ items:

For any $i' \in I$, $c \in C$ we create an item $i,c$ in $I''$.
        
The size distribution of the item is $S_{i,c} = \begin{cases} 
    1 - c & \text{with probability } \frac{1}{a} \\ 
    2 & \text{with probability } 1 - \frac{1}{a}
\end{cases}$,

and the value of the item is $a \cdot v_i \ \Pr(S_i + c \le 1)$.

In $I''$, $\sad$ may get at least the same value as it did in $I'$: After $\sad$ chooses item $0$ it sees the realization $c$ of $S_0$ and then for every choice of item $i'$ it would have chosen in $I'$, in $I''$ it will choose item $i,c$ and it will yield the same expected value (the reason is that for any $c' < c$ choosing $i,c'$ will result in a knapsack overflow, and for any $c' > c$ choosing $i,c'$ yields smaller value).

However, we will show that $\ALG(I'') \le \ALG(I')$. 

Indeed for any $i \in [n], c \in C$, Let $\ALG'_i$, $\ALG''_{i,c}$ denote the value the best non-adaptive policy gets from choosing item $0$ and then item $i$ in instance $I'$ and the value it gets from choosing item $0$ and then item $i,c$ in instance $I''$, respectively.

So $\ALG'_i = (1-p_0)\prn*{1 + \sum_{c \in C} \Pr(S_0 = c) \Pr (S_i + c \le 1) v_i}$. But then for any $c' \in supp\prn*{F_0}$: 
\begin{align*}
    & \ALG''_{i,c'} = (1-p_0)\prn*{1 + \sum_{c \in C} \Pr(S_0 = c) \Pr(S_{i,c'} + c \le 1) v_i \Pr(S_i + c' \le 1)} \\
    & = (1-p_0)\prn*{1 + \sum_{c \in C} \Pr(S_0 = c) \Pr(1 - c' + c \le 1) v_i \Pr(S_i + c' \le 1)} \\
    & = (1-p_0)\prn*{1 + \sum_{c \le c'} \Pr(S_0 = c) v_i \Pr(S_i \le 1 - c')} \\ 
    & \le (1-p_0)\prn*{1 + \sum_{c \le c'} \Pr(S_0 = c) v_i \Pr(S_i \le 1 - c)} \\
    & \le \ALG'_i.
\end{align*}

We deduce: $\ALG(I'') = \max_{i,c} \ALG''_{i,c} \le \max_{i} \ALG'_i = \ALG(I')$.

We've shown that $\ADAPT$ may not decrease while while $\ALG$ may not increase and thus $G(I'') \ge G(I') \ge G(I) - O(\frac{1}{a})$.

This concludes the proof of the first property.

The second property holds as removing probability‑zero leaves cannot increase $\ALG$ nor decrease $\ADAPT$. 

We now show the third property: Since $v_1 \ge \ldots \ge v_n$, since $\sad$ gets almost all its value from two items and since $\sad$ starts with item $0$, it must be that $p_i$ is the probability that item $i$ fits the knapsack after taking item $0$ (since the optimal adaptive policy for the last item is to take the largest valued item that fits). Items $[n]$ are terminals (\cref{lem:terminal-choice}), and so it must be that $s_1 \ge \ldots \ge s_n$ since otherwise the optimal policy would never choose them (contradicting the first property). Hence we may assume that $\forall i \in [n]:$ $S_0 = 1 - s_i$.
\end{proof}

\subsection{Proof of \cref{lem:nrsk-all_alg_options_yield_the_same_val}}\label{sec:proof_nrsk_all_alg_options_yield_the_same_val}
\begin{proof}
We show claim via induction:

Throughout the proof when we refer to $v_i$ as its effective value rather than value 
for any $i \in [n]$ (otherwise we could use $w_i = \frac{1}{a} v_i$ instead).
The base case is for $n=1$.
In this case there are two items: item $0$ and $1$. Consider the case where taking item $1$ yields higher value than taking item $0$ and then $1$. If $v_1 > (1-p_0)(1 +  p_1v_1)$ we could always decrease $p_0$ until either the two sides are the same or until $p_0 = 0$ (via intermediate value theorem). In the second option $v_1 > 1 + p_1 v_1$. We could increase $p_1$ such that $v_1 = 1 + p_1 v_1$. There is always such a $p_1$ (from intermediate value theorem) since $f(p_1) = v_1 - (1 + p_1 v_1)$ is positive for the initial $p_1$ but is negative for $p_1 = 1$.
Hence, in the new instance we obtain, the two options of $\salg$ yield the same expected value.
Similarly, if taking item $0$ and then $1$ yields higher value we could increase $p_0$ until the two options yield the same expected value.\\

For the induction step, assume the claim is true for any $n' < n$. Consider an instance of size $n$. We divide into cases, based on which option yields the strictly best (maximum value) option for $\ALG$ (if there is no such option there is nothing to prove). \\

Case (1): Assume taking item $1$ yields the (strictly) best expected value for $\salg$: $v_1 > (1-p_0)(1 + p_1 v_1)$: just like in the base case, we could always decrease $p_0$ and increase $p_1$ (while decreasing the probability mass from the rest of the options of $S_0$) until the two terms are equal. As before intermediate value theorem implies there are such $p_0$,$p_1$ values such that the two terms are equal. In the new resulting instance, $I'_1$, $ADPAT$ increases (as we shift probability from lower valued items to a bigger valued item) and $\ALG$ may only decrease.

Next, we have a series of transformations from $I'_i$ to $I'_{i+1}$ for any $i \in [n-1]$:
We increase the value $v_{i+1}$ until either (1a) it equals $v_i$ or (1b) choosing item $0$ and then item $i+1$ yields the same expected value as choosing item $1$. That is: $v_1 = (1 - p_0)(1 + \sum_{j=1}^i p_j v_{i+1})$.

In Case (1a): we could simply drop item $i+1$ from $I'_{i+1}$ since any policy (adaptive or not) will always choose item $i$ over item $i+1$ as it has smaller size and larger value. Also, as a result, $\ADAPT$ may only increase while $\ALG$ does not change and thus we get an instance with $n-1$ items. Via the induction hypothesis we get the desired.

In case (1b), via an additional simple induction, we get that all of the options of $\salg$ (of choosing item $1$ only, or choosing item $0$ and then item $j \in \crl*{1,\ldots,i+1}$) yield the same expected value. The process stops when either an instance of $n-1$ items with the same adaptivity gap is obtained (and can therefore use the induction hypothesis) or when we reach $I'_n$ and get that all options yield the same expected value, as required.\\

Case (2): Taking item $0$ and then item $1$ is the (strictly) best option. In this case, we can decrease $p_1$ and increase $p_0$ by the same amount until this is no longer the case. From intermediate value theorem there is such a $p_1$.
The adaptivity gap may only increase due to the change: Let $\ADAPT'$,$\ALG'$ be the expected value of $\sad$ and $\salg$ in the new resulting instance. Let $\Delta = (1-p_0)(p_1 - p'_1)v_1$ be the difference between $\ADAPT$ and $\ADAPT'$, and similarly between $\ALG$ and $\ALG'$.  Then 
$\frac{\ADAPT'}{\ALG'} = \frac{\ADAPT - \Delta}{\ALG - \Delta} \ge \frac{\ADAPT}{\ALG}$; the adaptivity gap has only increased. After the change the option of taking only item $1$ and the option of taking item $0$ and then item $1$ both yield the same expected value, which is the value of $\ALG$. We thus may continue as in case (1).\\

Case (3): If the (strictly) best option of $\salg$ is to take item $0$ and then some item $i > 1$, then $v_1 < (1-p_0)(1 + \sum_{j=1}^i p_j v_i)$. We could decrease $p_i$ and increase $p_1$ by the same amount until finally either (a) $p_i = 0$ (and therefore no optimal policy, adaptive or not, will choose item $i$) or (b) via intermediate value theorem we get that the option of choosing item $0$ and then item $1$ yields the same value as the option of choosing item $0$ and then item $i$.

Case (3a): In the first case, we get an instance with $n-1$ items where $\ADAPT$ may only increase but $\ALG$ may only decrease, yielding the desired via the induction hypothesis.

Case (3b): In the second case, the two options yield the same expected value (while the adaptivity gap may only increase). We continue this way for all $i > 1$ so there is no $i > 1$ such that taking it after item $0$ yields the strictly best option of $\salg$. In the resulting instance: either taking only item $1$ or taking item $0$ and then taking item $1$ is the best option.
If taking only item $1$ is strictly the better option, and we can continue according to case (1) or case (2).

\end{proof}

%% file: single-adaptive-choice-nrsk.tex
\subsection{Stronger Lower Bound on the Adaptivity Gap}\label{sec:nrsk-H2-gap}

In this section we show that $G(\Hc_2) = 1 + \frac{1}{e}$:

\begin{prop}\label{thm:single-adaptive-choice-gap-non-risky}
    In $\nrsk$ the adaptivity gap restricted to the family of instances $\Hc_2$ is: 
    $G(\Hc_2) = 1 + \frac{1}{e}$
\end{prop}
This, of course, implies a lower bound for general instances:
\begin{cor}\label{cor:nrsk_lb_of_1_plus_1_over_e}
    In $\nrsk$, the adaptivity gap is at least $1 + \frac{1}{e}$.
\end{cor}

We start by defining a simple but important item type: terminal items.
\begin{definition}(Terminal item)\label{def:terminal-item}
    We say item $i$ is a terminal item if w.p. at least $1 - \frac{1}{a}$ it has a size of $2$, for a very large $a \gg n$, and with probability $\frac{1}{a}$ it has a fixed size $s_i \in [0,1]$.
    An item that is not terminal is referred to as a non-terminal item.
\end{definition}

Intuitively, once a terminal item is chosen, subsequent decisions contribute at most $O(1/a) \ll 1$ additional value. The effective value of any terminal item with value $v_i$ is $w_i = \frac{v_i}{a}$.

We show that we can reduce to the family of instances where all items are terminals.

Unlike the case of $\rsk$, in $\nrsk$ we use terminal items to eliminate the options for $\salg$ of choosing item $i \in [n]$ first and then some other item.

\paragraph{Step~1: Forcing simple size distributions on the items.}

\begin{lem}\label{lem:terminal-choice}
    Let $\Hc'_2 \subseteq \Hc_2$ be the set of instances of $\Hc_2$ with items $0,1,\ldots,n$. For any item $i \in [n]$:
    \begin{enumerate}
        \item Item $i$ is a terminal item.
        \item $p_i \neq 0,1$.
        \item $p_i = \Pr(S_0 = 1 - s_i)$.
    \end{enumerate}
    Then $G(\Hc'_2) = G(\Hc_2) - O(\frac{1}{a})$.
\end{lem}
We defer the proof of \cref{lem:terminal-choice} to \cref{sec:proof-of-items-are-terminals}.
The idea in the proof of the first property is to fix an instance $I \in \Hc_2$ and replace any item $i \in [n]$ with an ``almost'' terminal item such that the effective value remains the same, and the size is the original size w.p. $\frac{1}{a}$ or overflows the knapsack (size $2$) otherwise. The second transformation is to replace each such new item with $|C|$ terminal items where $C$ is the support of the size distribution of item $0$. In the new resulting instance, all of the items except of item $0$ are terminals, and the adaptivity gap may only increase (up to $\frac{1}{a} \ll 1$).

The second property holds as removing probability‑zero leaves or merging duplicate sub‑trees cannot increase $\ALG$ nor decrease $\ADAPT$, and the idea in the third property proof is to show that ordering the $s_i$ and $v_i$ non‑increasingly preserves optimality because terminal items with larger value (and hence larger $s_i$) are always preferable.

\paragraph{Step~2: Equalizing all optimal non‑adaptive branches.}

\begin{lem}\label{lem:nrsk-all_alg_options_yield_the_same_val}
Let $\Hc''_2$ consist of those instances in $\Hc'_2$ for which every potentially optimal non‑adaptive strategy yields the \emph{same} expected value, that is: $\ALG = w_1 = (1-p_0)(1 + \sum_{j=1}^i p_j w_i)$ for any $i \in [n]$.  Then for any $I\in \Hc'_2$ there exists either
\begin{enumerate}[label=(\alph*)]
    \item a smaller instance $I'\in \Hc'_2$ ($|I'| < |I|$) with $G(I')\ge G(I)$, or
    \item an instance $I'\in \Hc''_2$ with $G(I')=G(I)$.
\end{enumerate}
\end{lem}
The proof is almost identical to the one given to $\rsk$ and follows the same probability mass and value redistribution argument.
For completeness, we provide the proof of \cref{lem:nrsk-all_alg_options_yield_the_same_val} in \cref{sec:proof_nrsk_all_alg_options_yield_the_same_val}.\\

Putting the pieces together gives:
\begin{cor}\label{cor:H_2-F-equivalence}
    Let $a \gg n$.
    Let $w_1 \ge \ldots \ge w_n$, and $s_1 \ge \ldots \ge s_n \in [0,1]$,
    $G(\Hc''_2) = G(\Hc_2) - O(\frac{1}{a})$ where $\Hc''_2$ is the family of instances of the form of \cref{tab:example_I2} s.t.
    $w_1 = (1-p_0)(1 + \sum_{j=1}^i p_j w_i)$ for any $i \in [n]$.
\end{cor}

\begin{table}[h]
    \centering
    \begin{tabular}{|c|c|c|}
        \hline
        \textbf{ID} & \textbf{Value} & \textbf{Size} \\
        \hline
        0 & $v_0 = 1$ & 
        \begin{tabular}[c]{@{}l@{}}
            $S_0 = 1 - s_i$ with probability $p_i$ for each $i \in [n]$ \\
            $S_0 = 2$ with probability $p_0$
        \end{tabular} \\
        \hline
        $i \in [n]$ & $v_i = a \cdot w_i$ & 
        $S_i = \begin{cases} 
            s_i & \text{with probability } \frac{1}{a} \\ 
            2 & \text{with probability } 1 - \frac{1}{a} 
        \end{cases}$ \\
        \hline
    \end{tabular}
    \caption{$\Hc_2$ worst-case adaptivity gap instance for $\nrsk$}
    \label{tab:example_I2}
\end{table}

\paragraph{Step~3: Finding the worst instance via optimization.}\label{subsec:single-adaptive-worst-case}

Due to \cref{cor:H_2-F-equivalence} we can analyze the adaptivity gap of $H_2$ by analyzing the adaptivity gap more structured $\Hc''_2$, as described in \cref{tab:example_I2}.

Let $q = 1 - p_0$ be the probability that the first item does not overflow.

We assumed $\sad$ picks item $0$ and then picks item $i$ w.p. $p_i$. and so:
\begin{equation}
    \ADAPT = (1-p_0)\prn*{v_0 + \sum_{i \in [n]} p_i \frac{1}{a} v_i} = q\prn*{1 + \sum_{i \in [n]} p_i w_i}
\end{equation}

On the other hand, from \cref{cor:H_2-F-equivalence}:
$\ALG = w_1 = q(1 + p_1 w_1)$, and also $\ALG = q(1 + p_1 w_1) = q(1 + \sum_{j=1}^i p_j w_i)$ which implies $w_i = \frac{p_1 w_1}{\sum_{j=1}^i p_j}$.

Hence the adaptivity gap is:
\begin{align*}
    \frac{\ADAPT}{\ALG} &= \frac{q\prn*{1 + \sum_{i \in [n]} p_i w_i}}{\ALG} = \frac{q(1 + p_1 w_1)}{\ALG} + \frac{q \sum_{i  = 2}^n p_i w_i}{\ALG} = 1 + \frac{\sum_{i  = 2}^n p_i \frac{p_1 w_1}{\sum_{j=1}^i p_j}}{w_1} \\
    & = 1 + p_1 \sum_{i=2}^n \frac{p_i}{\sum_{j=1}^i p_j}. \numberthis \label{eq:G-star-2}
\end{align*}


From \cref{eq:G-star-2}, \cref{lem:tech-darboux-sum} we deduce that:
\[
    G(I_2) = \max_{p_1 \in (0,1)} - p_1 \ln(p_1).
\]

A simple analysis of $f(p_1) := -p_1 \ln(p_1)$ shows that $\frac{1}{e}$ is the tight upper bound for $G(I_2)$.

We hence deduce the exact adaptivity gap of $\Hc''_2$ and hence (\cref{cor:H_2-F-equivalence}) the exact adaptivity gap of $\Hc_2$, showing \cref{thm:single-adaptive-choice-gap-non-risky}.

%% file: Bernoulli_zero_one_ub_2_nrsk.tex
\subsection{Adaptivity Gap for $\epsnoisy$ Distributions}
\label{sec:eps-ber-nrsk-ub}

In this section we show that for $\epsnoisy$ item size distributions the adaptivity gap is at most $2$ (rather than the general known gap upper bound of $4$).

\begin{prop}\label{thm:nrsk-noisy-bernoulli-ub-2}
    The adaptivity gap  for the $\nrsk$ problem of any instance $I$ consisting of  $\epsnoisy$ items is at most $2$: $G(I) \le 2$.
\end{prop}

While our method is similar to the one used for the lower bound of \cref{sec:eps-ber-lower-bound-rsk} (analyzing the convergence of a recursive optimization problem), there are a few differences. The first is that we reduce to using terminal items rather than only Bernoulli type items. The second is that $\sad$ does not have a positive probability of stopping in each step (unlike in $\rsk$). The third difference is that even though the adaptivity gap program converges to the same value, we get a different recursive optimization problem.

\paragraph{Step 1: Simplify the adaptive decision tree.}
The goal of the first step is to achieve a simpler form of the decision tree that will be easier to analyze.

\begin{lem}\label{lem:non-risky-noisy-bernoulli-instance-decision-tree}
    There exists instance $I'$ for the $\nrsk$ problem such that $G(I) \le G(I')$ s.t.:
    \begin{enumerate}
        \item All size distributions of the items of $I'$ are either $\epsnoisy$ or terminals.
        \item There is an optimal adaptive decision tree for $I'$ such that at any decision point one of the options is to pick a terminal item.
    \end{enumerate}
\end{lem}

\begin{proof}
Starting from an optimal adaptive tree $T$ on $I$, observe that whenever an item realizes to its “large’’ size, any further adaptive branching is useless (any nonzero size thereafter forces overflow).  Thus each such branch of $T$ is a single‐chain of small‐size realizations followed by an eventual stop.  Replace each such chain by a single terminal item whose value equals the expected value of that chain.  Call the new tree $T'$ and the augmented instance $I'$, see \cref{fig:bernoulli_adaptive_tree} for illustration.

\begin{figure}[h]
\centering
\includegraphics[scale=0.5]{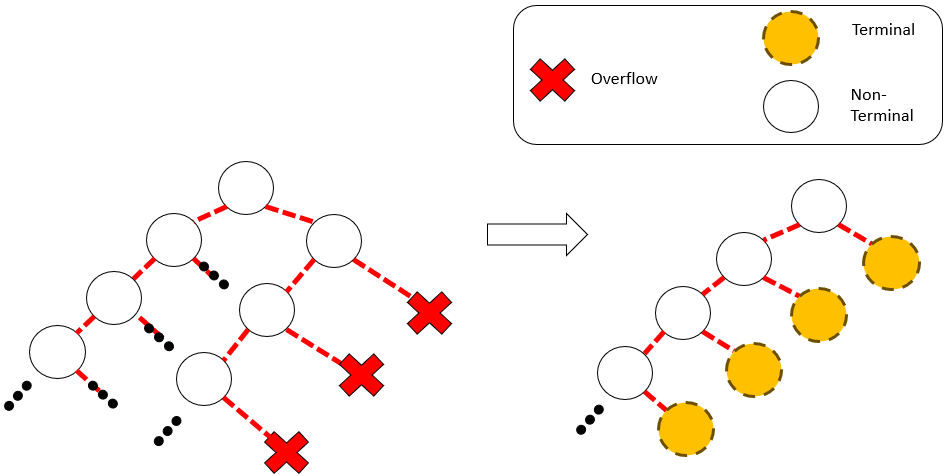}
\caption{An illustration of the reduction from $I$ (with $\epsnoisy$ items only) to $I'$ (where the optimal adaptive decision tree is to take some $\epsnoisy$ items and then a terminal item). On the left there is an illustration of $T$, an optimal adaptive decision tree for instance $I$. On the right there is an illustration of an optimal adaptive decision tree $T'$ for instance $I'$.}
\label{fig:bernoulli_adaptive_tree}
\end{figure}

The adaptive reward of $T'$ on $I'$ equals that of $T$ on $I$, so $\ADAPT(I)\le \ADAPT(I')$. Conversely, any non‐adaptive plan on $I'$ that picks a terminal can be simulated in $I$ by following the original chain, so $\ALG(I)\ge \ALG(I')$.  Hence
\[
  G(I)\;=\;\frac{\ADAPT(I)}{\ALG(I)}
  \;\le\;
  \frac{\ADAPT(I')}{\ALG(I')}
  \;=\;
  G(I').\qedhere
\]
\end{proof}

We are now ready to show the $2$ upper bound.

\begin{proof}[Proof of \cref{thm:nrsk-noisy-bernoulli-ub-2}]
By Lemma~\ref{lem:non-risky-noisy-bernoulli-instance-decision-tree}, it suffices to prove the bound for instances $I'$ whose optimal adaptive tree $T'$ has height~$k$ and at each level offers exactly one $\epsnoisy$ item and one terminal item.  Let
\[
  G_k \;=\; \sup\bigl\{\,G(J)\mid J\text{ has such a tree of depth }k\bigr\}.
\]
We prove by induction that $G_k\le2$ for all $k$, which implies the theorem.

\medskip\noindent\textbf{Base ($k=1$).}  A single‐node tree has no branching: $\ADAPT=\ALG$, so $G_1=1\le2$.

\medskip\noindent\textbf{Inductive Step.}  Fix $k>1$ and let $J$ be an instance with tree of height~$k$.  
For any $i \in [k]$ in the adaptive decision tree, let $(i,1)$, $(i,2)$ be the item corresponding to the non-terminal node and the item corresponding to the terminal node on height $i$ in the decision tree. Let $w_{i,1}$, $w_{i,2}$ denote the effective values of $(i,1)$ and $(i,2)$ respectively. These are the only two items on height $i$ (due to \cref{lem:non-risky-noisy-bernoulli-instance-decision-tree}), see \cref{fig:non-risky-g2-instance} for illustration.
\begin{figure}[h]
    \centering
    \includegraphics[width=0.7\linewidth]{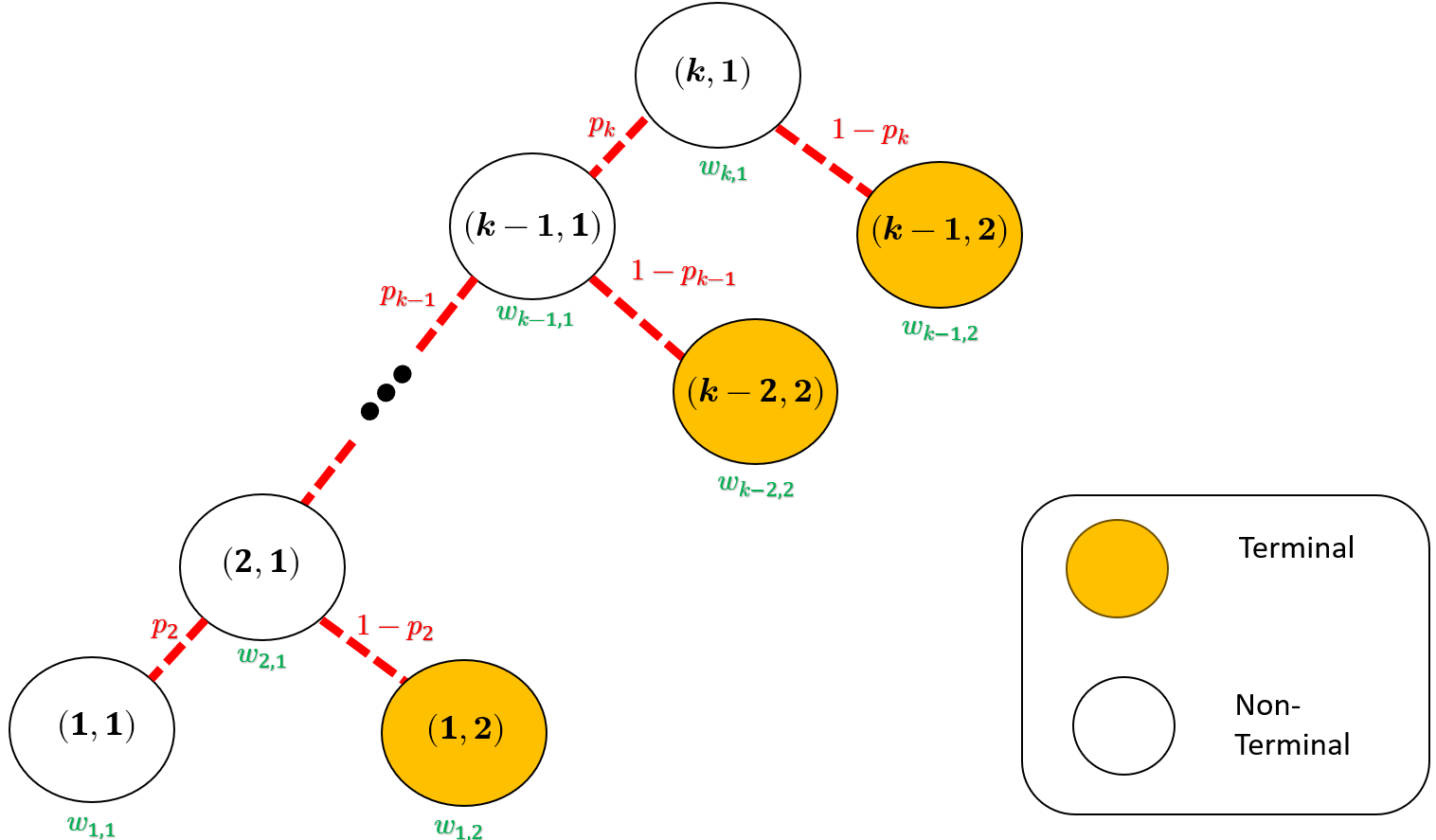}
    \caption{Illustration of $I'$ instance via the optimal adaptive decision tree, where $p_j$ is the probability item $(j,1)$ realizes to a small size.}
    \label{fig:non-risky-g2-instance}
\end{figure}
Let us assume that the overflow probability of the first item $\sad$ inserts is $0$, that is: $\Pr(S_{k,1} > 1) = 0$. Otherwise, the adaptivity gap may only decrease.
The value of $\ADAPT$ is the sum of the top-level item's value and the expected value from its sub-tree. Let $p_i$ be the probability that item $(i,1)$ realizes to a small size, and let $V_i$ be the value obtained by $\sad$ from the sub-tree rooted at item $(i,1)$. Then:
So:
\begin{equation}\label{eq:nrsk-noisy-ber-ub-2-adapt}
    \ADAPT = w_{k,1} + p_k \,V_{k-1} + (1 - p_k) \,w_{k-1,2}.
\end{equation}

On the other hand, consider the following three options for $\salg$: one option is to never take item $(k,1)$ and only take items from its sub-tree, yielding:
$\ALG \ge \frac{V_{k-1}}{G_{k-1}}$ (via the induction hypothesis).
$\salg$ may also take item $(k,1)$ and then continue to the sub-tree, yielding an expected value of at least $\ALG \ge w_{k,1} + p_k \frac{V_{k-1}}{G_{k-1}}$ or take item $(k,1)$ and then take item $(k-1,2)$ yielding value of $\ALG \ge w_{k,1} +  w_{k-1,2}$.

Consider a non-adaptive strategy $\salgt$ that picks the best of these three options. This may not be $\salg$ (which might achieve a higher expected value through a different order), but as we will show, bounding the gap between $\sad$ and $\salgt$ suffices for proving the upper bound of $2$. Let $\ALG'$ denote the expected value of $\salgt$. Then:
\begin{equation}\label{eq:alg_tag}
    \ALG \ge \ALG' = \max \crl*{\frac{V_{k-1}}{G_{k-1}}, w_{k,1} + p_k \frac{V_{k-1}}{G_{k-1}}, w_{k,1} + w_{k-1,2}}.
\end{equation}

\paragraph{Step 2: Equalize non-adaptive branches.} A similar argument to the one given in \cref{lem:nrsk-all_alg_options_yield_the_same_val} shows that the worst case is obtained where all of these three options yield the same expected value.
If choosing item $(k,1)$ and then $(k-1,2)$ yields different expected value than choosing item $(k,1)$ and then items lower in the tree, we may transform the instance by modifying $p_k$ (similar to the analysis of )
and get an instance with a larger adaptivity gap. Thus, we assume w.l.o.g. that the two options yields the same value, that is: $w_{k,1} + p_k \frac{V_{k-1}}{G_{k-1}} = w_{k,1} + w_{k-1,2}$. It follows that: \begin{equation}\label{eq:w_kminus1_2}
    w_{k-1,2} = p_k \frac{V_{k-1}}{G_{k-1}}.
\end{equation}

Let us bound the adaptivity gap $G_k$ for instance $I_k$:
\begin{align*}
    G_k & = \frac{\ADAPT}{\ALG} \le \frac{\ADAPT}{\ALG'} \stackrel{(\star)}{=}
    \frac{w_{k,1} + p_k V_{k-1} + (1 - p_k) p_k \frac{V_{k-1}}{G_{k-1}}}{\ALG'} \\
    & = \frac{w_{k,1} + p_k \frac{V_{k-1}}{G_{k-1}}}{\ALG'} + \frac{p_k V_{k-1}(1 - \frac{1}{G_{k-1}})}{\ALG'} + (1 - p_k) p_k \frac{\frac{V_{k-1}}{G_{k-1}}}{\ALG'}\\
    & \le 1 + p_k (G_{k-1} - 1) + (1-p_k) p_k = 1 + p_k G_{k-1} - p_k^2,
\end{align*}
where $(\star)$ is due to \cref{eq:nrsk-noisy-ber-ub-2-adapt}, \cref{eq:w_kminus1_2} and and the inequality is due to \cref{eq:alg_tag}.
\paragraph{Step 3: Optimize to find the bound on the gap.}
The function $f(p_k) = 1 + p_k G_{k-1} - p_k^2$ is maximized when $p_k = \frac{G_{k-1}}{2}$, 
yielding the following bound:
$G_k \le 1 + \frac{G_{k-1}^2}{4} \le 2$, where the last inequality is due to the induction hypothesis.
\end{proof}
